\documentclass{article}
\usepackage{graphicx} 
\usepackage{mathpazo}
\usepackage{mathtools,bbm,xspace}
\usepackage[pagebackref]{hyperref}
\usepackage[dvipsnames]{xcolor}
\usepackage{comment}
\hypersetup{
colorlinks=true,
urlcolor=blue,
linkcolor=RoyalBlue,
citecolor=OliveGreen,
linktocpage=true
}

\allowdisplaybreaks

\usepackage{amsmath}
\usepackage{amssymb}
\usepackage{amsfonts}
\usepackage{amsthm}
\usepackage{sectsty}
\usepackage{fullpage}
\usepackage{graphicx}
\usepackage{algorithm}
\usepackage{url}
\usepackage{tikz}
\usepackage{enumerate}

\usepackage[nameinlink,capitalize]{cleveref}
\crefname{lemma}{Lemma}{Lemmas}
\crefname{observation}{Observation}{Observations}
\crefname{fact}{Fact}{Facts}
\newcommand{\colorconstraints}{\text{Color Constraints}}
\crefname{colorconstraints}{(color constraints)}{Color Constraints}
\crefformat{colorconstraints}{#2\colorconstraints#3}
\crefname{indsetconstraints}{(indset constraints)}{IndSet Constraints}
\crefformat{indsetconstraints}{#2$\mathsf{IndSet\ Axioms}$#3}
\crefname{theorem}{Theorem}{Theorems}
\crefname{mtheorem}{Theorem}{Theorems}
\crefname{corollary}{Corollary}{Corollaries}
\crefname{claim}{Claim}{Claims}
\crefname{example}{Example}{Examples}
\crefname{algorithm}{Algorithm}{Algorithms}
\crefname{problem}{Problem}{Problems}
\crefname{definition}{Definition}{Definitions}
\crefname{notation}{Notation}{Notations}
\crefname{equation}{Eq.}{Eq.}
\crefname{table}{Table}{Tables}

\newcommand{\paren}[1]{(#1)}
\newcommand{\Paren}[1]{\left(#1\right)}

\newcommand{\Brac}[1]{\left[#1\right]}

\newcommand{\norm}[1]{\left\lVert#1\right\rVert}

\newcommand{\iprod}[1]{\langle#1\rangle}
\newcommand{\Iprod}[1]{\left\langle#1\right\rangle}

\newcommand{\op}{\mathrm{op}}

\DeclareMathOperator{\I}{\mathbb{I}}
\DeclareMathOperator*{\E}{\mathbb{E}}

\DeclareMathOperator{\R}{\mathbb{R}}
\DeclareMathOperator{\C}{\mathbb{C}}
\DeclareMathOperator{\N}{\mathbb{N}}

\DeclareMathOperator*{\Var}{\mathrm{Var}}

\DeclareMathOperator{\poly}{poly}

\DeclareMathOperator{\Tr}{\mathrm{Tr}}

\renewcommand{\[}{\begin{eqnarray*}}
\renewcommand{\]}{\end{eqnarray*}}

\newcommand{\set}[1]{ \left\{ #1 \right\}}

\newcommand{\one}{\mathbf{1}}

\newcommand{\eps}{\varepsilon}

\newtheorem{theorem}{Theorem}[section]
\newtheorem*{theorem*}{Theorem}
\newtheorem{lemma}[theorem]{Lemma}
\newtheorem*{lemma*}{Lemma}

\newtheorem*{proposition*}{Proposition}

\newtheorem{fact}[theorem]{Fact}

\theoremstyle{definition}

\newtheorem{definition}{Definition}[section]
\newtheorem{conjecture}{Conjecture}[section]
\theoremstyle{remark}
\newtheorem{remark}{Remark}[section]

\let\oldmarginpar\marginpar
\renewcommand\marginpar[1]{\-\oldmarginpar[\raggedleft\footnotesize #1]%
{\raggedright\footnotesize #1}}

\usepackage{titling}
\pretitle{\begin{center} \Large}
\posttitle{\end{center}}
	
\sectionfont{\large}
\subsubsectionfont{\normalsize \it}

\newcommand{\var}{\mathrm{Var}}

\newcommand{\dtv}{d_{\mathrm{TV}}}
\newcommand{\vol}{\mathrm{Vol}}
\DeclarePairedDelimiterX{\infdivx}[2]{(}{)}{%
  #1\,\delimsize\|\,#2%
}
\newcommand{\chisquare}{\chi^2\infdivx}
\newcommand{\chisquaretrunc}{\chi^2_D\infdivx}

\newcommand{\Planted}{\mathsf{P}}
\newcommand{\Null}{\mathsf{N}}
\newcommand{\Ber}{\mathrm{Ber}}
\newcommand{\Bin}{\mathrm{Bin}}

\newcommand{\Krawdaunt}{Krawtchouk\xspace}
\newcommand{\Kr}{\mathrm{Kr}}

\newcommand{\Ov}{\mathrm{Ov}}
\newcommand{\ShadowNull}{\nu}
\newcommand{\ShadowPlanted}{\pi}
\newcommand{\sgraph}{\vartheta}

\newcommand{\e}{\varepsilon}

\newcommand{\Abs}[1]{\left\lvert#1\right\rvert}

\DeclareMathOperator*{\Prob}{\mathbb{P}}
\DeclareMathOperator{\OU}{\mathrm{U}}
\DeclareMathOperator{\T}{\mathrm{T}}

\newcommand{\numm}{D}
\newcommand{\cN}{\calN} 
\newcommand{\berp}{\gamma} 
\newcommand{\berq}{\tilde{\gamma}} 

\usepackage[
letterpaper,
top=1in,
bottom=1in,
left=1in,
right=1in]{geometry}
\newcommand\MYcurrentlabel{xxx}
\newcommand{\MYstore}[2]{%
  \global\expandafter \def \csname MYMEMORY #1 \endcsname{#2}%
}
\newcommand{\MYload}[1]{%
  \csname MYMEMORY #1 \endcsname%
}
\newcommand{\MYnewlabel}[1]{%
  \renewcommand\MYcurrentlabel{#1}%
  \MYoldlabel{#1}%
}
\newcommand{\MYdummylabel}[1]{}
\newcommand{\torestate}[1]{%
  \let\MYoldlabel\label%
  \let\label\MYnewlabel%
  #1%
  \MYstore{\MYcurrentlabel}{#1}%
  \let\label\MYoldlabel%
}
\newcommand{\restatetheorem}[1]{%
  \let\MYoldlabel\label
  \let\label\MYdummylabel
  \begin{theorem*}[Restatement of \cref{#1}]
    \MYload{#1}
  \end{theorem*}
  \let\label\MYoldlabel
}
\newcommand{\restatelemma}[1]{%
  \let\MYoldlabel\label
  \let\label\MYdummylabel
  \begin{lemma*}[Restatement of \cref{#1}]
    \MYload{#1}
  \end{lemma*}
  \let\label\MYoldlabel
}
\newcommand{\restateprop}[1]{%
  \let\MYoldlabel\label
  \let\label\MYdummylabel
  \begin{proposition*}[Restatement of \cref{#1}]
    \MYload{#1}
  \end{proposition*}
  \let\label\MYoldlabel
}
\newcommand{\restatefact}[1]{%
  \let\MYoldlabel\label
  \let\label\MYdummylabel
  \begin{fact*}[Restatement of \cref{#1}]
    \MYload{#1}
  \end{fact*}
  \let\label\MYoldlabel
}
\newcommand{\restate}[1]{%
  \let\MYoldlabel\label
  \let\label\MYdummylabel
  \MYload{#1}
  \let\label\MYoldlabel
}

\renewcommand{\Pr}{\Prob}
\newcommand{\wt}[1]{\widetilde{#1}}
\newcommand{\wh}[1]{\widehat{#1}}
\newcommand{\ol}[1]{\overline{#1}}
\newcommand{\ul}[1]{\underline{#1}}

\newcommand{\pmo}{\{\pm1\}}

\newcommand{\TV}{d_{\mathrm{TV}}}

\newcommand{\dif}{\mathop{}\!\mathrm{d}}

\newcommand{\floor}[1]{\lfloor #1 \rfloor}

\newcommand{\polylog}{\operatorname{polylog}}

\newcommand{\Ind}{\boldsymbol{1}}
\newcommand{\supp}{\mathrm{supp}}
\newcommand{\Unif}{\operatorname{Unif}}

\newcommand{\grad}{\nabla}

\newcommand{\calA}{\mathcal{A}}
\newcommand{\calB}{\mathcal{B}}
\newcommand{\calC}{\mathcal{C}}
\newcommand{\calD}{\mathcal{D}}
\newcommand{\calE}{\mathcal{E}}

\newcommand{\calL}{\mathcal{L}}

\newcommand{\calN}{\mathcal{N}}
\newcommand{\calO}{\mathcal{O}}
\newcommand{\calP}{\mathcal{P}}

\newcommand{\calR}{\mathcal{R}}

\newcommand{\bg}{\boldsymbol{g}}

\newcommand{\bx}{{\boldsymbol{x}}}
\newcommand{\by}{\boldsymbol{y}}
\newcommand{\bz}{\boldsymbol{z}}

\newcommand{\bX}{\boldsymbol{X}}

\DeclarePairedDelimiter\abs{\lvert}{\rvert}
\DeclarePairedDelimiter\parens{\lparen}{\rparen}
\DeclarePairedDelimiter\braces{\lbrace}{\rbrace}
\DeclarePairedDelimiter\bracks{\lbrack}{\rbrack}
\DeclarePairedDelimiter\angles{\langle}{\rangle}

\newcommand\numberthis{\addtocounter{equation}{1}\tag{\theequation}}

\newcommand{\parhead}[1]{\medskip \noindent {\bfseries \ignorespaces {#1}}\hskip 0.9em}

\newcommand{\mper}{\,.}
\newcommand{\mcom}{\,,}

\renewcommand{\le}{\leqslant}
\renewcommand{\leq}{\leqslant}
\renewcommand{\ge}{\geqslant}
\renewcommand{\geq}{\geqslant}

\newcommand{\PlantedShadow}{\pi_{k,\eps}}
\newcommand{\FakePlantedShadow}{\pi'_{k,\eps}}
\newcommand{\NullShadow}{\nu_k}

\newcommand{\FakePlantedShadowHat}{\wh{\pi}'_{k,\eps}}
\newcommand{\NullShadowHat}{\wh{\nu}_k}

\newcommand{\erdos}{Erd\H{o}s\xspace}
\newcommand{\renyi}{R\'enyi\xspace}

\newif\ifjerryvictory
\jerryvictoryfalse

\ifjerryvictory

\newcommand{\rvg}{g}
\newcommand{\rvh}{h}

\newcommand{\rvx}{x}
\newcommand{\rvy}{y}
\newcommand{\rvz}{z}

\newcommand{\rvG}{G}

\newcommand{\rvM}{M}

\newcommand{\rvY}{Y}

\else

\newcommand{\rvg}{\boldsymbol g}
\newcommand{\rvh}{\boldsymbol h}

\newcommand{\rvx}{\boldsymbol x}
\newcommand{\rvy}{\boldsymbol y}
\newcommand{\rvz}{\boldsymbol z}

\newcommand{\rvG}{\boldsymbol G}

\newcommand{\rvM}{\boldsymbol M}

\newcommand{\rvY}{\boldsymbol Y}

\fi

\def\colorful{0}

\ifnum\colorful=1
\newcommand{\stefan}[1]{{\texttt{\color{Blue} Stefan: [{#1}]}}}
\newcommand{\tim}[1]{{\texttt{\color{Cerulean} Tim: [{#1}]}}}
\newcommand{\sid}[1]{{\texttt{\color{Fuchsia} Sidhanth: [{#1}]}}}
\newcommand{\pravesh}[1]{{\texttt{\color{Red} Pravesh: [{#1}]}}}
\newcommand{\jerry}[1]{{\texttt{\color{SeaGreen} Jerry: [{#1}]}}}

\else
\newcommand{\stefan}[1]{}
\newcommand{\tim}[1]{}
\newcommand{\sid}[1]{}
\newcommand{\pravesh}[1]{}
\newcommand{\jerry}[1]{}

\fi

\title{Rigorous Implications of the Low-Degree Heuristic} 
\author{
    Jun-Ting Hsieh\thanks{MIT. \texttt{juntingh@mit.edu}.} \and
    Daniel M. Kane\thanks{UCSD. \texttt{dakane@ucsd.edu}.} \and
    Pravesh K. Kothari\thanks{Princeton University. \texttt{kothari@cs.princeton.edu}.} \and
    Jerry Li\thanks{University of Washington, Seattle. \texttt{jerryzli@cs.washington.edu}.} \and
    Sidhanth Mohanty\thanks{Northwestern University. \texttt{sidhanth.mohanty@northwestern.edu}.} \and
    Stefan Tiegel\thanks{MIT. \texttt{stefan23@mit.edu}. Supported by an SNF Postdoc.Mobility Grant (P500-2\_235374)}
}
\date{\today}

\begin{document}
\maketitle

\begin{abstract}


Over the past decade, the \emph{low-degree heuristic} has been used to estimate the algorithmic thresholds for a wide range of average-case \emph{planted vs null} distinguishing problems.
Such results rely on the hypothesis that if the low-degree moments of the planted and null distributions are sufficiently close, then no efficient (noise-tolerant) algorithm can distinguish between them.
This hypothesis is appealing due to the simplicity of calculating the low-degree likelihood ratio (LDLR) --- a quantity that measures the similarity between low-degree moments.
However, despite sustained interest in the area, it remains unclear whether low-degree indistinguishability actually rules out any interesting class of algorithms.

In this work, we initiate the study and develop technical tools for translating LDLR upper bounds to rigorous lower bounds against concrete algorithms.
As a consequence, we prove:
for any permutation-invariant distribution $\Planted$,
\begin{enumerate}
    \item If $\Planted$ is over $\{0,1\}^n$ and is low-degree indistinguishable from $U = \Unif(\{0,1\}^n)$, then a noisy version of $\Planted$ is statistically indistinguishable from $U$.

    \item If $\Planted$ is over $\R^n$ and is low-degree indistinguishable from the standard Gaussian $\calN(0, 1)^n$, then no statistic based on symmetric polynomials of degree at most $O(\log n/\log \log n)$ can distinguish between a noisy version of $\Planted$ from $\calN(0, 1)^n$.

    \item If $\Planted$ is over $\R^{n\times n}$ and is low-degree indistinguishable from $\calN(0,1)^{n\times n}$, then no constant-sized subgraph statistic can distinguish between a noisy version of $\Planted$ and $\calN(0, 1)^{n\times n}$.
\end{enumerate}

\end{abstract}

\thispagestyle{empty}
\setcounter{page}{0}

\newpage

\tableofcontents

\thispagestyle{empty}
\setcounter{page}{0}

\newpage




\section{Introduction}
\label{sec:intro}

Over the past decade, there has been a surge of interest in understanding average-case signal detection and recovery problems. These problems arise as central challenges in multiple areas including algorithm design (e.g., detecting planted cliques in random graphs), statistics (e.g., learning the parameters of a Gaussian mixture), and cryptography (e.g., distinguishing the output of a pseudorandom generator from uniform bits). Such problems can usually be modeled as a pair of probability distributions: a \emph{planted} distribution with a signal of tunable ``strength'', and a \emph{null} distribution without any signal.
The main research question is whether the minimum signal strength required for \emph{efficient} signal detection/recovery exceeds the statistical minimum.
Decades of research suggest that information-computation gaps are ubiquitous not just in foundational problems, such as the examples above, but also in many new questions that arise in learning theory, statistics, and cryptography.
The modern area of average-case complexity has focused on building a principled theory to understand and predict such gaps.

How might we approach such an \emph{average-case} complexity theory?\footnote{The classical theory~\cite{levin1986average,bogdanov2006average,goldreich2011notes} of average-case complexity does not help understand algorithmic problems under input distributions such as those mentioned above.}
One strategy, in analogy to worst-case hardness reductions, is to build a web of reductions starting from some canonical average-case signal recovery problems.
This effort has made notable progress~\cite{berthet2013complexity,brennan2020reducibility,bresler2025computational}, but its success is still confined to a limited set of problems
due to the difficulty of designing reductions that transform a distribution on instances of one problem into the target distribution of another.\footnote{For example, despite significant efforts, we have so far failed to reduce refuting random $4$-SAT formulas with $n^{1+\Omega(1)}$ clauses to refuting random $3$-SAT formulas with $n^{1+\Omega(1)}$ clauses.}
Instead, most evidence of algorithmic vs.~statistical gaps in average-case complexity has come from a second approach --- showing lower bounds against restricted classes of algorithms (similar to circuit lower bounds in worst-case complexity).

\parhead{Hardness against restricted algorithms.}
A long sequence of works has focused on lower bounds against spectral methods \cite{MRZ16,BMVVX18,KWB19}, statistical query algorithms \cite{Kearns98,FGRVX17,DKS17,BBHLS21}, Markov chains \cite{Jerrum92,CMRW23,CSZ24,CMZ25}, approximate message passing \cite{GJ21,GJW24}, convex programs such as the Sum-of-Squares (SoS) semidefinite programming hierarchies \cite{BHKKMP16,HKPRSS17,MRX20,GJJPR20,KB21,JPRTX21,JPRX23,KPX24}.
In particular, in the context of SoS lower bounds, the discovery of the \emph{pseudocalibration} approach \cite{BHKKMP16}
suggested a strong relationship between SoS lower bounds and a certain \emph{low-degree likelihood ratio} (LDLR) between a planted and a null distribution that captures the closeness between low-degree moments of the two distributions.

\begin{definition}[Low-degree likelihood ratio and advantage] \label{def:LDLR-intro}
Let $\Planted$ (Planted) and $\Null$ (Null) be two probability distributions on $\{0,1\}^{\binom{n}{k}}$.
The likelihood ratio of $\Planted$ and $\Null$ is the function $L(x) = \frac{d\Planted}{d\Null}(x)$, and we use $L^{\leq D}$ to denote the projection of $L$ onto the span of polynomials of degree $\leq D$.

The degree-$D$ \emph{advantage} between $\Planted$ and $\Null$ is defined as $\sqrt{\Var_{\Null}[L^{\leq D}]}$, which is also equal to
\begin{equation}
\label{eq:ldlr}
    \max_{f\text{ deg-}D\text{ polynomial}} \frac{\E_{\Planted} f - \E_{\Null} f}{\sqrt{ \Var_{\Null}\, f }} \; .
\end{equation}
We use $\chisquaretrunc{\Planted}{\Null}$ to denote $\Var_{\Null}[L^{\leq D}]$, i.e., the square of the degree-$D$ advantage.
\end{definition}

In words, the degree-$D$ advantage measures how well degree-$D$ polynomials, normalized under $\Null$, can distinguish between $\Planted$ and $\Null$ in expectation.
The fact that $\chisquaretrunc{\Planted}{\Null}$ equals the square of \eqref{eq:ldlr} is a standard calculation; see Hopkins's thesis \cite{Hopkins18} or the survey of \cite{KWB19}.
Moreover, $\chisquaretrunc{\Planted}{\Null}$ can in fact be interpreted as the degree-$D$ truncation of the $\chi^2$-divergence between $\Planted$ and $\Null$, hence the notation.
Henceforth, we will use the term ``vanishing LDLR'' to mean $\chisquaretrunc{\Planted}{\Null} \leq o(1)$.

Starting with the works of \cite{HS17,HKPRSS17,Hopkins18}, it has been hypothesized that a small LDLR indicates computational hardness for hypothesis testing problems.
A precise conjecture is formalized in Hopkins's thesis \cite[Conjecture 2.2.4]{Hopkins18}, which roughly states that for symmetric distributions, small degree-$\polylog(n)$ advantage implies the absence of polynomial-time \emph{noise-tolerant} distinguishing algorithms.

\begin{conjecture}[The Low-Degree Hypothesis (informal) \cite{Hopkins18}] \label{conj:low-deg-conj}
    Fix $d \in \N$.
    Let $\Null$ be the uniform distribution on $\{0,1\}^{\binom {n}{d}}$ or the standard Gaussian $\calN(0,1)^{\binom{n}{d}}$ (viewed as symmetric $d$-tensors).
    Let $\Planted$ be a distribution on the same domain as $\Null$ which is $S_n$-symmetric\footnote{$S_n$-symmetry means that the distribution is invariant under the relabeling action of $S_n$.}
    and satisfies $\chisquaretrunc{\Planted}{\Null} \leq o(1)$ for $D = (\log n)^{1+\Omega(1)}$.
    Then, no polynomial-time computable test distinguishes $\T_{\eps} \Planted$ and $\Null$ with probability $1-o(1)$.
    Here, $\T_{\eps} \Planted$ is the noisy version of $\Planted$ (see \Cref{sec:overview-prelims} for precise definitions).
\end{conjecture}


We first make a few remarks about \Cref{conj:low-deg-conj}.
Both the $S_n$-symmetry assumption and adding noise to $\Planted$ are necessary.
Without noise, there are simple counter-examples where the efficient distinguishers are brittle algorithms such as Gaussian elimination or lattice basis reduction \cite{ZSWB22,DK22}.
The symmetry is also necessary due to counter-examples based on error-correcting codes~\cite{HW21}.

In a recent work~\cite{BHJK25} building on \cite{HW21}, Hopkins's conjecture (in the quasi-polynomial time regime) was disproved via a carefully constructed null vs.\ planted hypothesis testing problem.
Nevertheless, the LDLR method remains a popular method for deriving conclusions of hardness for average-case problems; 
see, e.g., \cite{KWB19,BHJK25,Wein25} and references therein for extensive lists of applications.
Moreover, the applications of the LDLR method extend beyond algorithm design and are popular even in statistics (to derive computational hardness of estimation) and cryptography~\cite{ABIKN23,BKR23,BJRZ24} (to obtain formal backing for hardness conjectures underlying security analyses).

The appeal of the LDLR framework stems from its \emph{simplicity} and \emph{predictive power}.
On the one hand, the degree-$D$ advantage is typically straightforward to compute for most hypothesis testing problems.
On the other hand, the hardness predictions from the LDLR method have been strikingly consistent with the conjectured hard regimes across a wide range of average-case problems,
including planted clique, planted coloring, planted dense subgraph problems, spiked Gaussian matrix and tensor models, planted vectors in random subspaces, sparse PCA, non-Gaussian component analysis, and more (see the survey of \cite{KWB19}).
In these problems, the hardness predicted by the LDLR method match the breakdown points of the best-known algorithms.


\parhead{What does a vanishing LDLR mean?}
Despite its remarkable ability to predict computational thresholds, the meaning of a vanishing LDLR remains poorly understood.
The work of Kunisky, Wein, and Bandeira~\cite{KWB19} made the first steps in this direction.
They showed that when $\Null$ is i.i.d.\ $\cN(0,1)$ or i.i.d.\ $\Unif(\pmo)$\footnote{The proof relies only on hypercontractivity, so any reasonable distribution would suffice.},
if there is a polynomial $p: \R^n \to \R$ of degree $\leq D/2$ that satisfies 
\begin{itemize}
    \item Large expectation under $\Planted$: $\E_{\rvx\sim\Planted}[p(\rvx)] \geq A$,
    \item Concentration under $\Null$: $\Pr_{\rvx\sim \Null}[p(\rvx)\geq B] \leq e^{-\Omega(D)}$,
\end{itemize}
for some $A > B > 0$,
then $\chisquaretrunc{\Planted}{\Null} \geq \Omega(1)$ (see \cite[Theorem 4.3]{KWB19} for the precise statement).
In contrapositive form, this result means that a vanishing $\chisquaretrunc{\Planted}{\Null}$ rules out the existence of low-degree polynomials that are concentrated under $\Null$ yet have large expectations under $\Planted$.

However, such guarantees on the expectations of low-degree polynomials do not formally rule out the use of low-degree polynomials as distinguishers.
For instance, even in the scenario where $\chisquaretrunc{\Planted}{\Null} = 0$ (i.e., all degree-$D$ moments match), we only know that the \emph{expectation} of any degree-$D$ polynomial is the same under the two distributions.
But such a matching of expectations does not rule out the possibility that some other statistic of a low-degree polynomial could distinguish between $\Null$ and $\Planted$.
Indeed, establishing rigorous hardness implications of the LDLR method was pointed out as an outstanding research direction in a recent workshop on the topic~\cite{AIM24}.

In this work, we initiate the first steps towards this research direction.
We focus on perhaps the most natural first question and ask:
\begin{center}
    {\it Does vanishing LDLR formally imply failure of distinguishers based on low-degree polynomials of inputs?}
\end{center}

Concretely, we would like to prove that, under the symmetry conditions of \Cref{conj:low-deg-conj} and assuming $\chisquaretrunc{\Planted}{\Null} \leq o(1)$ for $D = \omega(\log n)$, any low-degree polynomial $p$ satisfies $\TV(p(\Null), p(\T_{\eps} \Planted)) \leq o(1)$ for any constant $\eps > 0$.
This would rule out \emph{any} algorithm, whether efficient or not, that uses the evaluation of $p$ as a distinguisher.
In this work, we consider a stronger statement:
given a list of low-degree polynomials $\calP = (p_1, p_2,\dots,p_\ell)$, we would like the joint distributions of their evaluations under $\Null$ and $\T_\e \Planted$ to be statistically indistinguishable, i.e., $\TV(\calP(\Null), \calP(\T_\e \Planted)) \leq o(1)$.

In fact, in the Boolean vector case ($d=1$ in \Cref{conj:low-deg-conj}), we are able to prove the strongest possible guarantee that the two distributions themselves are statistically indistinguishable: $\TV(\Null, \T_\e \Planted) \leq o(1)$ (see \Cref{thm:binom_intro}).
This result exploits the strong restriction imposed by $S_n$-symmetry for vector-valued distributions.
In contrast, for matrix or tensor-valued problems ($d\geq 2$ in \Cref{conj:low-deg-conj}), we believe that such a strong statement is likely false.

\parhead{Bounded independence plus noise.}
From a technical perspective, our results have close connections to the pseudorandomness literature.
The condition $\chisquaretrunc{\Planted}{\Null} \leq o(1)$ can be viewed as an analogue of $\Planted$ being almost $D$-wise independent, and we are studying the noisy version of $\Planted$.
This resembles the ``bounded independence plus noise'' paradigm that has been studied in the context of pseudorandom generators (PRGs) \cite{AW85,HLV18,FK18,LV20}.
The key distinction is that these works focused on fooling specific function classes, i.e., showing that functions in a given class have similar expectations under the pseudorandom and truly random distributions.
In contrast, we aim to establish total variation closeness between certain functions under the two distributions, which is a stronger guarantee.
An important difference is that those works require ``pseudorandom noise'' to reduce the number of random bits, while in our setting, the added noise is i.i.d.\ by the definition of $\T_\e\Planted$.

Despite the differences mentioned above, the high-level strategy is similar:
given the Fourier expansion of a function, use $D$-wise independence to bound the low-degree terms, and use the noise to bound the high-degree ones.
In our case, to show TV closeness, we need \emph{point-wise} bounds on the differences between the probability density functions, which we obtain via the inverse Fourier transform and bounds on the characteristic functions $\wh{p}(\xi) = \E[e^{i \xi p}]$.
We follow the general template of bounding the low-frequency terms via moment matching and high-frequency terms via noise, though the technical details of analyzing the characteristic functions are substantially different from previous works on PRGs.
A detailed overview is provided in \Cref{sec:overview}.

We also mention a line of work on bounded independence alone (without noise) fooling several classes of functions, including $AC^0$ circuits, halfspaces, and low-degree polynomial threshold functions (in both the Boolean and Gaussian settings) \cite{Braverman08,DGJSV10,DKN10,Kane10,BIVW16,BDFIS22}.
Though we emphasize that, as in the LDLR hypothesis, the addition of noise is necessary in all our TV-closeness results --- for example, one can have a discrete distribution $\Planted$ that matches all low-degree moments of $\Null = \calN(0, \I_n)$, yet any function of it must have large TV distance.

\subsection{Our Results}

In this section, we will use $\T_\e \Planted$ (for the Boolean setting) and $\OU_\e \Planted$ (for the Gaussian setting) to denote an $\eps$-noisy version of the planted distribution $\Planted$; see \Cref{sec:overview-prelims} for precise definitions.
To begin with, we will show that $S_n$-symmetric distributions over $\set{\pm 1}^n$ with bounded LDLR (with noise) are in fact close in TV distance to the null distribution.
\begin{theorem}
\torestate{
\label{thm:binom_intro}
    Let $\berp,\eps \in (0,1)$ be absolute constants, and $\delta \geq 0$.
    Let $D \in \N$ such that $D \geq \frac{C_{\berp}}{\eps}\log n$, where $C_{\berp}$ is a constant depending only on $\berp$.
    Let $\Null = \Ber(\berp)^n$ and let $\Planted$ be an $S_n$-symmetric distribution over $\{\pm1\}^n$ such that $\chisquaretrunc{\Planted}{\Null} \leq \delta$. 
    Then,
    \[
        \TV(\Null, \T_\e \Planted) \leq O_{\berp}\Paren{\frac{\delta} {\e^2} + \frac{1}{\sqrt{n} \cdot \eps^2} }\,.
    \]
}
\end{theorem}
We remark that this result implies that no algorithm can distinguish the two distributions.
Observe that we prove a guarantee on the closeness of $\Null$ and the noisy version $\T_\e\Planted$ rather than $\Planted$ itself.
This is necessary to get a total variation guarantee, as simple counterexamples exist in the absence of noise.

Crucially, this result exploits that the $S_n$-symmetry reduces the question to a one-dimensional statement (see \cref{sec:overview_binom} for more details) as any $S_n$-invariant distinguisher must only depend on the Hamming weight of the input. 
This makes the analysis amenable to tools related to the ones used in the study of low-degree lower bounds.
In particular, we can upper bound the total variation distance by upper bounding the $\chi^2$-distance (after conditioning).

\parhead{Gaussian space.}
The Gaussian setting is significantly more challenging that its Boolean counterpart since $S_n$-symmetry is not sufficient to reduce to a one-dimensional problem. As a result, the Gaussian setting becomes significantly more challenging.
Nevertheless, we are able to derive the following two results.
The first concerns the vector-valued setting and arbitrary test statistics based on low-degree symmetric polynomials.
The second concerns the matrix-valued setting and a natural class of test statistics.

We begin with the vector-valued setting.
For $S_n$-symmetric planted distributions, it is natural to consider distinguishers based on low-degree symmetric polynomials.
However, we show that no statistic based on such polynomials can succeed.
In particular, we show that no such test can distinguish $\Null$ and $\OU_\e \Planted$ with non-vanishing probability if the LDLR is at most $o(k^{-O(k)} \log^{-1} n)$.

\begin{theorem}[See \cref{thm:main_technical} for full version]
\label{thm:vector-main}
    Let $\e \in (0,1)$ be a constant.
    Let $k, \numm \in \N$ be such that $\numm \geq \frac{C}{\eps^2} k^2\log n$ and $k\leq \frac{\eps \log n}{C\log \log n}$ where $C>0$ is a large universal constant.
    Let $\Null = \calN(0,\I_n)$ and $\Planted$ be a distribution over $\R^n$ such that $\chisquaretrunc{\Planted}{\Null} = \delta$.
    Let $\calP(\cdot)$ be the vector of symmetric polynomials of degree at most $k$.
    Then, it holds that
    \[
        \dtv\Paren{\calP(\Null),\calP(\OU_\e \Planted)} \leq \delta^{\Omega(1)} \cdot k^{k/2} + n^{-1/2} \,.
    \]
\end{theorem}

Note that since we only consider symmetric polynomials in \Cref{thm:vector-main}, we don't actually need the assumption that $\Planted$ is $S_n$-symmetric.

\begin{remark} \label{rem:local-CLT}
    Although establishing this is not the main focus of our work, we believe that our proof could, with some more work (e.g., a more careful analysis of \Cref{lem:moment_bounds}), yield \emph{local central limit theorems} (cf.\ \cite{Petrov12}).
    In particular, with an appropriate normalization of $\calP$, we expect to get $\TV(\calP(\Null), \calN(0, \I_k)) \leq o(1)$ and $\TV(\calP(\OU_\eps \Planted), \calN(0, \I_k)) \leq o(1)$, which would imply \Cref{thm:vector-main}.

    On the other hand, this can be viewed as an inherent barrier to extending \Cref{thm:vector-main} to $k \geq \Omega(\log n)$.
    We do not expect symmetric polynomials of degree $\Omega(\log n)$, such as the power-sum polynomial $\sum_{i=1}^n x_i^k$, to satisfy local central limit theorems, even under $\Null$.
\end{remark}

\parhead{Subgraph counts in matrix models.}
\stefan{To be checked again if we really need the extra $\log \log n$ factor, it might just have been for safety}
In the matrix-valued setting, we work in the setting where $\Null$ is the distribution of $n\times n$ symmetric matrices with independent unit variance Gaussian entries (up to symmetry), and $\Planted$ is some distribution for which $\chi^2_D\parens*{\Planted\|\Null}\le\delta$ where $D \ge \log n \cdot \log\log n$.

We study the distinguishing power of \emph{signed subgraph counts}.
Concretely, for a connected graph $\sgraph$ on $v$ vertices and $e$ edges, and a matrix $M$, we define its \emph{signed $\sgraph$-count} as
\begin{align*}
    \chi_{\sgraph}(M) \coloneqq \sum_{\substack{H \text{ copy of }\sgraph \\ \text{ in }K_n}} \prod_{ab\text{ edge in } H} M_{ab}\mper
\end{align*}
Our main result on the (lack of) distinguishing power of signed subgraph counts is that no statistic based on the signed subgraph count of a constant-sized subgraph can distinguish $\Null$ from a noisy version of $\Planted$.
Formally:
\begin{theorem}[See \cref{thm:main-subgraph} for the full version]
\label{thm:sub-graph-intro}
    Let $\ShadowNull$ refer to the law of $\chi_{\sgraph}(\rvM)$ for $\rvM\sim\Null$, and let $\ShadowPlanted$ refer to its law for $\rvM\sim\OU_{\eps}\Planted$ where $\eps > 0$ is an absolute constant.
    Suppose $\chisquaretrunc{\Planted}{\Null} \le \delta$ where $D \ge \log n\cdot \log\log n$.
    There is a sufficiently small constant $\alpha > 0$ for which:
    \[
        \dtv\parens*{\ShadowNull,\ShadowPlanted} \le O\parens*{ \delta^{\alpha\eps^2} + n^{-\alpha\eps^2} }\mcom
    \]
    where the $O\parens*{\cdot}$ hides constant factors depending on $\eps$ and $\vartheta$.
\end{theorem}

\begin{remark}
    In the special case where $\Planted = \Null$, our result establishes a \emph{local central limit theorem} for subgraph count polynomials in Gaussian matrices.
    Considerable efforts have been expended in establishing such results for \erdos--\renyi random graphs \cite{GK16,Ber18,SS22} via a ``bare-hands'' analysis of the Fourier spectrum of these distributions.
    Our analysis of the Fourier spectrum in the Gaussian case makes connections to second-order Poincar{\'e} inequalities of Chatterjee \cite{Cha09}, and operator norm bounds for graph matrices studied in the context of sum-of-squares lower bounds \cite{AMP16,BHKKMP16}.
\end{remark}

\begin{remark}
    To preclude state-of-the-art algorithms such as spectral methods, or distinguishers based on the largest entry of a matrix, we would need a version of \Cref{thm:sub-graph-intro} that holds for $O(\log n)$-sized subgraphs.
    While there are still technical barriers to surmount (cf.\ \Cref{rem:local-CLT}) to achieve this, our result is the first step in this direction.
\end{remark}

As an example, consider the \emph{Wigner sparse PCA} problem, where the planted distribution is described by $\lambda uu^\top + W$, where $\lambda > 0$, $W \sim \Null$ and $u$ is an $\ell$-sparse unit vector whose non-zero entries are, say, $\pm 1/\sqrt{\ell}$, and $\ell \ll \sqrt{n}$.
Note that in this case choosing $\e = \Theta(1)$, the noise operator only changes the signal strength $\lambda$ by a small constant.
We are interested in how large $\lambda$ needs to be so that we can distinguish $\Null$ and $\Planted$.
Statistically, this is possible if and only if $\lambda = \wt{\Omega}(\sqrt{\ell})$.
On the other hand, all known efficient algorithms require $\lambda = \wt{\Omega}(\ell)$ \cite{johnstone2009consistency}.

For concreteness, consider the setting where $\ell = n^{0.4}$ and $\lambda = n^{0.3}$.
Then, statistically, $\Null$ and $\Planted$ are distinguishable since $\lambda \gg n^{0.2}$.
Yet, known algorithms do not succeed since $\lambda \ll \wt{\Omega}(n^{0.4})$.
Indeed, it was verified in \cite[Proof of Theorem 2.24 (a)]{MR4760356-Ding24} that the LDLR up to degree $o(n)$ is bounded by $o(1)$ in this setting.
Thus, our theorem implies that any test based on a single, constant-sized subgraph statistic must necessarily have vanishing success probability.




\section{Technical Overview}
\label{sec:overview}
\subsection{Background and Notation}  \label{sec:overview-prelims}
We use lower case letters to denote scalars and vectors and upper case letter to denote matrices
In the paper, we will let denote a standard normal random variables by $\rvg$ (for scalars and vectors) or $\rvG$ (for matrices).

We first define the noise models that we use.
\begin{definition}[Noise Models]
\label{def:noise-models}
\mbox{}
    \begin{enumerate}
        \item \textbf{Vector-valued Binomial setting:} Let $\Planted$ be a distribution over $\{\pm 1\}^n$ and $\rvy \sim \Planted$. Let $\rvz \sim \Null$ be independent of $\rvy$. We denote by $\T_\e \Planted$ the distribution of $\wt{\rvy}$, where for each $i$ we independently set $\wt{\rvy}_i = \rvy_i$ with probability $1-\e $ and $\wt{\rvy}_i = \rvz_i$ with probability $\e$.
        \item \textbf{Vector-valued Gaussian setting:} Let $\Planted$ be a distribution over $\R^n$ and $\rvy \sim \Planted$. Let $\rvg \sim \calN(0,I_n)$ be independent of $\rvy$. We denote by $\OU_\e \Planted$ the distribution of $\sqrt{1-\e}\rvy + \sqrt{\e} \rvg$.
        \item \textbf{Matrix-valued Gaussian setting:} Let $\Planted$ be a distribution over symmetric matrices in $\R^{n \times n}$ and $\rvM \sim \Planted$. Let $\rvG \in \R^{n \times n}$ be symmetric, independent of $\rvM$ and have i.i.d. $\calN(0,1)$ entries (up to symmetry). We denote by $\OU_\e \Planted$ the distribution of $\sqrt{1-\e}\rvM + \sqrt{\e} \rvG$.
    \end{enumerate}
\end{definition}



\parhead{Fourier transform.}
For any sufficiently nice function $f: \R^d \to \R$, we let $\wh{f} (\xi)$ denote the standard Fourier transform of $f$:
\[
\wh{f} (\xi) = \tfrac 1 {\sqrt{2\pi}}\int_{\R^d} e^{-i \iprod{\xi, x}} f(x) \, dx \; .
\]
For any sufficiently nice random variable $\rvx$ over $\R^d$, let $\pi_{\rvx}$ denote the probability density function of $\rvx$, and let $\widehat{\pi}_{\rvx}$ denote the Fourier transform of $\rvx$:
\[
\wh{\pi}_{\rvx} (\xi) = \E_{\rvx} \left[ \exp (i \iprod{\xi, \rvx}) \right] \; .
\]
Equivalently, $\wh{\pi}_{\rvx}(\xi)$ is the Fourier transform of $\pi_{\rvx}$ at $\xi$.\footnote{Note that this coincides with the characteristic function of $\rvx$.}

\subsection{Lower Bounds in Binomial Space}
\label{sec:overview_binom}

For $S_n$-symmetric distributions over $\{\pm1\}^n$, the analysis for \Cref{thm:binom_intro} reduces to studying a single statistic --- the sum of the entries.
For this overview, with a slight abuse of notation, we will use $\Null$ and $\Planted$ to denote the distribution of the sum of the entries.
The null distribution is exactly the binomial distribution $\Bin(n,\gamma)$, and the associated orthonormal polynomials are the (shifted) \Krawdaunt polynomials $(\Kr_i)_{i\leq n}$, which depend on $\gamma$ and $n$ (see \Cref{sec:orthogonal-polynomials}).
Note that by the central limit theorem, $\Bin(n,\gamma)$ converges to a shifted Gaussian, so we expect $\Kr_i$ to behave roughly like the Hermite polynomials.

The natural strategy is to upper bound $\chisquare{\T_\eps \Planted}{\Null}$, which implies an upper bound on the TV distance.
Given the orthonormal basis, we can write
\begin{align*}
    \chisquare{\T_\eps \Planted}{\Null} = \sum_{i\geq 1} \parens*{\E_{\T_\eps \Planted}[ \Kr_i]}^2
    = \sum_{i \geq 1} (1-\eps)^i \parens*{\E_{\Planted}[\Kr_i]}^2 \mper
\end{align*}
With the $(1-\eps)^i$ decay, it is natural to split the summation into the low-degree and high-degree parts.
For the low-degree part, bounded LDLR, i.e., $\chisquaretrunc{\Planted}{\Null} \leq \delta$, implies that 
$\sum_{i=1}^D (1-\eps)^i \E_{\Planted}[\Kr_i]^2 \leq \delta$.
For the high-degree part, we have $\sum_{i > D} (1-\eps)^i (\E_{\Planted}[\Kr_i])^2 \leq e^{-\eps D} \sum_{i \geq 1} (\E_{\Planted}[\Kr_i])^2$, and note that $\sum_{i \geq 1} (\E_{\Planted}[\Kr_i])^2$ is exactly $\chisquare{\Planted}{\Null}$.

Unfortunately, a priori, we do not have any upper bound on $\chisquare{\Planted}{\Null}$, as $\Planted$ may have some small probability events that blows up the $\chi^2$-divergence.
To bypass this barrier, we observe that just matching low-degree moments suffices to obtain well-behaved tail bounds for $\Planted$.
Thus, we consider the truncated distribution $\ol{\Planted}$, which is $\Planted$ conditioned on being within $[-\tau, \tau]$ for $\tau = \Theta(\sqrt{\log n})$ --- the typical region of a standard Gaussian.
This event has probability $1-\frac{1}{\poly(n)}$ even under $\Planted$,
thus $\TV(\Planted, \ol{\Planted}) \leq \frac{1}{\poly(n)}$.

With the truncation, it is easy to see that $\chisquare{\ol{\Planted}}{\Null} \leq \poly(n)$.
Thus, setting $D = O(\frac{1}{\eps}\log n)$ suffices to make $(1-\eps)^D \poly(n) \leq o(1)$.
Some care is required to show that $\E_{\ol{\Planted}}[\Kr_i]$ still remains small for $i \leq D$ after truncation.
We prove this in \Cref{lem:basis-truncated} using LDLR and properties of \Krawdaunt polynomials.

\subsection{Lower Bounds in Gaussian Space}
We proceed to discuss our results about low-degree indistinguishability in Gaussian space.
As discussed before, the situation in this space is substantially more involved than in the binomial space.
Intuitively, this is because $S_n$-symmetry is a much weaker condition over the real numbers than it is for discrete distributions.
To prove our bounds in this space, we depart significantly from ``standard'' techniques in the literature on low-degree likelihood ratio lower bounds.
Instead of relating the TV distance to a quantity like $\chi^2$-divergence such as in the previous section, we will directly bound the $L_1$-distance between the probability density functions of polynomials of Gaussians and polynomials of the noisy planted distribution.
We note that our strategy resembles the blueprint of local limit theorems based on Fourier inversion \cite{sirazdinov1962mean,Petrov12,GK16,Ber18,SS22}.

As a running example, in this section we will sketch how we do this for a single hypothetical degree $k$ polynomial $p: \R^n \to \R$.
Note that for our actual results, we will actually have to consider arbitrary vector-valued symmetric polynomials, and structured polynomials over matrix inputs, i.e. $p: \R^{n \times n} \to \R$. 
We will explain how our actual techniques differ from this simplified exposition as the differences arise.
Let $\rvy \sim \Planted$ be a random variable so that $\chisquaretrunc{\Planted}{\Null} \leq \delta$ for some $D$ sufficiently large, and let $\eps$ be a small constant.
Let $\rvz = \OU_\e (\rvy)$ be a noisy version of $\rvy$.
Our goal will be to show that the total variation distance between $p(\rvg)$ and $p(\rvz)$ is vanishingly small.
To simplify the exposition somewhat, let us assume that $\rvy$ not only approximately matches $D$ moments with a Gaussian, but in fact exactly matches all moments of degree at most $D$.
Additionally, let us normalize $p$ so that $\Var[p(\rvg)] = 1$, and let $T = D / k$.

\parhead{Step 1: From TV distance to Fourier bounds.}
Our first move is rather standard.
Let $R$ to be chosen later.
Up to a factor of 1/2, we can write the total variation distance between $p(\rvg)$ and $p(\rvz)$ as
\begin{equation}
\label{eq:prelim-split}
\int_{\R} \left| \pi_{p(\rvg)} (x) - \pi_{p(\rvz)}(x) \right| d x = \int_{|x| \leq R} \left| \pi_{p(\rvg)} (x) - \pi_{p(\rvz)}(x) \right| d x +  \int_{|x| > R} \left| \pi_{p(\rvg)} (x) - \pi_{p(\rvz)}(x) \right| d x \; .
\end{equation}
By standard concentration bounds, we know that because $p$ has unit variance and is low degree, if we choose $R = \omega (1)$, then the second term is vanishing since both $p(\rvg)$ and $p(\rvz)$ place vanishing probability mass outside of a constant-length interval.

To bound the first term, our strategy will be to show a \emph{pointwise} bound on $|\pi_{p(\rvg)} (x) - \pi_{p(\rvz)}(x)|$.
To show such a bound, our main tool will be to analyze the Fourier transforms of $p(\rvg)$ and $p(\rvz)$.
In particular, for any fixed $y$, by the inverse Fourier transform, we have that
\begin{align*}
    2\pi \cdot |\pi_{p(\rvg)} (y) - \pi_{p(\rvz)}(y)| &= \left| \int_{\R} e^{i \xi y} \left( \widehat{p(\rvg)} (\xi) - \widehat{p(\rvz)} (\xi) \right) d \xi \right| \\
    &\leq \int_{\R} \left| \widehat{p(\rvg)} (\xi) - \widehat{p(\rvz)} (\xi) \right| d \xi \\
    &= \underbrace{\int_{|\xi| \leq R'} \left| \widehat{p(\rvg)} (\xi) - \widehat{p(\rvz)} (\xi) \right| d \xi}_{\mathrm{(A)}} + \underbrace{\int_{|\xi| > R'} \left| \widehat{p(\rvg)} (\xi) - \widehat{p(\rvz)} (\xi) \right| d \xi}_{\mathrm{(B)}} \; , \numberthis \label{eq:prelim-fourier-split}
\end{align*}
for some appropriate choice of parameter $R'$.
It now remains to bound both of the terms in~\eqref{eq:prelim-fourier-split}.

\parhead{Step 2: Frequency matching.}
We first show how to bound (A) in~\eqref{eq:prelim-fourier-split}.
This is the regime in which we use moment matching.
Recall that $T = D / k$.
By Taylor series expansion, we obtain
\begin{align*}
\widehat{p(\rvg)} (\xi) &= \E [\exp (i \xi p(\rvg))]\\
&= \sum_{\ell = 0}^{T - 1}  \E \left[ \frac{(i \xi p(\rvg))^\ell}{\ell!}\right] + O \left( \E \left[ \frac{(\xi p(\rvg))^{T}}{T!} \right] \right) \; ,
\end{align*}
and similarly
\[
\widehat{p(\rvz)} (\xi) = \sum_{\ell = 0}^{T - 1}  \E \left[ \frac{(i \xi p(\rvz))^\ell}{\ell!}\right] + O \left( \E \left[ \frac{(\xi p(\rvz))^{T}}{T!} \right] \right) \; .
\]
Since $G$ and $Z$ match $D$ moments it holds that $\E\left[ p(\rvz)^\ell \right] = \E \left[ p(\rvg)^\ell \right]$ for all $\ell \leq D/k$.
Hence,
\[
\left| \widehat{p(\rvg)} (\xi) - \widehat{p(\rvz)} (\xi)\right| \leq O \left( \E \left[ \frac{(\xi p(\rvg))^{T}}{T!} \right] \right) \; .
\]
The crucial thing we will require of $p$ is that the degree-$T$ moment of $p(\rvg)$ must grow strictly slower than $O(T^T)$; that is, $p(\rvg)$ must have smaller degree-$T$ moments than a sub-exponential random variable has.
For instance, if the $T$-th moment of $p(\rvg)$ behaves like that of a sub-Gaussian random variable, i.e. $\E [p(\rvg)^T] = O(T^{T/2})$, then we would have that
\begin{equation}
    \label{eq:prelims-frequencyt-matching}
    \left| \widehat{p(\rvg)} (\xi) - \widehat{p(\rvz)} (\xi)\right| \leq O \left( \frac{|\xi|^2}{T} \right)^{T/2} \; ,
\end{equation}
and so we will obtain that 
\[
\mathrm{(A)} \leq O(T) \cdot \exp(-\Omega (T)) \ll 1 \; ,
\]
as long as $|R'|^2 \leq O(T)$.

For the structured polynomials with matrix-valued inputs we consider (known as subgraph count polynomials), such moment bounds are known in the literature when $T$ is a large constant (see, e.g., \cite{SS22}).
Unfortunately, we will require $T$ to be slowly growing with $n$ (i.e. $O(\log n / \log \log n)$), so these bounds do not suffice.
Instead, we derive these bounds by hand via involved combinatorial arguments.

The situation is a bit more complicated for low-degree symmetric polynomial with vector-valued inputs.
Not only do we have to deal with the fact that we wish to control the entire distribution of \emph{all} symmetric polynomials simultaneously, but 
also, it is patently untrue that even a single low-degree symmetric polynomial of a Gaussian will always have sub-Gaussian moments.
For instance, consider the symmetric degree-$2$ polynomial 
\[
p(x) = \left( \frac{1}{\sqrt{n}} \sum_{j = 1}^n x_i \right)^2 \; . 
\]
Then, $p(\rvg)$ is distributed as a chi-squared random variable, which clearly does not have sub-Gaussian tails.

However, we observe that while not all symmetric polynomials may have sub-Gaussian tails, there exists a nice family of symmetric polynomials that do in fact exhibit such tail behavior, and which form a sufficient statistic for all low-degree symmetric polynomials.
Specifically, we will exploit that any symmetric degree-$k$ polynomial can be written as a deterministic function of the polynomials 
\begin{equation}
\label{eq:phi}
\phi_\ell = \frac{1}{\sqrt{n}} \sum_{j = 1}^n h_\ell (x_j) \; ,
\end{equation}
for $\ell = 1, \ldots, k$, where $h_\ell$ is the $\ell$-th (probabilist's) Hermite polynomial.
Note that this is a somewhat non-standard family of symmetric polynomials, however, it turns out to be a very natural choice for our purposes.
Crucially, $\phi_\ell (\rvg)$ is a sum of $n$ independent random variables, and so one can readily verify that all such polynomials will have sub-Gaussian moment bounds for all values of $T$ that we require, so long as $k \leq O(\log n / \log \log n)$.

Our strategy will be to directly show that the distribution of the vector-valued polynomial $F_k: \R^n \to \R^k$ where $(F_k)_i = (\phi_\ell)_{1 \le \ell\le k}$ under $\rvg$ and under $\rvz$ are close in TV distance, by showing that their multivariate Fourier transforms are pointwise close.
By showing that the distributions of $F_k (\rvg)$ and $F_k(\rvz)$ are close, by the data processing inequality, this implies that the joint distribution of \emph{all} symmetric low-degree polynomials under $\rvg$ and $\rvz$ are close to each other. 
This in turn immediately implies that there is no possible distinguisher that uses any amount of symmetric low-degree information that can distinguish between $\rvg$ and $\rvz$.
We emphasize that this choice of sufficient statistics turns out to be very important for our proof: not only does it yield the necessary moment bounds here, but the independence structure in these polynomials also turns out to be important for obtaining the right Fourier decay bounds, as we explain now.

\parhead{Step 3: Fourier decay.}
Finally, we turn to the task of bounding term (B) in~\eqref{eq:prelim-fourier-split}.
We first write
\[
\mathrm{(B)} \leq \int_{|\xi| > T} \left| \widehat{p(\rvg)} (\xi) \right| d \xi + \int_{|\xi| > T} \left| \widehat{p(\rvz)} (\xi) \right| d \xi \; ,
\]
and we will bound each term separately, although the second term will be almost strictly harder to control than the first one, so here we will primarily discuss how to control the second term. 

That is to say, the last remaining conceptual challenge is to demonstrate Fourier decay for $p(\rvz)$.
Here, we must somehow use the fact that $\rvz = \OU_\e [\rvy]$, i.e. that $\rvz$ has Gaussian noise built into it.
But this by itself will be insufficient, as we also need control over the behavior of $\rvy$.
It turns out that the most convenient way to do this, in both settings we consider, will be to not directly work with $p(\rvz)$, but with a random variable which is close in TV distance to $p(\rvz)$, which satisfies certain regularity conditions.
So, in our actual proofs, we run through Steps 1-3 as outlined above not with $p(\rvz)$, but with regularized versions of $p(\rvz)$, and then eventually relate the bounds back to $p(\rvz)$ with a triangle inequality.

The exact notion of regularity we will enforce will be dependent on the specific setting.
In the symmetric setting, recall we are working with the basis of polynomials $\phi_\ell$ as defined in \Cref{eq:phi}, and that this basis has the nice property that the polynomials are a sum of $n$ independent random variables, when the input is Gaussian.
Because of this independence structure, it turns out that the magnitude of the Fourier coefficients depends primarily on the variance of the polynomials $q_j(x) = h_\ell (\sqrt{1 - \e^2} {\rvy}_j + \e x)$ for ``typical'' $\rvy$ under Gaussian input $x$.
While it is not true that this variance is well-controlled for typical $\rvy$ for all of the coordinates $j = 1, \ldots, d$, we can show that with high probability over the choice of $\rvy$, there exist many coordinates for which this variance is well-controlled, and so if we condition on this event happening, we are able to achieve the desired bound on the Fourier coefficients of the original polynomial.
Summarizing, we show that there is a high probability event for $\rvy$ so that if we condition on this event, $\widehat{p(\OU_\e [\rvy])} (\xi)$ is bounded.
So here, our regularized version of $p(\OU_\e [\rvy])$ is just $p(\OU_\e [\tilde{\rvy}])$, where $\tilde{\rvy}$ drawn from the same distribution as $\rvy$, conditioned on this high probability event.

\parhead{Subgraph count polynomials.}
For the subgraph count polynomials, unfortunately, we lose this independence structure, and so we cannot use this type of argument to bound the Fourier coefficients.
However, for any single, sufficiently nice function $f$, there exist strong local limit theorems due to Chatterjee (see \Cref{lem:shivam}) which state that $f(\OU_\e [\rvY])$ is close in TV distance to $f(\rvY) + \sigma \rvG'$, where $\rvG'$ is Gaussian, for some appropriate scale factor $\sigma$.
Since computing Fourier transforms of Gaussians is easy, we can take this to be the regularized version of $f(\OU_\e [\rvY])$ in this setting.
The main technical challenge here is demonstrating that Chatterjee's bound is sufficiently tight to obtain a non-trivial local limit theorem (LLT) for these subgraph count polynomials, which entails bounding the $\ell_2$-norm of the gradient, the spectral norm of the Hessian, and give a lower bound on the variance of these polynomials.
Doing so requires fairly involved graph matrix moments that we control using careful combinatorial bounds, as well as techniques developed originally for proving sum-of-squares lower bounds by \cite{AMP16}.

\parhead{Remark: Relation to local limit theorems.}
As a brief aside, these sorts of Fourier decay bounds are very closely related to LLTs, and it should be no surprise that LLTs arise in our proofs.
A local limit theorem (LLT) is a strengthening of the classic central limit theorems (CLT) from probability theory.
CLTs, or more generally, invariance principles, state that for random variables $\rvx \in \R^n$ with ``nice'' i.i.d. coordinates, and for ``nice'' functions $f:\R^n \to \R$ we have that $f(\rvx)$ is close in distribution to an appropriately scaled Gaussian, i.e., its cumulative distribution function is everywhere close to the cdf of the Gaussian.
In contrast, a local limit theorem states that the \emph{pdf} of $f(\rvx)$ is everywhere close to that of a Gaussian---in particular, this implies that all of the local behavior of $f(\rvx)$ is also similar to that of a Gaussian, hence the name.
Because the higher order Fourier coefficients of a Gaussian are exponentially vanishing, to prove an LLT, one typically needs to first prove that the higher order Fourier coefficients of $f(\rvx)$ are also exponentially small.
While there are some limited classes of polynomials for which such LLTs are known, unfortunately they do not suffice for our purposes, 
and which is why we prove these Fourier decay bounds by hand for the polynomials that arise in our proofs.
In fact, our bounds can be interpreted as a type of new local limit theorem for random variables of the form $p(\OU_\e [\rvy])$.

\section{Technical Preliminaries}
\label{sec:prelims}

\subsection{Statistical Distances}

\begin{definition} \label{def:chi-squared}
    The $\chi^2$-divergence, denoted $\chisquare{P}{Q}$, is defined as
    \begin{align*}
        \chisquare{P}{Q} = \E_{\rvx\sim Q} \bracks*{\parens*{ \frac{P(\rvx)}{Q(\rvx)}-1 }^2 }
        = \E_{\rvx\sim Q} \bracks*{\parens*{ \frac{P(\rvx)}{Q(\rvx)}}^2 } - 1 \mper
    \end{align*}
    In other words, $\chisquare{P}{Q}$ is the variance of the likelihood ratio $P/Q$ under $Q$.
\end{definition}

We note that the $\chi^2$-divergence is not a distance, since it is not symmetric.
However, from Jensen's inequality, we get the following simple fact that relates the $\chi^2$-divergence to the TV distance.
\begin{fact} \label{fact:tv-chi-squared}
    $\TV(P,Q) \leq \frac{1}{2}\sqrt{\chisquare{P}{Q}}$.
\end{fact}

When $Q$ is a product distribution over $\Omega^n$, it has an associated orthonormal basis $\{\psi_\alpha\}_{\alpha\in \N^n}$ for functions $f: \Omega^n \to \R$.
Then, we have the following expression for the $\chi^2$-divergence.

\begin{lemma} \label{lem:chi-squared-fourier}
    Let $Q$ be a product distribution over $\Omega^n$, and let $\{\psi_\alpha\}_{\alpha\in \N^n}$ be an orthonormal basis with respect to $Q$.
    Then, $\chisquare{P}{Q} = \sum_{\alpha \neq 0} (\E_{\rvx\sim P}[\psi_\alpha(\rvx)])^2$.
\end{lemma}
\begin{proof}
    Let $L(x) = P(x) / Q(x)$.
    We can express $L$ in terms of the basis: $L(x) = \sum_{\alpha} \wh{L}(\alpha) \psi_{\alpha}(x)$, where $\wh{L}(\alpha) = \E_{\rvx\sim Q}[L(\rvx) \psi_{\alpha}(\rvx)] = \sum_{x\in \Omega^n} Q(x) \cdot \frac{P(x)}{Q(x)}\psi_{\alpha}(x) = \E_{\rvx\sim P}[\psi_{\alpha}(\rvx)]$.
    We also have that $\wh{L}(0) = \E_Q[L] = 1$.

    From \Cref{def:chi-squared}, we can write $\chisquare{P}{Q}$ as $\E_Q[L^2]-1$.
    By Parseval's theorem and the above expressions for $\wh{L}(\alpha)$, we get $\chisquare{P}{Q} = \sum_{\alpha} \wh{L}(\alpha)^2 - 1 = \sum_{\alpha\neq 0} (\E_P[\psi_\alpha])^2$.
\end{proof}

The following is a simple but crucial statement we will need to bound the $\chi^2$-divergence between two distributions.

\begin{lemma} \label{lem:chi-squared-naive-bound}
    Suppose $\supp(P) \subseteq \supp(Q)$, then $\chisquare{P}{Q} \leq \frac{1}{\min_{x\in \supp(P)} Q(x)}$.
\end{lemma}
\begin{proof}
    From \Cref{def:chi-squared}, $\chisquare{P}{Q} = \E_Q[(P/Q)^2] - 1 = \sum_{x\in \supp(P)} \frac{P(x)^2}{Q(x)} -1 \leq \frac{1}{\min_{x\in \supp(P)} Q(x)}$.
\end{proof}

We also need the well-known data processing inequality:

\begin{fact}[Data processing inequality] \label{fact:data-processing}
    If $\kappa$ is an arbitrary channel that transforms $P$ and $Q$ into $P_\kappa$ and $Q_\kappa$ respectively, then $\TV(P_\kappa, Q_\kappa) \leq \TV(P, Q)$.
\end{fact}

\subsection{Bernoulli and Binomial Distributions}
\label{sec:binomial-prelims}

The following is a standard fact that follows from Taylor expansion.

\begin{fact} \label{fact:KL-bernoulli}
    Denote the KL divergence between $\Ber(\berp)$ and $\Ber(\berq)$ as $D(\berp\|\berq) \coloneqq \berp\log \frac{\berp}{\berq} + (1-\berp) \log \frac{1-\berp}{1-\berq}$.
    Then for $\berp \leq 1/2$ and $\eps \leq \berp/2$ it holds that $D(\berp+\eps \|\berp) \leq \frac{\eps^2}{2\berp(1-\berp)} + O(|\eps|^3 /\berp^2)$.
\end{fact}

The following is a folklore result.
We provide a proof in \Cref{sec:binomial-appendix} for completeness.

\begin{lemma}
\torestate{\label{lem:binomial-prob}
    Let $\berp \leq 1/2$.
    For any $y \in \R$ such that $|y| \leq o(\sqrt{\berp n})$ and $n\berp + y\sqrt{n\berp(1-\berp)} \in \{0,1,\dots, n\}$,
    \begin{align*}
        \Pr\bracks*{\Bin(n,\berp) = n\berp + y\sqrt{n\berp(1-\berp)}} \geq \frac{1}{\sqrt{2n}} \exp \parens*{-\frac{y^2}{2} - O(|y|^3/\sqrt{\berp n})} \mper
    \end{align*}
}
\end{lemma}

We give a quantitative version of the well-known fact that the moments of the centered binomial random variable match those of a standard gaussian, and give a proof in \Cref{sec:binomial-appendix}.
\begin{lemma}[Binomial moments]
\torestate{\label{lem:binomial-moments}
    Let $\berp \leq 1/2$, and let $\rvx \sim \Ber(\berp)^n$ and $\rvy = \frac{1}{2\sqrt{n\berp(1-\berp)}}\sum_{i=1}^n (\rvx_i - (2\berp-1))$.
    Then, for any $k \in \N$,
    $\E[\rvy^{2k}] \leq (2k-1)!! \cdot (1+O(\frac{k^3}{\berp n}))$.
}
\end{lemma}

For the following fact is standard; see, e.g., \cite[Theorem 10.21]{o2014analysis}.
\begin{lemma}[Hypercontractivity] \label{lem:hypercontractivity}
    Fix a constant $\berp \in (0,1)$.
    There is a constant $C_{\berp} > 0$ such that for any $n$-variate polynomial $p$ of degree $d$, $\E_{\rvx\sim\Ber(\berp)^n}[p(\rvx)^4] \leq C_{\berp}^{d} \cdot \E_{\rvx\sim\Ber(\berp)^n}[p(\rvx)^2]^2$.
\end{lemma}

\subsection{Orthogonal Polynomials}
\label{sec:orthogonal-polynomials}

In this paper, we will crucially use several families of orthogonal polynomials, depending on the specific application, which we review here.

\paragraph{Probabilist's Hermite polynomials}
For a non-negative integer $k$, let $h_k: \R \to \R$ denote the $k$-th normalized (probabilist's) Hermite polynomial~\cite{Sze75}:
\[
h_k (x) \coloneqq \frac{1}{\sqrt{k!}} \mathrm{He}_k(x) = \frac{1}{\sqrt{k!}} (-1)^k e^{x^2/2} \frac{d^k}{d x^k} e^{-x^2 / 2} \; .
\]
Here, $\mathrm{He}$ is the usual notation for the (unnormalized) Hermite polynomials.
A standard property of the Hermite polynomials is that for all $k \geq 0$, the polynomials $\{h_\ell\}_{\ell = 0}^k$ form a complete orthonormal basis for all degree-$k$ polynomials with respect to the Gaussian inner product $\iprod{\cdot, \cdot}$, defined by $\iprod{p, q} = \E_{\rvg \sim \calN(0,1)} [p(\rvg) q(\rvg)]$.
In particular,
\begin{align*}
    \E_{\rvg \sim \calN(0,1)} [h_k(\rvg) h_{\ell}(\rvg)] = \one(k=\ell) \; .
\end{align*}
Moreover, we have the recurrence $\mathrm{He}_{\ell}'(x) = \ell \cdot \mathrm{He}_{\ell-1}(x)$ for all $x$, and thus $h_{\ell}'(x) = \sqrt{\ell} \cdot h_{\ell-1}(x)$ for all $\ell \geq 1$.

We also use the following standard estimates, which will prove useful:
\begin{lemma}
\label{lem:deriv-bound}
    Let $\rvg \sim \calN(0,1)$.
    For any polynomial $p$ of degree $k$, we have
    $\Var[p(\rvg)] \leq \E[p'(\rvg)^2] \leq k\cdot \Var[p(\rvg)]$.
\end{lemma}
\begin{proof}
    Write $p(\rvg) = \sum_{\ell = 0}^k \alpha_\ell h_\ell (\rvg)$, so that $\Var[p(\rvg)] = \sum_{\ell =1}^k \alpha_\ell^2$.
    Then, since $h'_\ell = \sqrt{\ell} \cdot h_{\ell - 1}$ for all $\ell > 0$, and $h'_0 = 0$, we have that $\E[p'(\rvg)^2] = \sum_{\ell = 1}^{k} \ell \alpha_\ell^2$, from which the claim trivially follows.
\end{proof}

We recall hypercontractivity of Gaussian polynomials, see, e.g., \cite[Theorem 9.21 and note of Chapter 9]{o2014analysis} and \cite{nelson1973free}.
\begin{fact}[Gaussian hypercontractivity]
    \label{fact:gauss_hypercontractivity}
    Let $k \in \N$, $q \geq 2$, and let $p$ be an arbitrary polynomial of degree at $k$.
    Then $\E[|p(\rvg)|^q] \leq (q-1)^{qk/2} \cdot \E[p(\rvg)^2]^{q/2}$.
\end{fact}
\noindent
We also require the following standard anti-concentration bound for polynomials in Gaussian space:
\begin{fact}[Carbery-Wright inequality~\cite{CW01}]
\label{fact:carbery-wright}
    Let $\rvg \sim \calN (0, \I_n)$, and let $p: \R^n \to \R$ be a degree-$k$ polynomial.
    Then there is some universal constant $B > 0$ so that
    \[
    \Pr \left[ |p(\rvg)| \leq \eps \var[p(\rvg)]^{1/2} \right] \leq B k \eps^{1/k} \; .
    \]
\end{fact}

\parhead{Symmetric Hermite polynomials} We will also crucially use the following symmetric multivariate extension of the Hermite polynomials:
\begin{definition}
Let $k,n$ be positive integers. Define the vector valued polynomial $F_k:\R^n\rightarrow\R^k$ by
$$
(F_k)_i := \frac{1}{\sqrt{n}} \sum_{j=1}^n h_i(x_j) \; .
$$
\end{definition}
\noindent
By the fundamental theorem of symmetric polynomials any degree-$k$ symmetric polynomial can be written as a function of the elementary symmetric polynomials up to degree $k$.
It is not hard to see that the elementary symmetric polynomials up to degree $k$ can be written as a function of the vector $F_k$ above.
As a consequence, any degree-$k$ symmetric polynomial can be written as some function of $F_k$.

Note that this is a slightly non-standard basis for symmetric polynomials in Gaussian space, as one would typically take Hermite variants of the elementary symmetric polynomials instead.
However, this basis turns out to be much more convenient for us.

\parhead{\Krawdaunt polynomials}
Any symmetric polynomial over variables $x_1, x_2,\dots, x_n$ can be uniquely expressed as a polynomial of the elementary symmetric polynomials.
In particular, suppose $\calD$ is the distribution where $x_i = 1$ with probability $\berp \in (0,1)$ and $-1$ with probability $1-\berp$ independently for each $i \in [n]$, then the following polynomials are the associated orthogonal basis:
\begin{align*}
    s_k(x) = \sum_{S\subseteq[n]: |S|=k} \chi_S(x) \mcom
\end{align*}
where $\chi_S(x) = \prod_{i\in S} \frac{x_i - (2\berp-1)}{2\sqrt{\berp(1-\berp)}}$.
(For ease of notation, we will omit the dependence on $\berp$).
Indeed, since $\E[x_i] = 2\berp-1$ and $\Var[x_i] = 4\berp(1-\berp)$, we have that $\E_{x\sim \calD}[s_k(x)] = 0$ for $k\geq 1$ and $\E_{x\sim \calD}[s_k(x) s_{\ell}(x)] = \binom{n}{k}\delta_{k\ell}$ for all $k, \ell \leq n$.

Moreover, $s_k(x)$ can be expressed as a degree-$k$ univariate polynomial in $\sum_{i=1}^n x_i$.
These are known as the (shifted and normalized) \Krawdaunt polynomials~\cite{Kra29}.

\begin{definition}[Shifted \Krawdaunt polynomials] \label{def:krawdaunt}
    Fix $\berp\in (0,1)$ and $n\in \N$.
    We define $\Kr_k$ to be the unique degree-$k$ univariate polynomial, depending on $n$ and $\berp$, such that
    \begin{align*}
        \Kr_k(y;\ n, \berp) = \binom{n}{k}^{-1/2} s_k(x)
    \end{align*}
    for all $x \in \pmo^n$ and $y = \frac{1}{\sqrt{n}} s_1(x) = \frac{1}{\sqrt{n}} \sum_{i=1}^n \frac{x_i - (2\berp-1)}{2\sqrt{\berp(1-\berp)}}$.
    We will omit the dependence on $n,\berp$ when they are clear from context.
\end{definition}

One can easily see that $\{\Kr_k\}_{k\leq n}$ are orthonormal with respect to the distribution of $y$.
Intuitively, since $y \xrightarrow{d} \calN(0,1)$ as $n\to \infty$ by the central limit theorem, $\Kr_k(y)$ converges to the normalized (probabilist's) Hermite polynomials $h_k$, which are orthonormal with respect to the standard Gaussian distribution~\cite{Sze75}.

The following is a key fact we will need:
\begin{lemma}
\torestate{
\label{lem:krawtchouk-bound}
    Fix a constant $\berp \in (0,1)$.
    For any $k\leq n/2$ and $|y| \leq o(n^{1/6})$, we have $|\Kr_k(y)| \leq O(k^{1/4}) \cdot e^{y^2/4}$.
}
\end{lemma}

\Cref{lem:krawtchouk-bound} is analogous to the bound $|h_k(y)| \leq O(e^{y^2/4})$ for Hermite polynomials~\cite{BC90}.
For the case when $\berp=1/2$, a.k.a.\ the binary \Krawdaunt polynomials, precise bounds (without the extra $k^{1/4}$ factor) were given in \cite{Kra01} and recently \cite{Pol19,DILV24}.
However, we could not find a reference for \Cref{lem:krawtchouk-bound} for an arbitrary $\berp\in (0,1)$.
Therefore, we provide a proof in \Cref{sec:krawtchouk-bound-proof} that closely follows the proof given in \cite{Pol19,DILV24}.

\section{Binomial Space: Full Indistinguishability}
\label{sec:bernoulli}

In this section, we prove \cref{thm:binom_intro} restated below.
\restatetheorem{thm:binom_intro}

Throughout this section, we will assume that $\berp$ is a fixed constant and omit the dependence on $\berp$ for ease of notation.

Given $\bx \in \{\pm1\}^n$ drawn from $\Null$ (Null) or $\Planted$ (Planted), we denote the \emph{centered} and \emph{normalized} sum of coordinates as
\begin{equation*}
    \by = \frac{1}{2\sqrt{n\berp(1-\berp)}} \sum_{i=1}^n (\bx_i - (2\berp-1)) \mper
    \numberthis \label{eq:centered-sum}
\end{equation*}
Let $\nu, \pi$ be the distributions of $\by$ when $\bx$ is sampled from $\Null$ and $\Planted$ respectively.
Note that $\E_\nu[\by] = 0$ and $\E_\nu[\by^2] = 1$ since under the null, $\bx_i \sim \Ber(\berp)$ has mean $2\berp-1$ and variance $4\berp(1-\berp)$.

With a slight abuse of notation, let $\T_\e \pi$ denote the distribution of $\by$ when $\bx \sim \T_\e \Planted$.
Since $\Null$ and $\Planted$ (and also $\T_\e \Planted$) are symmetric distributions, we have the following lemma which states that we can focus on the distributions of $\by$:
\begin{lemma} \label{lem:pass-to-y}
    For symmetric distributions $\Null$ and $\Planted$, $\TV(\Null, \T_\e \Planted) = \TV(\nu, \T_\e \pi)$.
\end{lemma}
\begin{proof}
    Note that any symmetric function $f: \{\pm1\}^n \to \R$ can be uniquely expressed as a function of the sum of coordinates, i.e., $f(x) = g(\sum_{i=1}^n x_i)$ for all $x\in \{\pm1\}^n$ and some function $g: \R \to \R$.
    For symmetric distributions, the probability mass functions (and thus likelihood ratio) are symmetric functions, which means that they only depend on $y$.
    The statement then follows immediately since $\T_\e \Planted(x) / \Null(x) = \T_\e \pi(y) / \nu(y)$.
\end{proof}

Therefore, to prove \Cref{thm:binom_intro}, it is equivalent to upper bound $\TV(\nu, \T_\e \pi)$.

We will also use the following fact that the induced noise operator for $\T_\e \pi$ acts as expected with respect to the Krawtchouk basis.
\begin{fact}
\label{fact:noisy_krawtchouk}
Let $\pi$ be induced by any distribution over $\{\pm 1\}^n$.
For any $\ell \geq 1$ it holds that
$$
\E_{\wt{\rvy} \sim \T_\e \pi} \Kr_\ell(\wt{\rvy}) = (1-\e)^\ell \E_{\rvy \sim \pi} \Kr_\ell(\rvy) \mper
$$
\end{fact}
\begin{proof}
    Let $\rvx \sim \Planted$ and $\rvz \sim \Null$ be independent.
    Let $\wt{\rvx} = \T_\e [\rvx]$, i.e., such that for each $i$ independently $\wt{\rvx}_i = \rvx_i$ with probability $1-\e$ and $\wt{\rvx}_i = \rvz_i$ with probability $\e$.
    Then 
    \begin{align*}
        \Kr_\ell(\wt{\rvy}) = \binom{n}{\ell}^{-1/2} \sum_{S \subseteq [n] \colon \abs{S} = \ell}  \chi_S(\wt{\rvx}) \quad\quad\text{and}\quad\quad\Kr_\ell(\rvy) = \binom{n}{\ell}^{-1/2} \sum_{S \subseteq [n] \colon \abs{S} = \ell}  \chi_S(\rvx) \,.
    \end{align*}
    Thus, by linearity it is enough to show that for every $S \subseteq [n]$ of size $\ell$ it holds that $\E_{\wt{\rvx} \sim \T_\e \Planted} \chi_S(\wt{\rvx}) = (1-\e)^\ell \E_{\rvx \sim \Planted} \chi_S(\rvx)$.
    Since conditioned on $\rvx$ the entries of $\wt{\rvx}$ are independent, it holds that
    \begin{align*}
        \E_{\wt{\rvx} \sim \T_\e \Planted} \chi_S(\wt{\rvx}) = \E_{\rvx \sim\Planted}\E_{\wt{\rvx} \sim \T_\e \Planted }\bracks*{ \chi_S(\wt{\rvx}) \mid \rvx} = \E_{\rvx \sim\Planted}\E_{\wt{\rvx} \sim \T_\e \Planted }\bracks*{ \prod_{j \in S} \frac{\wt{\rvx}_i - (2\gamma-1)}{2\sqrt{\gamma(1-\gamma)}} \mid \rvx} = \E_{\rvx \sim\Planted} \prod_{j \in S} \E_{\wt{\rvx} \sim \T_\e \Planted }\bracks*{ \frac{\wt{\rvx}_i - (2\gamma-1)}{2\sqrt{\gamma(1-\gamma)}} \mid \rvx} 
    \end{align*}
    Now, conditioned on $\rvx$, either $\wt{\rvx}_j = \rvx_j$ or $\wt{\rvx}_j = \rvz_j$.
    In the second case, the inner expectation is 0 by construction.
    Thus, 
    $$
        \E_{\wt{\rvx} \sim \T_\e \Planted} \chi_S(\wt{\rvx}) = \E_{\rvx \sim\Planted} \prod_{j \in S} \parens*{1-\e} \cdot \parens*{ \frac{\rvx_i - (2\gamma-1)}{2\sqrt{\gamma(1-\gamma)}}} = (1-\e)^\ell \E_{\rvx \sim \Planted} \chi_S(\rvx) \mper
    $$

\end{proof}

\subsection{Truncated Planted Distribution}

In this section, we analyze properties of the truncated planted distribution $\ol{\pi}$, i.e., $\pi$ conditioned on $|\by| \leq \tau$ for some threshold $\tau$.
We start with a simple bound which follows directly from lower bounds on the probability mass function of the binomial distribution.

\begin{lemma} \label{lem:chi-squared-of-truncated}
    Let $\ol{\pi}$ be the distribution of $\pi$ conditioned on $|\by| \leq \tau = 2\sqrt{\log n}$.
    Then, $\chisquare{\ol{\pi}}{\nu}\leq O(n^3)$.
\end{lemma}
\begin{proof}
    Note that $\supp(\ol{\pi}) \subseteq \supp(\nu)$.
    Thus, by \Cref{lem:chi-squared-naive-bound}, it suffices to lower bound $\nu(y)$ for $|y|\leq \tau$.
    Since $\tau \leq o(\sqrt{\berp n})$, by \Cref{lem:binomial-prob} we have
    \begin{align*}
        \nu(y) = \Pr_{\bx\sim \Ber(\berp)^n} \bracks*{\sum_{i=1}^n \bx_i = (2\berp-1)n + 2\sqrt{n\berp(1-\berp)} \cdot y}
        \geq \frac{1}{\sqrt{2n}} \exp\parens*{-\frac{y^2}{2} - O\parens*{|y|^3/\sqrt{\berp n}}} \mper
    \end{align*}
    For $\tau = 2\sqrt{\log n}$, the above is at least $\Omega(n^{-5/2})$.
    Then, by \Cref{lem:chi-squared-naive-bound}, $\chisquare{\ol{\pi}}{\nu} \leq \frac{1}{\min_{|y|\leq \tau} \nu(y)} \leq O(n^3)$.
\end{proof}

Since $\nu$ is the (centered and normalized) binomial distribution, we know that its moments and tail bounds are close to those of a standard Gaussian (\Cref{lem:binomial-moments}).
We show the following general result: assuming low-degree indistinguishability, any statistic with Gaussian-like moments under $\Null$ also has Gaussian-like tail bounds under $\Planted$.
\begin{lemma}[Gaussian-like tail bounds] \label{lem:tail-bound}
    Suppose $\Null$ and $\Planted$ satisfy $\chisquaretrunc{\Planted}{\Null} \leq \delta$, and suppose $p$ is a multivariate polynomial of degree $d$ such that $\E_{\bx\sim \Null}[p(\bx)^{2k}] \leq B(2k-1)!!$ for some $B \geq 1$ for all $k \leq \frac{D}{2d}$.
    Then, for any $t \leq \sqrt{D/d}$, we have $\Pr_{\bx\sim \Planted}[|p(\bx)| \geq t] \leq O(B) \cdot (\delta + 2^{-t^2/4}) e^{-t^2/4}$.
\end{lemma}

\begin{proof}
    Let $k = \floor{t^2/4} \leq D/4d$.
    If $t\leq 2$, the statement is trivial, so we may assume $k \geq 1$.
    Let $\eta \coloneqq \Pr_{\bx\sim\Planted}[|p(\bx)| \geq t]$.
    Consider the degree-$2dk$ polynomial $p^{2k}$.
    First, we have
    \begin{align*}
        \E_{\bx\sim\Planted}[p(\bx)^{2k}] \geq \eta\cdot \E_{\bx\sim\Planted} \bracks*{p(\bx)^{2k} \mid |p(\bx)|\geq t}
        \geq \eta t^{2k} \mper
    \end{align*}
    On the other hand, $\chisquaretrunc{\Planted}{\Null} \leq \delta$ and $4dk \leq D$ imply that
    \begin{equation*}
        \E_{\bx\sim\Planted}[p(\bx)^{2k}] \leq \E_{\bx\sim\Null}[p(\bx)^{2k}] + \delta \sqrt{\Var_{\bx\sim\Null}[p(\bx)^{2k}]}
        \leq \E_{\bx\sim\Null}[p(\bx)^{2k}] + \delta \sqrt{\E_{\bx\sim\Null}[p(\bx)^{4k}]} \mper
    \end{equation*}
    By assumption, we have $\E_{\Null}[p^{2k}] \leq B(2k-1)!!$, which is at most $B\cdot \sqrt{2} (\frac{2k}{e})^k$ by Stirling's approximation.
    Similarly, we have $\E_{\Null}[p^{4k}] \leq B \cdot \sqrt{2} (\frac{4k}{e})^{2k}$ since $2k \leq D/2d$.
    Thus,
    \begin{align*}
        \eta \leq t^{-2k} \E_{\bx\sim\Planted}[p(\bx)^{2k}]
        \leq t^{-2k} \cdot 2B\parens*{\parens*{\frac{2k}{e}}^k + \delta \parens*{\frac{4k}{e}}^k }
        \leq 2B(\delta + 2^{-k}) \parens*{\frac{4k}{et^2}}^k \mper
    \end{align*}
    As $k = \floor{t^2/4} \geq t^2/4-1$, the above is $O(B) \cdot (\delta + 2^{-t^2/4})  e^{-t^2/4}$.
\end{proof}

Next, we study the behavior of the expectation of functions under truncation.
We prove the following general bound assuming that the distribution has exponential tail bounds and that the function is bounded.



\begin{lemma}[Bounded functions under truncation]
\label{lem:function-truncated}
    Let $\tau > 0$ and $\alpha > 0$.
    Let $\pi$ be a distribution such that $\Pr_{\by\sim\pi}[|\by| \geq \tau] \leq \alpha e^{-\tau^2/4}$.
    Then, for any $f: \R \to \R$,
    \begin{align*}
        \abs{\E_{\by\sim\pi}[f(\by) \cdot \Ind(|\by|\leq \tau)]}
        \leq \abs{\E_{\by\sim\pi}[f(\by)]} + \sqrt{\alpha e^{-\tau^2/4} \E_{\by\sim\pi}[f(\by)^2]} \mper
        \numberthis \label{eq:small-i-bound}
    \end{align*}
    Moreover, suppose $T \geq \tau \geq 2$ and $\beta > 0$ such that $\Pr_{\by\sim\pi}[|\by| \geq t] \leq \alpha e^{-t^2/4}$ for all $t\in [\tau, T]$ and $|f(y)| \leq \beta e^{y^2/4}$ for all $|y| \in [\tau, T]$, then
    \begin{align*}
        \abs{\E_{\by\sim\pi}[f(\by) \cdot \Ind(|\by|\leq \tau)]}
        \leq \abs{\E_{\by\sim\pi}[f(\by)]} + \frac{1}{4} \alpha \beta T^2 + \sqrt{\alpha e^{-T^2/4} \E_{\by\sim\pi}[f(\by)^2]} \mper
        \numberthis \label{eq:large-i-bound}
    \end{align*}
\end{lemma}
\begin{proof}
    Let $\calA$ be the event that $|\by| > T$, and let $\calB$ be the event that $\tau < \by \leq T$.
    Then,
    \begin{align*}
        \E_{\pi}[f]
        = \E_{\pi} \bracks*{ f \cdot \Ind(|\by| \leq \tau)} + \E_{\pi} \bracks*{ f \cdot \Ind(\calB) } + \E_{\pi} \bracks*{ f \cdot \Ind(\calA)} \mper
        \numberthis \label{eq:Ki-3-terms}
    \end{align*}
    To establish \Cref{eq:small-i-bound}, we set $T = \tau$ and thus the second term is $0$.
    For the third term, by Cauchy-Schwarz, we have $\E_{\pi}[f \cdot \Ind(\calA)] \leq \sqrt{\E_{\pi}[f^2] \cdot \Pr_{\pi}[\calA]} \leq \sqrt{\alpha e^{-T^2/4} \E_{\pi}[f^2]}$.
    
    To establish \Cref{eq:large-i-bound}, we consider $T > \tau$ and bound the second term of \Cref{eq:Ki-3-terms} using the assumption that $|f(y)| \leq \beta e^{y^2/4}$ for $|y| \in [\tau, T]$:
    \begin{align*}
        \abs{ \E_{\pi} \bracks*{ f \cdot \Ind(\calB) }}
        \leq \E_{\pi} \bracks*{ |f| \cdot \Ind(\calB) }
        \leq \beta \cdot \E_{\pi} \bracks*{e^{\by^2/4} \cdot \Ind(\calB)} \mper
    \end{align*}
    Denote $\bz \coloneqq e^{\by^2/4} \cdot \Ind(\calB)$, which is a non-negative random variable, and $\bz = 0$ unless $\tau < \by \leq T$.
    Thus, for any $\theta > 0$,
    \begin{equation*}
        \Pr_{\pi}[\bz \geq \theta] \leq
        \begin{cases}
            \Pr_{\pi}[|\by| \geq \tau] \mcom & 0 < \theta < e^{\tau^2/4} \mper \\
            \Pr_{\pi}[|\by| \geq 2\sqrt{\log \theta}] \mcom  & e^{\tau^2/4} \leq \theta \leq e^{T^2/4} \mper \\
            0 \mcom & \theta > e^{T^2/4} \mper
        \end{cases}
    \end{equation*}
    For the first two cases, we have $\Pr[|\by| \geq 2\sqrt{\log \theta}] \leq \alpha e^{-\log \theta} = \alpha \theta^{-1}$ by assumption (since $2\sqrt{\log\theta} \leq T$).
    Thus,
    \begin{align*}
        \E_{\pi}[\bz] = \int_0^{\infty} \Pr_{\pi}[\bz \geq \theta] \ d\theta
        &\leq \alpha e^{-\tau^2/4} \cdot e^{\tau^2/4} +  \int_{e^{\tau^2/4}}^{e^{T^2/4}} \alpha \theta^{-1} \ d\theta 
        \leq \alpha + \frac{\alpha}{4} (T^2 - \tau^2)
        \leq \frac{1}{4} \alpha T^2 \mper
    \end{align*}
    Then, $|\E_{\pi}[f \cdot \Ind(\calB)]| \leq \frac{1}{4}\alpha \beta T^2$.
    Plugging into \Cref{eq:Ki-3-terms} completes the proof.
\end{proof}

As a corollary, we get the following.

\begin{lemma} 
\label{lem:basis-truncated}
    Let $\delta \in [n^{-1/2}, 1]$,
    $\tau = 2\sqrt{\log n}$, and $2\tau^2 \leq D \leq o(\berp n)^{2/3}$.
    Suppose $\nu$ and $\pi$ satisfy $\chisquaretrunc{\Planted}{\Null} \leq \delta$.
    Then, there is a constant $C_{\berp} > 1$ (depending only on $\berp$) such that $|\E_{\pi}[\Kr_\ell(\by) \cdot \Ind(|\by|\leq \tau)]| \leq O(\delta \ell^{5/4})$ for all $1 \leq \ell \leq D/C_{\berp}$.
\end{lemma}
\begin{proof}
    We first establish a tail bound using \Cref{lem:tail-bound}.
    By \Cref{lem:binomial-moments}, we have that $\E_{\nu}[\by^{2k}] = (1+o(1)) (2k-1)!!$ for all $k \leq o(\berp n)^{1/3}$.
    Thus, applying \Cref{lem:tail-bound} with $f(x)$ as in \Cref{eq:centered-sum} (a degree-$1$ polynomial), we have that for all $\tau \leq t \leq \sqrt{D}$,
    \begin{align*}
        \Pr_{\pi}[|\by| > t] \leq O(\delta + 2^{-t^2/4}) \cdot e^{-t^2/4}
        \leq O(\delta) \cdot e^{-t^2/4} \mcom
    \end{align*}
    since we assume that $\delta \geq \Omega(n^{-1/2}) \geq 2^{-\tau^2/4}$.
    In particular, $\Pr_{\pi}[|\by| > \tau] \leq O(\delta) \cdot n^{-1}$.

    Next, we bound $\E_{\pi}[\Kr_\ell^2]$.
    $\chisquaretrunc{\Planted}{\Null} \leq \delta$ implies that
    \begin{align*}
        \E_{\pi}[\Kr_\ell^2] \leq \E_{\nu}[\Kr_\ell^2] + \delta \sqrt{\E_{\nu}[\Kr_\ell^4]}
        \leq 1 + \delta e^{C_{\berp}' \ell} \mcom
    \end{align*}
    by hypercontractivity (\Cref{lem:hypercontractivity}) for some constant $C_{\berp}'$.

    For $\ell \leq \tau^2/ 8C_{\berp}'$, we can directly use the first bound from \Cref{lem:function-truncated} (with $\alpha = O(\delta)$):
    \begin{align*}
        \abs{\E_{\by\sim\pi}[\Kr_\ell(\by) \cdot \Ind(|\by|\leq \tau)]}
        &\leq \abs{\E_{\by\sim\pi}[\Kr_\ell(\by)]} + \sqrt{O(\delta) e^{-\tau^2/4} (1+ \delta e^{\tau^2/8})} \\
        &\leq \abs{\E_{\by\sim\pi}[\Kr_\ell(\by)]} + O(\sqrt{\delta} n^{-1/2}) + O(\delta n^{-1/4}) \mper
    \end{align*}

    For $\tau^2 / 8C_{\berp}' \leq \ell \leq D/8C_{\berp}'$, we set $T = \sqrt{8C_{\berp}' \ell} \leq \sqrt{D}$ and apply \Cref{lem:function-truncated} with $\alpha = O(\delta)$ and $\beta = O(\ell^{1/4})$ (since $|\Kr_\ell(y)| \leq O(\ell^{1/4}) \cdot e^{y^2/4}$ by \Cref{lem:krawtchouk-bound}):
    \begin{align*}
        \abs{\E_{\by\sim\pi}[\Kr_\ell(\by) \cdot \Ind(|\by|\leq \tau)]}
        &\leq \abs{\E_{\by\sim\pi}[\Kr_\ell(\by)]} + O(\delta \ell^{5/4}) + \sqrt{O(\delta) e^{-2C_{\berp}' \ell} (1+ \delta e^{C_{\berp}' \ell})} \\
        &\leq \abs{\E_{\by\sim\pi}[\Kr_\ell(\by)]} + O(\delta \ell^{5/4}) + O(\delta n^{-1/2})\mper
    \end{align*}
    Finally, $\chisquaretrunc{\Planted}{\Null} \leq \delta$ implies that $\abs{\E_{\by\sim\pi}[\Kr_\ell(\by)]} \leq \delta$, thus completing the proof.
\end{proof}

\subsection{Proof of \texorpdfstring{\Cref{thm:binom_intro}}{Theorem~\ref{thm:binom_intro}}}

\begin{proof}[Proof of \Cref{thm:binom_intro}]
    For our assumption that $\chisquaretrunc{\Planted}{\Null} \leq \delta$, we may assume that $\delta \geq 1/\sqrt{n}$ since our target bound is $O(\delta+1/\sqrt{n})/\eps^2$.
    Let $\nu, \pi$ be the distributions of $\by = \frac{1}{2\sqrt{n\berp(1-\berp)}} \sum_{i=1}^n (\bx_i - (2\berp-1)) $ when $\bx$ is sampled from $\Null$ (null) and $\Planted$ (planted) respectively.
    By \Cref{lem:pass-to-y}, we can equivalently upper bound $\TV(\nu, \T_\e \pi)$.

    Let $\tau = 2\sqrt{\log n}$.
    We first consider $\ol{\pi}$, the distribution of $\pi$ conditioned on $|\by| \leq \tau$.
    Since $D \gg \tau^2$ and the moments of $\nu$ match those of a standard Gaussian (\Cref{lem:binomial-moments}), by \Cref{lem:tail-bound}, we have that $\Pr_{\by\sim\pi}[|\by| > \tau] \leq O(\delta + 2^{-\tau^2/4}) \cdot e^{-\tau^2/4} \leq O(\delta n^{-1})$.
    Therefore, by the data processing inequality (\Cref{fact:data-processing}),
    \begin{align*}
        \TV(\T_\e \pi, \T_\e \ol{\pi}) \leq \TV(\pi, \ol{\pi}) \leq O(\delta n^{-1}) \mper
        \numberthis \label{eq:noised-planted-vs-truncated}
    \end{align*}
    Thus, it suffices to bound $\TV(\nu, \T_\e \ol{\pi})$, and we will do so by bounding $\chi^2( \T_\e \ol{\pi} \| \nu)$ and using \Cref{fact:tv-chi-squared}.
    By \Cref{lem:chi-squared-fourier} (with $n = 1$) and \cref{fact:noisy_krawtchouk},
    \begin{align*}
        \chi^2(\T_\e \ol{\pi}\| \nu)
        = \sum_{\ell \geq 1} \E_{\by \sim \T_\e \ol{\pi}}[\Kr_\ell(\by)]^2
        = \sum_{\ell\geq 1} (1-\eps)^\ell \E_{\by\sim \ol{\pi}}[\Kr_\ell(\by)]^2 \mper
        \numberthis \label{eq:noised-truncated-to-null}
    \end{align*}
    Next, we split the summation into $\ell \leq T$ and $\ell > T$ for $T = \frac{4}{\eps} \log n$.

    For $\ell\leq T$, we use the assumption that $\chisquaretrunc{\Planted}{\Null} \leq \delta$.
    By \Cref{lem:basis-truncated}, we have that $|\E_{\ol{\pi}}[\Kr_\ell]| \leq O(\delta \ell^{5/4})$ for all $\ell \leq T$.
    Here, we need $D \geq C_{\berp} T$, where $C_{\berp}$ is the constant in \Cref{lem:basis-truncated}.

    For $\ell > T$, we use the bound $\chi^2(\ol{\pi}\|\nu) \leq O(n^3)$ from \Cref{lem:chi-squared-of-truncated} (since $\tau = 2\log n$) and the fact that $\chi^2(\ol{\pi}\|\nu) = \sum_{\ell\geq 1} \E_{\ol{\pi}}[\Kr_\ell]^2$ by \Cref{lem:chi-squared-fourier}.
    Thus, we can upper bound \Cref{eq:noised-truncated-to-null} as follows:
    \begin{align*}
        \chi^2(\T_\e \ol{\pi}\| \nu)
        &\leq \sum_{1 \leq \ell \leq T} (1-\eps)^\ell \cdot O_{\berp}(\delta \ell^{5/4})^2 + (1-\eps)^T \sum_{\ell > T} \E_{\ol{\pi}}[\Kr_\ell]^2 \\
        &\leq O_{\berp}(\delta^2) \sum_{1\leq \ell \leq T} \ell^3 e^{-\ell\eps}  + e^{-\eps T} \cdot \chi^2(\ol{\pi}\|\nu) \\
        &\leq O_{\berp}(\delta^2) \sum_{1\leq \ell \leq T} \ell^3 e^{-\ell\eps} + O(1/n) \mper
    \end{align*}
    The function $z^3 e^{-\eps z}$ is maximized at $3/\eps$.
    Thus, $\sum_{\ell=1}^T \ell^3 e^{-\ell\eps} \leq O(1/\eps^4) + \int_{1}^T z^3 e^{-\eps z} dz \leq O(1/\eps^4)$.
    Then, since $\TV(\nu, \T_\e \ol{\pi}) \leq \frac{1}{2} \sqrt{\chi^2(\T_\e \ol{\pi}\|\nu)}$ (\Cref{fact:tv-chi-squared}), we get
    \begin{align*}
        \TV(\nu, \T_\e \ol{\pi})
        \leq O_{\berp}(\delta/\eps^{2} + 1/\sqrt{n}) \mper
    \end{align*}
    Finally, by \Cref{eq:noised-planted-vs-truncated} and the triangle inequality, we have $\TV(\nu, \T_\e \pi) \leq \TV(\nu, \T_\e \ol{\pi}) + \TV(\T_\e \pi, \T_\e \ol{\pi}) \leq O_{\berp}(\delta + 1/\sqrt{n}) / \eps^2$, completing the proof.
\end{proof}

\section{Gaussian Space I: Low-Degree Symmetric Polynomials}
\label{sec:tv_closeness}
\parhead{Notation.}
In this section, let $\Null = N(0,\I_n)$ and $\Planted$ be a distribution over $\R^n$.
We will use $\delta$ to denote $\chisquaretrunc{\Planted}{\Null}$.

The main goal of this section is to prove that the distributions over vectors $\NullShadow \coloneqq F_k(\Null)$ and $\PlantedShadow \coloneqq F_k(\OU_\e \Planted)$ are close in total variation distance when $\delta$ is sufficiently small.
Recall that $F_k(\cdot) = (\wt{h}_i(\cdot))_{i=1}^k$ is the vector valued polynomial, where $\wt{h}_i(x) = \tfrac 1 {\sqrt{n}} \sum_{j=1}^n h_i(x_j)$ for the $i$-th normalized Hermite polynomial $h_i(\cdot)$.
In particular, we will show the following theorem.
\begin{theorem}[Formal version of \Cref{thm:vector-main}]     \label{thm:main_technical}
    Let $0 < \e \leq 1, \delta > 0$ and $n, k, D \in \N$.
    Let $\Planted, \Null$ be such that $\chisquaretrunc{\Planted}{\Null} \leq \delta$.
    Further, assume that for a sufficiently large constant $C$:
    \[
        k\le \frac{\e \log (n)}{C\log \log (n/\e)},\qquad
        D \geq C \cdot \frac{k^2 \log(k) \cdot  \log(1/\e)}{\e}\mper
    \]
    Then,
    \[
        \dtv(\NullShadow,\PlantedShadow) \leq \exp\parens*{ -\Omega\parens*{ \frac{\eps D}{k} } }+ O(\delta)^{\Omega(1)} \cdot k^{k/2} + n^{-1/2} \,.
    \]
\end{theorem}

\parhead{Outline of this section.}
To prove \cref{thm:main_technical}, we bound the total variation distance by directly computing the $\ell_1$ distance between the density functions by integrating the difference.
Our strategy for bounding this integral is to split it up into two terms, the integral over the ``bulk'', i.e.\ the radius-$\tau$ ball for an appropriately chosen parameter $\tau$, and the integral over the ``tail'', i.e.\ the region outside the radius-$\tau$ ball.
We prove sufficiently strong tail decay for both $\NullShadow$ and $\PlantedShadow$ that the tail integral is negligible.
To control the bulk integral, we use a Fourier analytic argument to prove that the densities are close in $\ell_{\infty}$-norm (\Cref{thm:infty_bound}), which constitutes the bulk of our technical work, carried out in \Cref{sec:infty_closeness}.

\parhead{Definitions and helper lemmas.}
We will also use the following definitions.
The first one is a ``regularity'' condition that ensures sufficiently strong decay of the Fourier transform of $\PlantedShadow$ (see \cref{sec:fourier_decay}).
\begin{definition}[Regularity]
    \label{def:regular_point}
    A fixed vector $y\in \R^n$ is \emph{$\ell$-regular} for a positive integer $\ell$ if for every univariate degree-$\le 2\ell$ sum-of-squares polynomial $p$, we have
    $$
        \frac{1}{2} \E_{\rvg\sim\Null}[p(\rvg)] \leq \frac{1}{n}\sum_{i=1}^n p(y_i) \leq \frac{3}{2}\E_{\rvg\sim\Null}[p(\rvg)]\mper
    $$
    We will use $\calR_\ell\subseteq\R^n$ to refer to the set of all $\ell$-regular points.
\end{definition}

\begin{definition}[Conditional planted distribution]
    We define the \emph{conditional planted distribution} $\Planted'_{\ell}$ as $\Planted | \calR_{\ell}$.
    Throughout the section, we set $\ell = 4k$, and use $\Planted'$ to refer to $\Planted_{4k}'$.
    We use $\FakePlantedShadow$ to refer to the distribution $F_k\parens*{\OU_{\eps}\Planted'}$.
\end{definition}

We can show that a draw from a moment-matching distribution is $(\ell,m)$-regular with high probability.
\begin{lemma}
\torestate{
    \label{lem:regular_points}
    Let $\numm, \ell\in \N$ and $c > 0$ be a sufficiently small constant such that $\numm = 4 \ell$ and $\ell < c \log n$.
    Let $\Planted$ be a distribution such that $\chisquaretrunc{\Planted}{\Null} \leq \delta$.
    Then $\rvy\sim\Planted$ is $\ell$-regular with probability at least $1-(1+\delta)\frac{2^{O(\ell)}}{n}$.}
\end{lemma}

The next lemma states that the distributions $\NullShadow$ and $\PlantedShadow$ have subgaussian moments in every direction up to order $\exp(\log n / k)$.
\begin{lemma}[Moment Bounds]
    \torestate{
        \label{lem:moment_bounds}
        Let $0 < \e \leq 1, 0 < \delta$.
        Let $T, k, \numm \in \N$ such that $k T \leq \numm$ and $T \leq \exp(\log(\tfrac n C)/k)$ for a sufficiently large enough constant $C$.
        For $\Planted$ such that $\chisquaretrunc{\Planted}{\Null} \leq \delta$ and unit $\xi \in \R^k$, we have:
        \begin{enumerate}
            \item $\E_{\rvz\sim\NullShadow} \bracks*{\angles*{\rvz, \xi}^T} \leq O(T)^{T/2}$,
            \item $\E_{\rvz\sim\PlantedShadow} \bracks*{\angles*{\rvz, \xi}^T} \leq (1+\delta) \cdot O(T)^{T/2}$.
        \end{enumerate}
    }
\end{lemma}

\paragraph{$\ell_\infty$-closeness and putting things together.}

The following is our main theorem regarding $\ell_\infty$-closeness of densities at the regular points.
\begin{theorem}
\torestate{
    \label{thm:infty_bound}
    Let $\delta \geq 0, D \in \N$ and $\Planted, \Null$ such that $\chisquaretrunc{\Planted}{\Null} \leq \delta$.
    Let $C,c$ be a sufficiently large (respectively small) absolute constant, and let $0 < \e \leq 1,R > 0$ and $k \in \N$ such that:
    \begin{enumerate}
        \item $k\le \frac{\e \log (n)}{C\log \log (n/\e)}$,
        \item $D \geq C \cdot \frac{k^2 \log(1/\e)}{\e}$.
    \end{enumerate}
    Then:
    \[
        \sup_{y \in \calR_{k}} \Abs{\NullShadow(y) - \FakePlantedShadow(y)} \leq e^{-\Omega\parens*{\e \cdot \tfrac D k}} + O(\delta) + \frac{2^{O(k)}} n \mper
    \]
    In particular, the above is always at most 
    $$
        e^{-\Omega\parens*{\e \cdot \tfrac D k}} + O(\delta) + n^{-1+o(1)} \,.
    $$
}
\end{theorem}

Using these tools, we can prove \cref{thm:main_technical}.
\begin{proof}[Proof of \cref{thm:main_technical}]
    For $\tau \geq C \sqrt{k}$ to be chosen later, define $B_{\tau}$ as the $\ell_2$-ball of radius $\tau$.
    Recall that $\calR_{\ell}$ is the set of all $\ell$-regular points and $\FakePlantedShadow$ is the distribution of $F_k(\OU_\e [\Planted | \calR_{4k}])$.
    We may write:
    \begin{align*}
        \dtv(\NullShadow,\PlantedShadow) &\le \int_{\R^k} \abs*{ \NullShadow(x) - \FakePlantedShadow(x) } \cdot \Ind[\NullShadow(x) \ge \FakePlantedShadow(x)] \dif x + \dtv(\PlantedShadow, \FakePlantedShadow) \\
        &\le \NullShadow(\ol{B_{\tau}}) +  \vol(B_{\tau}) \cdot \sup_{x\in B_{\tau}\cap\calR_{4k}} \abs*{ \NullShadow(x) - \PlantedShadow(x) } + \dtv(\PlantedShadow, \FakePlantedShadow)\mper
    \end{align*}
    We now bound each of the terms in the above sum.

    \parhead{Bound on $\NullShadow(\ol{B_{\tau}})$.}
    By \Cref{lem:moment_bounds} with $T = k$ and Markov's inequality, we know that for any fixed unit vector $\xi\in\R^k$,
    \[
        \Prob_{\rvz\sim\NullShadow}\bracks*{\Iprod{\xi, \rvz} \geq \tau } \leq \Paren{\frac {O(k)} {\tau^2}}^{k/2} \mper
    \]
    Note that this is at most $2^{-10k}$, but potentially smaller depending on the exact choice of $\tau$.
    Thus, by a union bound over choices of $\xi$ over a net, $\bz\sim\NullShadow$ is in $B_{\tau}$ except with probability at most $\paren{\frac {O(k)} {\tau^2}}^{k/2}$.

    \parhead{Bound on $\vol(B_{\tau}) \cdot \sup_{x\in B_{\tau}\cap\calR_{k,m}} \abs*{ \NullShadow(x) - \PlantedShadow(x) }$.}
    By \Cref{thm:infty_bound}, we have:
    \begin{align*}
        \vol(B_{\tau}) \cdot \sup_{x\in B_{\tau}\cap\calR_{k,m}} \abs*{ \NullShadow(x) - \PlantedShadow(x) } \le O(\tau^k) \cdot \parens*{ e^{-\Omega\parens*{ \tfrac{\eps D}{k} }} + O(\delta) + n^{-1+o(1)} } \mper
    \end{align*}

    \parhead{Bound on $\dtv(\PlantedShadow, \FakePlantedShadow)$.}
    By \Cref{lem:regular_points}, this is bounded by $(1+\delta) \tfrac {2^{O(k)}} n$.

    \parhead{Putting things together}
    Since $t^k \geq 1$ and $2^{O(k)} = n^{o(1)}$, overall, this shows that
    \begin{align*}
        \dtv(\NullShadow,\PlantedShadow) &\leq O\parens*{\frac{O(k)}{\tau^2}}^{k / 2} + O(\tau^k) \cdot \parens*{e^{-\Omega\parens*{\e \cdot \tfrac D k}} + O(\delta) + n^{-1+o(1)}}\mper
    \end{align*}
    Choosing $\tau = C\sqrt{k} + \min \left\{ e^{\eps\cdot\frac{D}{Ck^2}}, O\left(\frac{1}{\delta}\right)^{1/Ck}, n^{1/Ck} \right\}$ completes the proof.
\end{proof}

\subsection{\texorpdfstring{$\ell_\infty$}{l-infinity}-Closeness of Densities}
\label{sec:infty_closeness}

In this section we will prove \cref{thm:infty_bound}, restated below.
\restatetheorem{thm:infty_bound}

We will prove this by considering the Fourier transform of both distributions.
In particular, we will show that both distributions have sufficient Fourier decay so that we can safely ignore the contribution of large-norm frequency vectors, and for all other frequency vectors we use the moment-matching property to show that the two must be close.
We will state formally what we require and then derive a proof of \cref{thm:infty_bound}.

More specifically, we will prove the following Fourier decay result whose prove we give in \cref{sec:fourier_decay}.
\begin{lemma}[Fourier decay]
\torestate{
    \label{lem:fourier_decay}
    Let $\e\in (0, 1]$, let $k \leq \frac{\log(\eps n)}{C\log \log n}$, and $R > C\sqrt{k/\eps}$ for a sufficiently large constant $C>0$.
    Then,
    \begin{align*}
        \int_{\norm{\xi} > R} \Abs{\NullShadowHat(\xi)} \, d\xi
        \quad \text{and} \quad
        \int_{\norm{\xi} > R} \Abs{\FakePlantedShadowHat(\xi)} \, d\xi
        \leq \Omega(\eps)^{-k/2} \cdot e^{-\Omega(\eps R^2)} + e^{-\sqrt{n}} \mper
    \end{align*}
}
\end{lemma}

For frequency vectors of low norm, we will prove the following (the proof is given in \cref{sec:frequency_matching}).
\begin{lemma}[Frequency matching]
\torestate{
    \label{lem:frequency_matching_NEW}
    Let $\delta, R \geq 0$ and $k,D \in \N$ such that $R^2 \leq \min\left\{c \cdot \frac {D} k, \exp\parens*{\log(cn)/k} \right\}$ where $c$ is sufficiently small.
    Then
    $$
        \sup_{\norm{\xi} \leq R} \Abs{\FakePlantedShadowHat(\xi) - \NullShadowHat(\xi) } \leq 2^{-\Omega(R^2)} + O(\delta) +(1+\delta) \cdot \frac{2^{O(k)}} n\,.
    $$
}
\end{lemma}

\paragraph{Deriving $\ell_\infty$ bounds.}

We will next derive \cref{thm:infty_bound}.

\begin{proof}[Proof of \cref{thm:infty_bound}]

Let $R$ to be chosen below.
Then, for all $y \in \calR_{k,m}$,
\begin{align*}
    \Abs{\NullShadow(y) - \FakePlantedShadow(y)} &= \Abs{\int_{\R^k} e^{-i\iprod{\xi,y}} \Paren{{\NullShadowHat}(\xi) -{\FakePlantedShadowHat}(\xi)} \dif\xi } \\
    &\leq \int_{\R^k} \Abs{{\NullShadowHat}(\xi) - {\FakePlantedShadowHat}(\xi)} \dif\xi \\
    &\leq \int_{\norm{\xi} \leq R} \Abs{{\NullShadowHat}(\xi) -{\FakePlantedShadowHat}(\xi)} \dif\xi + \int_{\norm{\xi} > R} \Abs{{\NullShadowHat}(\xi)} \dif\xi + \int_{\norm{\xi} > R} \Abs{{\FakePlantedShadowHat}(\xi)}  \dif\xi \,.
\end{align*}

\paragraph{Choosing parameters.}
We next choose $R$ and such that the pre-conditions of \cref{lem:fourier_decay,lem:frequency_matching_NEW} are met.
Set $R^2 =\min\left\{c \cdot \frac {D} k, \exp\parens*{\log(cn)/k}\right\}$ for a sufficiently small constant $c$.
We first verify that $R^2 \geq C' \frac k \e$ for a large constant $C'$.
Note that by our condition on $D$ we have that $\frac{D}{k} \geq C \frac k \e$.
It remains to check that 
$$
    C \cdot \frac k \e \leq \exp\parens*{\log(cn)/k} \,.
$$
Taking logs this is implied by
$$
    C k \log (k/\e) \leq \log n \mcom 
$$
which holds when $k \leq \e \log(n) / C' \log \log (n/\e)$.
Lastly, it also holds that $k \leq \log(n/\e)/C \log\log n$ as required by \cref{lem:fourier_decay}.

\paragraph{Final bound.}
By \cref{lem:fourier_decay}, the last two terms are at most 
\begin{align*}
    \Omega(\e)^{-k/2} \cdot e^{-\Omega(\e R^2)} + e^{-\sqrt{n}} 
    &= \Omega(\e)^{-k/2} \cdot \parens*{e^{-\Omega\parens*{\e \cdot \tfrac D k}}+ e^{-\e \cdot \Omega(n)^{1/k}}} + e^{-\sqrt{n}} \\
    &= \Omega(\e)^{-k/2} \cdot \parens*{e^{-\Omega\parens*{\e \cdot \tfrac D k}}+ e^{-\e \cdot \Omega(n)^{1/k}}} \,.
\end{align*}
We claim that under our parameter choices this simplifies to 
$$
e^{-\Omega\parens*{\e \cdot \tfrac D k}}+ e^{-\e \cdot \Omega(n)^{1/k}} \,.
$$
For this, it is enough to show that $\Omega(n)^{1/k} \geq \frac{C'} \e k \log(1/\e)$ and $\e \cdot \frac D k \geq C' k \log(1/\e)$.
The second one follows immediately by our lower bound on $D$.
For the first one, note that $k \leq \e \log n/C \log \log (n/\eps)$ so $\frac{C'} \e k \log (1/\e) \leq \log n \log (1/\e)$.
Taking logs we obtain
\begin{align*}
    \frac{\log(cn)}{k} \geq \frac{\log n}{2k} \geq \frac {C \cdot \log \log (n/\eps)} {2\e} \geq \frac{C \cdot \parens*{\log \log n + \log \log (1/\e)}}{2} \geq C' \log k + C' \log \log(1/\e) \,.
\end{align*}

Further, by \cref{lem:frequency_matching_NEW}, the first term is at most
    $$
        2^{-\Omega(R^2)} + O(\delta) + (1+\delta) \cdot \frac{2^{O(k)}}{ n} \,.
    $$
    Note that $2^{-\Omega(R^2)} \leq e^{-\Omega(\e R^2)}$, and $\delta \cdot \frac{2^k} n \leq O(\delta)$ (by our choice of $k$).
    Further, since $\log(cn)/k \geq \frac {C \cdot \log \log n} \e$, it holds that $\Omega(n)^{1/k} \geq \log^{(C/\e)}(n) \geq \log^C(n)/\e$.
    So $e^{-\e \cdot \Omega(n)^{1/k}} \leq e^{-\log^{C}(n)} \ll 2^{O(k)} / n$.
    Thus, we obtain the overall bound
\begin{equation*}
   \Abs{\NullShadow(y) - \FakePlantedShadow(y)} \leq e^{-\Omega\parens*{\e \cdot \tfrac D k}} + O(\delta) + \frac{2^{O(k)}} n \mper \qedhere
\end{equation*}

For the simplified form, we notice that $2^{O(k)} = n^{o(1)}$.
\end{proof}

\subsubsection{Fourier Decay}
\label{sec:fourier_decay}
In this section we will prove \cref{lem:fourier_decay}, restated below.
\restatelemma{lem:fourier_decay}

\noindent
The key estimate we need for the proof of this lemma is the following set of bounds on the expectations of Fourier characters of low-degree polynomials.
The proof is given in \Cref{sec:characteristic}.
\begin{lemma}
\torestate{
    \label{lem:PolyFourierBoundLem}
    Let $p$ be a univariate polynomial of degree $\le k$.
    Then:
    \begin{itemize}
        \item If $\var_{\rvg\sim\calN(0,1)}[p(\rvg)] \leq 9^{-k}$, we have that 
        $$
            \left|\E_{\rvg\sim\calN(0,1)}[\exp(i p(\rvg))]\right| \leq 1 - \frac{1}{4}\var_{\rvg\sim\calN(0,1)}[p(\rvg)] \mper
        $$
        \item If $\var_{\rvg\sim\calN(0,1)}[p(\rvg)] = k^{c k}$ for any $c\in \R$, we have that
        $$
            \left|\E_{\rvg\sim\calN(0,1)}[\exp(i p(\rvg))]\right| \leq 1-\frac{1}{16k} k^{- |c| k} \mper
        $$
        \item Finally, if $\var_{\rvg\sim\calN(0,1)}[p(\rvg)]\geq k^{C k}$ for a sufficiently large $C > 0$, we have that
        $$
            \left|\E_{\rvg\sim\calN(0,1)}[\exp(i p(\rvg))]\right| \leq O \left( \var_{\rvg\sim\calN(0,1)}[p(\rvg)]^{-\frac{1}{4k}} \right) \mper
        $$
    \end{itemize}
}
\end{lemma}

In the proof of \Cref{lem:fourier_decay}, we also need a few well-known facts.

\begin{fact} \label{fact:ball-volume-gaussian-norm}
    Let $x$ be $k$-dimensional.
    For any $r > 0$, $c > 0$, and $R > 4\sqrt{k/c}$,
    \begin{align*}
        \vol(B_r) = \int_{\|x\| \leq r} dx \leq O(r^k)
        \mcom  \quad
        \int_{\|x\| \geq R} e^{-c\|x\|^2}\ dx \leq (2\pi/c)^{k/2} \cdot e^{-cR^2/2} \mper
    \end{align*}
\end{fact}

\noindent
With these estimates, we can now complete the proof of \Cref{lem:fourier_decay}:


\begin{proof}[Proof of \Cref{lem:fourier_decay}]
    For any $\xi = (\xi_1,\ldots,\xi_k) \in \R^k$, let $p_\xi (x): \R \to \R$ be the univariate polynomial $p_\xi (x) = \sum_{\ell = 1}^k \xi_\ell h_\ell (x)$. 
    Observe that by the orthogonality of the Hermite polynomials, $\Var p_\xi (\bg) = \|\xi\|^2$.
    Recall that the vector $F_k$ (on input $x \in \R^n$) has coordinates $\tfrac 1 {\sqrt{n}} \sum_{j=1}^n h_\ell(x_j)$ for $\ell = 1, \ldots, k$.
    Thus,
    \begin{align*}
        \iprod{F_k(x), \xi} = \sum_{\ell = 1}^k (F_k(x))_\ell \cdot \xi_\ell = \frac 1 {\sqrt{n}} \sum_{\ell = 1}^k \parens*{\sum_{j=1}^n h_\ell(x_j)} \xi_\ell = \frac 1 {\sqrt{n}} \sum_{j=1}^n \parens*{\sum_{\ell=1}^k h_\ell(x_j) \xi_\ell} = \frac 1 {\sqrt{n}} \sum_{j=1}^n p_\xi(x_j) \,.
    \end{align*}
    The Fourier coefficients $\NullShadowHat(\xi)$ and $\FakePlantedShadowHat(\xi) \in \C$ are the characteristic functions of $\angles*{F_k, \xi}$ under $\Null$ and $\OU_\e \Planted'$ respectively (where $\Planted'$ is the planted distribution conditioned on being $4k$-regular; see \Cref{def:regular_point}):
    $\NullShadowHat(\xi) = \E_{\rvx\sim \Null}[e^{i \angles{F_k(\rvx), \xi}}]$
    and $\FakePlantedShadowHat(\xi) = \E_{\rvx\sim \OU_\e\Planted'}[e^{i \angles*{F_k(\rvx), \xi}}]$.
    We will focus on bounding $\FakePlantedShadowHat(\xi)$, since $\NullShadowHat(\xi)$ has the same bound if we set $\eps = 1$.
    
    Note that $\rvx \sim \OU_\e \Planted'$ can be sampled by first sampling $\rvy \sim \Planted'$ and setting $\rvx = \sqrt{1-\eps} \rvy + \sqrt{\eps}\rvg$ for $\rvg \sim \calN(0, \I_n)$.
    Our strategy is to fix $y\in \R^n$, which is assumed to be regular, and use the randomness of $\rvg$ to upper bound the characteristic functions (for large frequencies $\|\xi\| > R$).

    For a fixed $y\in \R^n$, let $\rvz = \sqrt{1-\eps}y + \sqrt{\eps} \rvg$ (depending on $y$ implicitly), and
    \begin{align*}
        \E_{\rvg\sim \calN(0,\I_n)}[e^{i \angles{F_k(\rvz), \xi}}]
        &= \E_{\rvg\sim \calN(0, \I_n)} \bracks*{\exp \parens*{\frac{i}{\sqrt{n}} \sum_{j=1}^n p_{\xi}(\sqrt{1-\eps} y_j + \sqrt{\eps} \rvg_j) } } \\
        &= \prod_{j=1}^n \E_{\rvg_j \sim \calN(0,1)} \bracks*{ \exp\parens*{\frac{i}{\sqrt{n}} \cdot p_{\xi}(\sqrt{1-\eps} y_j + \sqrt{\eps} \rvg_j) }} \mper
        \numberthis \label{eq:fourier-product-form}
    \end{align*}

    Next, let $\mu(y_j) \coloneqq \E_{\rvg\sim\calN(0,1)}[p_{\xi}(\sqrt{1-\eps}y_j + \sqrt{\eps} \rvg)]$ and $v(y_j) \coloneqq \Var_{\rvg\sim \calN(0,1)}[p_{\xi}(\sqrt{1-\eps} y_j + \sqrt{\eps} \rvg)] = \E_{\rvg}[p_{\xi}(\rvz_j)^2] - \mu(y_j)^2$
    (here we drop the dependence on $\xi,\eps$ for convenience).
    Note that by the definition of the noise operator, $\mu(y_j) = \sum_{\ell=1}^k (1-\eps)^{\ell/2} \xi_{\ell} h_k(y_j)$.
    Moreover, observe that $v(y_j)$ is a univariate sum-of-squares polynomial of degree at most $2k$, since the variance can be written as an expected square.
    Thus, by the $4k$-regularity of $y$ (\Cref{def:regular_point}),
    we have
    \begin{align*}
        \frac{1}{n} \sum_{j=1}^n v(y_j) 
        &\geq \frac{1}{2} \E_{\rvh \sim \calN(0,1)}[v(\rvh)]
        = \frac{1}{2} \E_{\rvh} \bracks*{ \E_{\rvg} \bracks*{p_{\xi}(\sqrt{1-\eps} \rvh + \sqrt{\eps} \rvg)^2} - \mu(\rvh)^2 } \\
        &= \frac{1}{2} \parens*{ \E_{\rvg}[p_{\xi}(\rvg)^2] - \E_{\rvg}[\mu(\rvg)^2]}
        = \frac{1}{2} \parens*{ \|\xi\|^2 - \sum_{\ell=1}^k (1-\eps)^{\ell} \xi_\ell^2 } \\
        &\geq \frac{\eps}{2} \|\xi\|^2 \mper
        \numberthis \label{eq:variance-sum-lower-bound}
    \end{align*}
    On the other hand, by the $4k$-regularity of $y$ again, we have
    \begin{align*}
        \frac{1}{n} \sum_{j=1}^n v(y_j)^2 
        &\leq \frac{3}{2} \E_{\rvh\sim \calN(0,1)}[v(\rvh)^2]
        \leq \frac{3}{2} \E_{\rvh} \bracks*{ \E_{\rvg} \bracks*{ p_{\xi}(\sqrt{1-\eps} \rvh + \sqrt{\eps} \rvg)^2 }^2 } \\
        &\leq \frac{3}{2} \E_{\rvg} \bracks*{p_{\xi}(\rvg)^4}
        \leq \frac{3}{2} \cdot 9^k \E_{\rvg}[p_{\xi}(\rvg)^2]^2
        = \frac{3}{2} \cdot 9^k \|\xi\|^4 \mcom
        \numberthis \label{eq:squared-variance-upper-bound}
    \end{align*}
    where we use Jensen's inequality and the 2-to-4 hypercontractivity for polynomials in Gaussian space (\Cref{fact:gauss_hypercontractivity}).
    The above upper bound can be interpreted as a ``spreadness'' condition of $\{v(y_j)\}_{j\in[n]}$, which will be important later.

    To bound \Cref{eq:fourier-product-form}, we will use \Cref{lem:PolyFourierBoundLem} which gives upper bounds in terms of the variances.
    In our case, for each $j\in [n]$, we have $\Var[\frac{1}{\sqrt{n}} p_{\xi}(\sqrt{1-\eps}y_j + \sqrt{\eps} \rvg_j] = v(y_j) / n$. Thus,
    \begin{align*}
        \prod_{j=1}^n \abs*{ \E_{\rvg_j \sim \calN(0,1)} \bracks*{ \exp\parens*{\frac{i}{\sqrt{n}} \cdot p_{\xi}(\sqrt{1-\eps} y_j + \sqrt{\eps} \rvg_j) }}}
        \leq \exp \parens*{- \sum_{j=1}^n  \rho(v(y_j)/n)} \mcom
    \end{align*}
    where $\rho$ is defined as
    \begin{align*}
        \rho(a) \coloneqq \Omega(1) \cdot
        \begin{cases}
            a & 0 \leq a \leq 2^{-Ck} \mcom \\
            k^{-Ck} &  2^{-Ck} \leq a \leq k^{Ck} \mcom \\
            \frac{1}{k} \log a & a \geq k^{Ck} \mcom
        \end{cases}
    \end{align*}
    for a large enough universal constant $C > 1$.
    Note that in the regime $a \geq 2^{-Ck}$, $\rho(a)\geq \Omega(k^{-Ck})$.
    
    For a fixed $y$, let $B_1 \coloneqq \{j\in[n]: v(y_j)/n \leq 2^{-Ck} \}$ and $B_2 \coloneqq \{j\in [n]: v(y_j)/n > 2^{-Ck}\}$.
    By \Cref{eq:variance-sum-lower-bound}, we must have either
    \begin{enumerate}[(1)]
        \item $\sum_{j\in B_1} v(y_j)/n \geq \frac{\eps}{4} \|\xi\|^2$, or
        
        \item $\sum_{j\in B_2} v(y_j)/n \geq \frac{\eps}{4} \|\xi\|^2$.
    \end{enumerate}
    In the first case,
    \begin{align*}
        \exp \parens*{- \sum_{j=1}^n  \rho(v(y_j)/n)}
        \leq \exp\parens*{-\sum_{j\in B_1} \rho(v(y_j)/n)}
        \leq \exp\parens*{-\Omega(1) \sum_{j\in B_1} v(y_j)/n}
        \leq \exp\parens*{-\Omega(\eps)\|\xi\|^2} \mper
    \end{align*}
    In the second case, we use \Cref{eq:squared-variance-upper-bound} to get a lower bound on $|B_2|$:
    \begin{align*}
        \frac{\eps}{4}\|\xi\|^2 \leq \sum_{j\in B_2} \frac{v(y_j)}{n}
        \leq \sqrt{|B_2| \cdot \sum_{j\in B_2} \frac{v(y_j)^2}{n^2}}
        \leq \sqrt{|B_2| \cdot \frac{1}{n} \cdot \frac{3}{2} 9^k \|\xi\|^4} \mper
    \end{align*}
    This implies that $|B_2| \geq \Omega(\eps^2 9^{-k} n)$.
    By definition of $B_2$ and $\rho$, we have $\rho(v(y_j)/n) \geq \Omega(k^{-Ck})$ for each $j\in B_2$.
    Thus,
    \begin{align*}
        \exp \parens*{- \sum_{j=1}^n  \rho(v(y_j)/n)} \leq \exp\parens*{-|B_2| \cdot \Omega(k^{-Ck})}
        \leq \exp\parens*{-\Omega(\eps^2) k^{-Ck} n} \mper
    \end{align*}
    Therefore, combining both cases, it follows that any $y \in \R^n$ that is $4k$-regular satisfies
    \begin{align*}
        \exp \parens*{- \sum_{j=1}^n  \rho(v(y_j)/n)} 
        \leq e^{-\Omega(\eps \|\xi\|^2)} + e^{-\Omega(\eps^2 k^{-Ck} n)} \mper
    \end{align*}
    We will use the bound above for the regime $R^2 \leq \|\xi\|^2 \leq T \coloneqq n \cdot k^{Ck}/\eps$.
    \begin{align*}
        \int_{R \leq \|\xi\| \leq \sqrt{T}} |\FakePlantedShadowHat(\xi)|\ d\xi
        &\leq \int_{\|\xi\| \geq R} e^{-\Omega(\eps \|\xi\|^2)}\ d\xi
        + \int_{\|\xi\|\leq \sqrt{T}} e^{-\Omega(\eps^2 k^{-Ck} n)} \ d\xi \\
        &\leq \Omega(\eps)^{-k/2} \cdot e^{-\Omega(\eps R^2)} + O(T^{k/2}) \cdot e^{-\Omega(\eps^2 k^{-Ck} n)}  \mper
        \numberthis \label{eq:bound-small-frequency}
    \end{align*}
    Here, the second inequality uses \Cref{fact:ball-volume-gaussian-norm} and $R \gg \sqrt{k/\eps}$.
    Moreover, for $k \leq \frac{\log (\eps n)}{C'\log \log n}$ with a large enough constant $C' > C$, we have $\eps^2 k^{-Ck}n \gg \sqrt{n}$.
    Thus, we can simplify the second term in \Cref{eq:bound-small-frequency} to $e^{-\sqrt{n}}$.

    Now, we focus on the regime $\|\xi\|^2 > T$.
    By \Cref{eq:variance-sum-lower-bound,eq:squared-variance-upper-bound}, we have $\E_{j\sim[n]}[v(y_j)] \geq \frac{\eps}{2} \|\xi\|^2$ and $\E_{j\sim [n]}[v(y_j)^2] \leq \frac{3}{2}9^k \|\xi\|^4$.
    Thus, by the Paley-Zygmund inequality, we have
    $\Pr_{j\sim [n]}[v(y_j) \geq \frac{\eps}{4}\|\xi\|^2] \geq \Omega(\eps^2 9^{-k})$.
    In other words, $\Omega(\eps^2 9^{-k} n)$ of coordinates $j\in[n]$ have $v(y_j)/n \geq \frac{\eps}{4n} \|\xi\|^2$, which is larger than $k^{Ck}$, meaning $\rho(v(y_j)/n) \geq \Omega(\frac{1}{k}) \cdot \log(\frac{\eps\|\xi\|^2}{4n})$.
    Thus,
    \begin{align*}
        \exp \parens*{- \sum_{j=1}^n  \rho(v(y_j)/n)}
        \leq \parens*{\frac{\eps \|\xi\|^2}{4n}}^{-\Omega(\eps^2 9^{-k}n / k)} \mper
    \end{align*}
    Integrating this over the region $\|\xi\| > \sqrt{T}$, we get
    \begin{align*}
        \int_{\|\xi\| > \sqrt{T}} |\FakePlantedShadowHat(\xi)|\ d\xi
        \leq \int_{\sqrt{T}}^\infty \parens*{\frac{\eps r^2}{4n}}^{-\Omega(\eps^2 9^{-k} n/k)} O(r^{k-1}) \ dr
        \leq \parens*{\frac{4n}{\eps T}}^{\Omega(\eps^2 9^{-k} n/k)} \cdot O(T)^{k/2} \mper
    \end{align*}
    For our choice $T = n k^{Ck}/\eps$ and $k \ll \log n$, the above is $o(e^{-\sqrt{n}})$.
    Combined with \Cref{eq:bound-small-frequency}, this gives the desired bound $\int_{\|\xi\|>R} |\FakePlantedShadowHat(\xi)|\ d\xi \leq \Omega(\eps)^{-k/2} \cdot e^{-\Omega(\eps R^2)} + e^{-\sqrt{n}}$.

    As mentioned earlier, the bound for $\NullShadowHat(\xi)$ follows directly by setting $\eps = 1$.
    This completes the proof.
\end{proof}

\subsubsection{Frequency Matching}
\label{sec:frequency_matching}
In this section we will prove \cref{lem:frequency_matching_NEW}, restated below.
\restatelemma{lem:frequency_matching_NEW}

We will prove this by Taylor expanding the exponential function inside the Fourier transform.
By moment-matching, we can then argue that the first few terms of the expansion match approximately and that the error terms are small because both $\NullShadow$ and $\PlantedShadow$ have sub-gaussian moments.

\begin{proof}[Proof of \cref{lem:frequency_matching_NEW}]
    Let $T = CR^2$ for a sufficiently large constant $C$ and for simplicity assume this is an integer.
    Fix one $\xi$ of norm at most $R$.
    We have to bound
    $$
        \left| \E_{\rvz\sim\NullShadow}[\exp(i \iprod{\xi, \rvz})] - \E_{\rvz\sim\FakePlantedShadow}[\exp(i \iprod{\xi, \rvz})] \right| \,.
    $$
    We first note that $\exp(i \iprod{\xi, \cdot})$ is a function of magnitude at most 1.
    Thus, since $\TV(\PlantedShadow,\FakePlantedShadow) \leq (1+\delta) \frac{2^{O(k)}} n$ by \cref{lem:regular_points}, it is sufficient to show that
    $$
        \left| \E_{\rvz\sim\NullShadow}[\exp(i \iprod{\xi, \rvz})] - \E_{\rvz\sim\PlantedShadow}[\exp(i \iprod{\xi, \rvz})] \right| \leq 2^{-\Omega(R^2)} + O(\delta)\,.
    $$
    Since $e^{i y} = \cos(y) + i \cdot \sin(y)$ it is enough to bound 
    $$
        \left| \E_{\rvz\sim\NullShadow}[\cos( \iprod{\xi, \rvz})] - \E_{\rvz\sim\PlantedShadow}[\cos(\iprod{\xi, \rvz})] \right| + \left| \E_{\rvz\sim\NullShadow}[\sin( \iprod{\xi, \rvz})] - \E_{\rvz\sim\PlantedShadow}[\sin(\iprod{\xi, \rvz})] \right| \,.
    $$
    We focus on the first term, the second term is completely analogous.
    Consider a fixed vector $z$.
    Since all derivatives of $\cos(\cdot)$ are at most one in magnitude, it follows by Taylor expansion that there exist $r_z$ of magnitude at most 1 such that 
    $$
        \cos(\iprod{\xi,z}) = \underbrace{\sum_{\ell=0}^{T-1} \frac{(-1)^\ell}{(2\ell)!} \iprod{\xi,z}^{2\ell}}_{\coloneqq p_T(\iprod{\xi,z})} + r_z \cdot \frac{\iprod{\xi,z}^{2T}}{(2T)!} \mper
    $$
    Thus, by small LDLR and the triangle inequality we obtain
    \begin{align*}
        &\left| \E_{\rvz\sim\NullShadow}[\cos( \iprod{\xi, \rvz})] - \E_{\rvz\sim\PlantedShadow}[\cos(\iprod{\xi, \rvz})] \right| \\
        &\leq \left| \E_{\rvz\sim\NullShadow}[p_T( \iprod{\xi, \rvz})] - \E_{\rvz\sim\PlantedShadow}[p_T(\iprod{\xi, \rvz})] \right| + \left| \E_{\rvz\sim\NullShadow} \bracks*{r_{\rvz} \frac{\iprod{\xi,z}^{2T}}{(2T)!}} \right| + \left| \E_{\rvz\sim\PlantedShadow} \bracks*{r_{\rvz} \frac{\iprod{\xi,z}^{2T}}{(2T)!}} \right| \\
        &\leq \delta \cdot \sqrt{\Var_{\NullShadow}\bracks*{p_T(\iprod{\xi,\rvz})}} + \left| \E_{\rvz\sim\NullShadow} \bracks*{r_{\rvz} \frac{\iprod{\xi,z}^{2T}}{(2T)!}} \right| + \left| \E_{\rvz\sim\PlantedShadow} \bracks*{r_{\rvz} \frac{\iprod{\xi,z}^{2T}}{(2T)!}} \right|  \mper
    \end{align*}
    Note that $T \cdot k = CR^2 \leq D$ and $T = C R^2 \leq \exp(\log(c \cdot n)/k)$, when $c$ is small enough.
    Thus, $\iprod{\xi,\rvz}$ has order-$T$ sub-gaussian moments under both $\PlantedShadow$ and $\NullShadow$ (cf.~\cref{lem:moment_bounds}) and we can bound the last two terms as
    \begin{align*}
        \left| \E_{\rvz\sim\NullShadow} \bracks*{r_{\rvz} \frac{\iprod{\xi,z}^{2T}}{(2T)!}} \right| + \left| \E_{\rvz\sim\PlantedShadow} \bracks*{r_{\rvz} \frac{\iprod{\xi,z}^{2T}}{(2T)!}} \right| \leq (1+\delta) \cdot \frac{O\parens*{\norm{\xi}^2 T}^T}{(2T)!} = (1+\delta) \cdot O\parens*{\frac{R^2}{T}}^T  \leq (1+\delta) \cdot 2^{-\Omega(R^2)} \mper
    \end{align*}

    It remains to bound the variance term.
    Using $(a+b)^2 \leq 2a^2 + 2b^2$ and sub-gaussian moments again, we obtain
    \begin{align*}
        \Var_{\rvz \sim \NullShadow}\bracks*{p_T(\iprod{\xi,\rvz})} &\leq \E_{\rvz \sim \NullShadow} \bracks*{p_T(\iprod{\xi,\rvz})^2} \leq 2 \E_{\rvz \sim \NullShadow} \bracks*{\parens*{\cos(\iprod{\xi,\rvz}) -p_T(\iprod{\xi,\rvz})}^2} + 2\E_{\rvz \sim \NullShadow} \bracks*{\cos(\iprod{\xi,\rvz})^2} \\
        &\leq 2 \E_{\rvz \sim \NullShadow} \bracks*{\parens*{\frac{\iprod{\xi,z}^{2T}}{(2T)!}}^2} + 2 = O\parens*{\frac{R^2}{T}}^{2T} + 2 = O\parens*{1}\mper
    \end{align*}

    Overall, we obtain 
    $$
        \left| \E_{\rvz\sim\NullShadow}[\cos( \iprod{\xi, \rvz})] - \E_{\rvz\sim\PlantedShadow}[\cos(\iprod{\xi, \rvz})] \right| \leq \delta \cdot O(1) + (1+\delta) \cdot 2^{-\Omega(R^2)} = O(\delta) + 2^{-\Omega(R^2)} \mper
    $$
    
\end{proof}

\subsection{Proofs of Small Helper Lemmas}
\label{sec:proofs_helper}

In this section, we give the proofs of \cref{lem:regular_points} (high probability regular points) and \cref{lem:moment_bounds} (Gaussian moment bounds).

\paragraph{High probability regularity.}
We start with the proof of \cref{lem:regular_points}, restated below.
\restatelemma{lem:regular_points}

\begin{proof}
    We prove the desired statement via an $\e$-net argument.
    Let $\calP_\ell$ be the set of degree-$\le\ell$ polynomials over $\R$ and define the following two norms on $\calP_\ell$:
    \[
        \norm{p}^2 = \E_{\rvg\sim\calN(0,1)}[p(\rvg)^2] \; \mbox \; , \; \; \norm{p}_y^2 = \frac{1}{n}\sum_{i=1}^n p(y_i)^2 \; .
    \]
    We need to show that
    \[
        \Prob_{\rvy\sim\Planted}\bracks*{\forall p \neq 0 \in \calP_k \colon \frac{\norm{p}_{\rvy}^2}{\norm{p}^2} \in [0.5,1.5]} \geq 1 - \frac{(1+\delta)\cdot 2^{O(\ell)}}{n}\,.
    \]
    The more general case for sum-of-squares polynomials follows readily.

    \parhead{Reducing to a net.}
    Let $\calC$ be an $\eta$-net of the unit ball under $\norm{\cdot}$ comprised of unit norm elements.
    Note that since $\norm{\cdot}$ is an $\ell$-dimensional Euclidean metric, there is such a net on $\frac{1}{\eta^{O(\ell)}}$ elements.
    We first prove:
    \begin{quote}
    For a fixed $y$, suppose for all $p\in\mathcal{C}$, $\frac{\norm{p}_{y}}{\norm{p}}\in [1-\eta,1+\eta]$ for $\eta \le 0.1$, then $y$ is $\ell$-regular.
    \end{quote}

    Note that any $p$ such that $\norm{p} = 1$ can be written as $p = p_0 + \eta p'$, where $p_0 \in \calC$, and $\norm{p'} \leq 1$.
    Iterating this on $p'$, we can write 
    \[
        p = \sum_{i=0}^{\infty} \eta^i c_i p_i \,,
    \]
    where all $p_i \in \calC$, and all $|c_i| \le 1$.
    By the triangle inequality:
    \[
        \norm{p_0}_y - \sum_{i\ge 1} \eta^i \norm{p_i}_y \le \norm{p}_y \le \norm{p_0}_y + \sum_{i\ge 1} \eta^i \norm{p_i}_y
    \]
    Using the assumption on $p_i\in\calC$ and $\norm{p_i} = 1$, we obtain:
    \begin{align*}
        \norm{p}_y &\le (1+\eta)\parens*{\norm{p_0} + \sum_{i\ge 1} \eta^i \norm{p_i} } = \frac{1+\eta}{1-\eta}\mper \\
        \norm{p}_y &\ge (1-\eta)\norm{p_0} - (1+\eta)\sum_{i\ge 1} \eta^i \norm{p_i} \ge \frac{1-3\eta}{1-\eta}
    \end{align*}

    \paragraph{Union bound over the net.}
    We have left to prove that $\frac{\norm{p}_{\rvy}}{\norm{p}} \in [0.9,1.1]$ for all $p \in\mathcal{C}$ with probability at least $1-(1+\delta) \frac{2^{O(\ell)}} n$.
    We prove this via the second moment method and a union bound.
    In particular, since $\calC$ has size at most $2^{O(\ell)}$ and $\ell = c\log n$ for some sufficiently small constant $c$, it suffices to prove that for a fixed $p \in \calC$
    \[
        \Prob\bracks*{\frac{\norm{p}_{\rvy}}{\norm{p}} \not\in [0.9,1.1]} \leq \frac{(1+\delta) 2^{O(\ell)}}{n}\mper
    \]
    Note that this is implied by showing that $\lvert \norm{p}_{\rvy}^2 - \norm{p}^2\rvert \geq (0.1)^2 \norm{p}^2$ with at most the same probability.
    We will show that
    $$
    \E_{\rvy\sim\Planted}\bracks*{\parens*{\norm{p}_{\rvy}^2 - \norm{p}^2}^2 } \leq \frac{(1+\delta) \cdot 2^{O(\ell)}}{n}\cdot\norm{p}^4 \mcom
    $$
    which implies the claim via Markov's Inequality.
    Note that:
    \begin{align*}
        \E_{\rvy\sim\Planted}\bracks*{\parens*{\norm{p}_{\rvy}^2 - \norm{p}^2}^2 }
        &\le \E_{\rvg\sim\Null} \bracks*{ \parens*{ \norm{p}_{\rvg}^2 - \norm{p}^2 }^2 } + \delta\sqrt{ \E_{\rvg\sim\Null}\bracks*{ \parens*{ \norm{p}_{\rvg}^2 - \norm{p}^2 }^4 } } \\
        &\le \parens*{1+\delta\cdot 2^{O(\ell)}} \cdot \E_{\rvg\sim\Null} \bracks*{ \parens*{ \norm{p}_{\rvg}^2 - \norm{p}^2 }^2 } \\
        &= \parens*{1+\delta\cdot 2^{O(\ell)}} \cdot
        \frac{1}{n} \Var_{\rvg\sim\calN(0,1)}\left[p(\rvg)^2\right] \\
        &\le \frac{(1+\delta)\cdot2^{O(\ell)}}{n}\cdot\E_{\rvg\sim\calN(0,1)}\bracks*{ p(\rvg)^2 }^2 = \frac{(1+\delta) \cdot 2^{O(\ell)}}{n}\cdot\norm{p}^4\mcom
    \end{align*}
    where in the first line we used a bound on $\chisquaretrunc{\Planted}{\Null}$, in the second line we use $2$-to-$4$ hypercontractivity for polynomials over Gaussian space (\Cref{fact:gauss_hypercontractivity}), in the third line we expanded out the variance, and in the fourth line we once again used hypercontractivity.
\end{proof}

\paragraph{Gaussian moment bounds.}
We next give the proof of \cref{lem:moment_bounds}, restated below.
\restatelemma{lem:moment_bounds}

\begin{proof}
We will first derive the second statement from the first (which, overloading notation, we will prove up to $2T$).
Recall that $\NullShadow, \PlantedShadow$ are the laws of $F_k(\rvx)$ with $\rvx \sim \Null$ and $\rvx\sim \OU_\e \Planted$ respectively.
Note that $\Iprod{\xi, F_k(x)}^T$ is a polynomial in $x$ of degree at most $kT \leq D$.
Since $\chisquaretrunc{\Planted}{\Null} \leq \delta$ (and hence $\chisquaretrunc{\OU_\e \Planted}{\Null} \leq \delta$), it follows that
\begin{align*}
    \E_{\rvx\sim \OU_\e \Planted} \Brac{\Iprod{\xi, F_k(\rvy)}^T} &\leq \E_{\rvx\sim \Null} \Brac{\Iprod{\xi, F_k(\rvx)}^T} + \delta \sqrt{\Var_{\rvx\sim\Null}\Brac{\Iprod{\xi, F_k(\rvx)}^T}} \\
    &\leq \E_{\rvx\sim\Null} \Brac{\Iprod{\xi, F_k(\rvx)}^T} + \delta \sqrt{\E_{\rvx\sim\Null}\Iprod{\xi, F_k(\rvx)}^{2T}} \\
    &\leq (1+ \delta) \cdot O(T)^{T/2} \,.
\end{align*}

\paragraph{The Gaussian case: $\Null = \calN(0,\I_n)$.}
We next prove the first statement (note that here we will only use the condition $(CT)^k \leq n$ and hence changing $T$ by a factor of 2 does not matter).
Fix a unit vector $\xi \in \R^k$ and define $p_\xi \colon \R \rightarrow \R$ as $p_\xi(x) = \sum_{i=1}^k \xi_i h_i(x)$.
Note that $p_\xi$ is a polynomial of degree at most $k$ and satisfies $\E_{\rvg\sim\calN(0,1)}[p_\xi(\rvg)] = 0$ and $\E_{\rvg\sim\calN(0,1)}[p_\xi(\rvg)^2] = 1$.
Further, it holds that
$$
\iprod{\xi, F_k(x)} = \sum_{i=1}^k \xi_i \tilde{h}_i(x) = \tfrac 1 {\sqrt{n}} \sum_{j=1}^n \sum_{i=1}^k \xi_i h_i(x_j) = \tfrac 1 {\sqrt{n}} \sum_{j=1}^n p_\xi(x_j) \,.
$$
For simplicity, we will drop the dependence on $\xi$ in $p_\xi$.
We need to bound
$$
n^{-T/2} \E_{\rvg}\left[\left(\sum_{j=1}^n p(\rvg_j) \right)^T \right].
$$
We deal with this expectation by expanding out the power in the expectation.
We set up some notation first.
Denote by $\alpha$ a multi-index in $\N^n$ and by $\abs{\alpha} = \sum_{j=1}^n \alpha_j$ and by $\norm{\alpha}_0$ the number of non-zero entries of $\alpha$.
Then
$$
\E_{\rvg}\left[\left(\sum_{j=1}^n p(\rvg_j) \right)^T \right] = \sum_{\abs{\alpha} = T} \prod_{j=1}^n \E_{\rvg_j} \Brac{p(\rvg_j)^{\alpha_j}} \,.
$$
Note that the product is 0 for all $\alpha$ for which one of $\alpha_j = 1$.
Further, when $\alpha_j = 2$, the corresponding term is exactly 1.
Define the set
\[
    A_{m,T} = \set{\alpha \in \N^n \mid \forall j: \alpha_j \neq 1 \,, \abs{\alpha} = T \,, \norm{\alpha}_0 = m} \,.
\]
Further, by Gaussian hypercontractivity (cf. \cref{fact:gauss_hypercontractivity}) we know that for any integer $\ell$, $\E_{\rvg}[p(\rvg)^\ell] \leq \ell^{k \ell /2}$.
Thus, applying hypercontractivity for any $\alpha_j > 2$, it follows that 
$$
\E\left[\left(\sum_{j=1}^n p(\rvg_j) \right)^T \right] \leq \sum_{m=1}^{T/2} \sum_{\alpha \in A_{m,T}} \prod_{j=1}^n O(\alpha_j)^{k \alpha_j /2} = \sum_{m=1}^{T/2} \sum_{\alpha \in A_{m,T}} \prod_{j \colon \alpha_j > 2 } O(\alpha_j)^{k \alpha_j /2}\,.
$$
We claim that the product is at most $(O(T)^k)^{T/2 - m}$.
Deferring the proof to the end of the section, we first see how this yields the final bound.
Note that $\abs{A_{m,T}} \leq \binom{n}{m}\cdot  m^T \leq (n/m)^m m^T$.
Thus, using that $m \leq T$
\begin{align*}
    n^{-T/2} \E\left[\left(\sum_{j=1}^n p(\rvg_j) \right)^T \right] &\leq \sum_{m=1}^{T/2} \Paren{\frac{n}{m}}^m m^T n^{-T/2} \Paren{O(T)^k}^{T/2 - m} \\
    &\leq T^{T/2} \sum_{m=1}^{T/2} \Paren{\frac m n}^{T/2 - m} \Paren{O(T)^k}^{T/2 - m} \\
    &=  T^{T/2} \sum_{m=1}^{T/2} \Paren{\frac {m \cdot \Paren{O(T)^k} } n}^{T/2 - m} \,.
\end{align*}
Since $m \leq T/2$ and $(CT)^k \leq n$ as for a large constant $C$, it follows that each term in the sum is at most $1$ and hence the overall expression is at most $O(T)^{T/2}$.

To argue that the product is at most $(O(T)^k)^{T/2 -m}$, we will show that it is maximized when all but one of the non-zero coefficients of $\alpha$ are equal to 2 and the remaining one is equal to $T - 2m$.
Note that in this case, this is value of the product.
Let us pretend the $\alpha_j$ are real-valued but such that all are at least 2 and they still sum to $T$ (ignoring all the entries that are 0).
Note that this is a polytope with corners corresponding to entries in which all but one coordinates are 2 and the remaining one is $T-2m$.
Taking logarithms in the product, we observe that the function $\tfrac k 2 \sum_{j} \alpha_j \log(\alpha_j)$ is convex and thus maximized at one of the corners of this polytope.
This implies the claim.
\end{proof}

\section{Gaussian Space II: Connected Subgraph Counts}
\label{sec:subgraph-stats}
In this section we will prove our result of TV-closeness of subgraph counts for matrix-valued inputs.

\paragraph{Notation for this section.}
We will use the following notation throughout this section.
$n$ is the dimension of the problem, the null distribution $\Null$ is given by a random symmetric $n\times n$ matrix with independent Gaussian entries, $\Planted$ is a distribution that is $\delta$-low-degree indistinguishable from $\Null$, i.e., $\chisquaretrunc{\Planted}{\Null} \le \delta$.
We use $\eps$ to denote the amount of noise we add to the planted distribution.
Without loss of generality, we can take $\eps < 1/2$.
In particular, $\OU_\e \Planted$ denotes the distribution of $\sqrt{1-\e} \rvY + \e \rvG$ where $\rvY \sim \Planted, \rvG \sim \Null$ are indpendent.

We use $\sgraph$ to refer to a constant-sized connected graph.
We use $v(\sgraph)$ to refer to the number of vertices in $\sgraph$, and $e(\sgraph)$ to refer to the number of edges in $\sgraph$.
$E(\sgraph)$ refers to the set of all edges in $\sgraph$.
Let $\calL_{\sgraph}$ refer to the set of all \emph{distinct} injective labelings $\pi$ from vertices of $\sgraph$ to $[n]$ where we consider two labelings $\pi,\pi'$ identical if there is an automorphism of $\sgraph$ that transforms $\pi$ to $\pi'$.
For an edge $e = (u,v)$, we denote $\pi(e) = \set{\pi(u),\pi(v)}$.
For a matrix $M$, define
$$
    \chi_\sgraph(M) \coloneqq
    \frac{1}{ \sqrt{ \abs{\calL_{\sgraph}} } }
    \sum_{\pi\in\calL_{\sgraph}} \prod_{ab\in E(\sgraph)} M_{\pi(a),\pi(b)}\mper
$$

In this section, we will prove the following.
\begin{theorem} \label{thm:main-subgraph}
    Let $\delta \geq 0,  D \ge \log n \cdot \log \log n$ and $\Planted, \Null$ be such that $\chisquaretrunc{\Planted}{\Null} \leq \delta$.
    Let $\pi_{\sgraph}$ be the law of $\chi_{\sgraph}\parens*{\rvM}$ for $\rvM\sim\OU_{\eps} \Planted$, and let $\nu_{\sgraph}$ be the law of $\chi_{\sgraph}\parens*{\rvM}$ for $\rvM\sim \Null$.
    Then, for a sufficiently small constant $\alpha > 0$, we have:
    \[
        \dtv\parens*{ \nu_{\sgraph}, \pi_{\sgraph} } \le O\parens*{ \delta^{\alpha \eps^2} + n^{-\alpha\eps^2} + \frac{\polylog n}{\sqrt{n}}}\mcom
    \]
    where the $O\parens*{\cdot}$ hides constants depending on $\sgraph$.
\end{theorem}

In service of proving \Cref{thm:main-subgraph}, we use the following result of Chatterjee \cite[Theorem 2.2]{Cha09} which allows us to pretend that adding noise to $\rvM$ ``behaves like'' adding noise to $\chi(\rvM)$, i.e., after we apply the subgraph polynomial.
\begin{lemma}   \label{lem:shivam}
    Let $F:\R^d\to\R$ be a twice-differentiable function.
    For the statistics:
    \begin{align*}
        \kappa_1 &\coloneqq \parens*{\E_{\rvg\sim\calN(0,I_d)}\norm{\grad F\parens*{\rvg}}_2^4}^{1/4}
        \mcom
        \quad
        \kappa_2 \coloneqq \parens*{\E_{\rvg\sim\calN(0,I_d)}\norm{\grad^2 F\parens*{\rvg}}_{\op}^4}^{1/4}
        \mcom\\
        \mu &\coloneqq \E_{\rvg\sim\calN(0,I_d)} F(\rvg)
        \mcom
        \quad
        \sigma^2 \coloneqq \Var_{\rvg\sim\calN(0,I_d)} F(\rvg)\mcom
    \end{align*}
    we have
    \[
        \dtv\parens*{ F(\rvg), \calN(\mu, \sigma^2) } \le \frac{2\sqrt{5} \kappa_1 \kappa_2}{\sigma^2}\mper
    \]
\end{lemma}

\parhead{Overall strategy for proving \Cref{thm:main-subgraph}.}
Let $\rvG\sim\Null$.
For $\rvM\sim\Planted$, let $\rvM_{\eps} \coloneqq \sqrt{1-\eps}\rvM + \sqrt{\eps}\rvG$.
Our overall proof strategy is the same as in \cref{sec:tv_closeness}.
The main difference is in the first step whose goal is to pass to a pair of distributions for which showing bounds on higher-order Fourier coefficients is easier. 
\begin{itemize}
    \item We first use \Cref{lem:shivam} to show that the law of $\chi_{\sgraph}\parens*{\rvM_{\eps}}$ is close in total variation distance to that of $\sqrt{1-\eps}^{e(\sgraph)}\chi_{\sgraph}\parens*{\rvM} + \sigma_{\rvM}\cdot\rvg$ where $\rvg\sim\calN(0,1)$ is independent of $\rvM$, and $\sigma_{\rvM}$ has ``nontrivial'' magnitude \emph{everywhere}.
    More specifically, denoting the second distribution as $\ul{\pi}_{\sgraph}$, we will prove that $\dtv\parens*{ \pi_{\sgraph}, \ul{\pi}_{\sgraph} } \lesssim \polylog n/\sqrt{n}$.
    We will similarly define $\ul{\nu}_{\sgraph}$ and prove a similar bound on the total variation distance to $\nu_\sgraph$.
    Thus, in the rest of the argument it is sufficient to bound the total variation distance between $\ul{\pi}$ and $\ul{\nu}$.
    This corresponds to \cref{lem:noisy-planted} below.
    \item To bound $\dtv\parens*{\ul{\pi}_{\sgraph},\ul{\nu}_{\sgraph}}$, we bound $\abs{\ul{\pi}_{\sgraph}(x) - \ul{\nu}_{\sgraph}(x)}$ for $x$ in some high-probability region $[-\tau,\tau]$:
    In particular, if we prove a bound of $\eta$ in that region, we get:
    \[
        \dtv\parens*{\ul{\pi}_{\sgraph}, \ul{\nu}_{\sgraph}} \le \Pr_{\bx\sim\nu}\bracks*{|\bx|\ge \tau} + 2\eta\tau\mper
    \]
    We then show that there is a trade-off between $\eta$ and $\tau$ such that this bound becomes small enough.
    \item Let us fix $x$ in this high-probability region. To control $\abs{\ul{\pi}_{\sgraph}(x) - \ul{\nu}_{\sgraph}(x)}$, we take a Fourier-analytic approach. 
    Specifically, we use the identity:
    \[
        \abs*{\ul{\pi}_{\sgraph}(x) - \ul{\nu}_{\sgraph}(x)} = \tfrac 1 {\sqrt{2\pi}}\abs*{\int_{-\infty}^{\infty}\parens*{\wh{\ul{\pi}}_{\sgraph}(\xi) - \wh{\ul{\nu}}_{\sgraph}(\xi)}\cdot\exp(-i\xi x) }\dif\xi \mper 
    \]
    We split this integral into two parts:
    For large $\xi$ (on the order of $\xi \ge T = C\cdot \log n$), we will prove that $\wh{\ul{\pi}}_{\sgraph}(\xi)$ and $\wh{\ul{\nu}}_{\sgraph}(\xi)$ are individually at most $1/\poly(n)$.
    \item For $\xi\in[-T,T]$, we will show that the Fourier coefficients of both distributions are close: Using the fact that $\wh{\ul{\pi}}_{\sgraph}(\xi) = \wh{\pi}_{\sgraph}(\xi) \pm \dtv\parens*{ \ul{\pi}_{\sgraph}, \pi_{\sgraph} }$ and $\wh{\ul{\nu}}_{\sgraph}(\xi) = \wh{\nu}_{\sgraph}(\xi) \pm \dtv\parens*{ \ul{\nu}_{\sgraph}, \nu_{\sgraph} }$ (where the TV distance are small as discussed above), it is enough to control $\abs{\wh{\pi}_{\sgraph}(\xi) - \wh{\nu}_{\sgraph}(\xi)}$.
    This follows via moment-matching ($\chisquaretrunc{\Planted}{\Null}$ being small). 
    See \cref{lem:infty-norm-bound} below for the precise statement.
\end{itemize}

In the remainder of this subsection, we will state \cref{lem:noisy-planted,lem:infty-norm-bound} as well as a lemma that shows that the moments of $\chi_{\sgraph}\parens*{\rvM}$ for $\rvM \sim \Null$ are sub-gaussian (cf.~\cref{lem:subgraph-moments}).
We will then show how these imply \cref{thm:main-subgraph}.
We will prove \cref{lem:noisy-planted} in \cref{sec:noisy-subgraph} (this constitute the bulk of the proof), \cref{lem:infty-norm-bound} in \cref{sec:infty-norm-bound}, and \cref{lem:subgraph-moments} in \cref{sec:subgraph-moments}.

\parhead{Passing to noisy subgraph counts and moment bounds.}
We start with the following definition.

\begin{definition}[Noisy subgraph counts]
    \label{def:noisy-subgraph-counts}
    For a fixed matrix $M \in \R^{n \times n}$ and $\rvG$ a symmetric $n \times n$ matrix with i.i.d. $\calN(0,1)$ entries (up to symmetry), define
    \[
        \wt{\sigma}^2(M) \coloneqq \max\braces*{ \Var_{\rvG\sim\Null} \chi_{\sgraph}\parens*{ \sqrt{1-\eps}M + \sqrt{\eps}\rvG }, \frac{\eps}{2}\parens*{1-\eps}^{e(\sgraph)-1}   }\mper
    \]
    Let $\rvM \sim \Planted, \rvg \sim \calN(0,1)$ be independent.
    We denote by $\ul{\pi}_{\sgraph}$ the law of $\sqrt{1-\eps}^{e(\sgraph)} \chi_{\sgraph}(\rvM) + \wt{\sigma}\parens*{\rvM}\cdot\rvg$.
    Similarly, when $\rvM \sim \Null$ instead, we denote the resulting law by $\ul{\nu}_{\sgraph}$.
\end{definition}

The following lemma bounds the total variation distance between $\pi_{\sgraph}$ and $\ul{\pi}_{\sgraph}$, and $\nu_{\sgraph}$ and $\ul{\nu}_{\sgraph}$ respectively.
\begin{lemma}  
\torestate{
\label{lem:noisy-planted}
    Assume that $D\ge\log n\cdot \log\log n$ and let $\Planted$ be such that $\chisquaretrunc{\Planted}{\Null} \le n$,
    then we have:
    
    \begin{align*}
        \dtv\parens*{ \pi_{\sgraph}, \ul{\pi}_{\sgraph} } \le O\parens*{\frac{\sqrt{\eps}\cdot\polylog n}{\sqrt{n}} + \frac{1}{n^2}}\quad\text{and}\quad\dtv\parens*{ \nu_{\sgraph}, \ul{\nu}_{\sgraph} } \le O\parens*{\frac{\sqrt{\eps}\cdot\polylog n}{\sqrt{n}} + \frac{1}{n^2}} \mper
    \end{align*}
}
\end{lemma}
Note that the second part follows immediately by the first by setting $\Planted = \Null$.

We also use the following lemma showing that polynomially many moments of $\chi_\vartheta(\rvM)$ are sub-gaussian when $\rvM \sim \Null$.
\begin{lemma}  
\torestate{
\label{lem:subgraph-moments}
    There exists a sufficiently small constant $\alpha=\alpha(\sgraph)$ such that for any $0\le q \le n^{\alpha}$:
    \begin{align*}
        \E_{\rvM\sim\Null} \chi_{\sgraph}\parens*{\rvM}^{q} \le \sqrt{q}^{q}\cdot\parens*{1+\frac{1}{\poly(n)}}\mper
    \end{align*}
}
\end{lemma}

\paragraph{$\ell_\infty$-closeness and putting things together.}

We first show the following bound on the infinity distance of $\ul{\pi}_{\sgraph}$ and $\ul{\nu}_{\sgraph}$.\footnote{Note that the restriction on $\e$ is not a big restriction since \cref{thm:main-subgraph} is vacuous for $\e = o(1/\sqrt{\log n})$.}
\begin{lemma}
    \torestate{
    \label{lem:infty-norm-bound}
    Let $D \geq \log n \log \log n$ and $\delta \geq 0$ be such that $\chisquaretrunc{\Planted}{\Null} \leq \delta$ and assume that $\e \geq 1/ \polylog n$.
    Then there is an absolute constant $c > 0$ such that
    \[
        \norm{ \ul{\pi}_{\sgraph} - \ul{\nu}_{\sgraph} }_{\infty} \le O\parens{\delta^{c\e^2} + n^{-c\e^2}} \mper
    \]
    }
\end{lemma}

We are now finally ready to prove \Cref{thm:main-subgraph}.

\begin{proof}[Proof of \Cref{thm:main-subgraph}]
    First, note that when $\e = o(1/\sqrt{\log n})$ or $\e = o(1/\sqrt{\log(1/\delta)})$ the right-hand side is at least 1, so the bound trivially holds.
    We can thus assume this is not the case.
    
    By triangle inequality and \Cref{lem:noisy-planted} it holds that:
    \[
        \dtv\parens*{ \pi_{\sgraph}, \nu_{\sgraph} } \le \dtv\parens*{ \ul{\pi}_{\sgraph}, \ul{\nu}_{\sgraph} } + O\parens*{\frac{\sqrt{\e}\polylog n}{\sqrt{n}} + \frac 1 {n^2}} \leq \dtv\parens*{ \ul{\pi}_{\sgraph}, \ul{\nu}_{\sgraph} }  + O\parens*{\frac {\polylog n} n }\mper
    \]
    So it is enough to show that $\dtv\parens*{ \ul{\pi}_{\sgraph}, \ul{\nu}_{\sgraph} } \leq O(\delta^{\alpha \e^2} + n^{-\alpha \e^2})$ for a small absolute constant $\alpha$.
    Let $\beta = \delta^{c\e^2} + n^{-c\e^2}$ and $\tau = \sqrt{\log(1/\beta)}$, where $c$ is the constant in the bound of \cref{lem:infty-norm-bound}.
    Note that since $\e \geq \Omega(1/\sqrt{\log n})$ and $\e \geq \Omega(1/\sqrt{\log (1/\delta})$, $\beta$ is at most (say) 0.001.
    Then, 
    \begin{align*}
        \dtv\parens*{\ul{\pi}_{\sgraph}, \ul{\nu}_{\sgraph}} &= \int_{-\infty}^{\infty} \abs{\ul{\pi}_{\sgraph}(x) - \ul{\nu}_{\sgraph}(x)} \cdot \Ind[\ul{\nu}_{\sgraph}(x) \ge \ul{\pi}_{\sgraph}(x) ] \dif x \\
        &= \int_{-\tau}^{\tau} \abs{\ul{\pi}_{\sgraph}(x) - \ul{\nu}_{\sgraph}(x)}\dif x + \int_{x\notin[-\tau,\tau]} \parens*{\ul{\nu}_{\sgraph}(x) -  \ul{\pi}_{\sgraph}(x)} \dif x \\
        &\leq 2\tau \cdot \norm{ \ul{\pi}_{\sgraph} - \ul{\nu}_{\sgraph} }_{\infty}  + \Pr_{\bx\sim\ul{\nu}_{\sgraph}} \bracks*{ |\bx| \ge \tau }\mper  
    \end{align*}

    Note that since $\e \leq 1$, $\tau \leq O(\sqrt{\log n})$, so by \cref{lem:subgraph-moments} (with $q = 0.01\tau^2$) and Markov's Inequality, it follows that
    $$
        \Pr_{\bx\sim\ul{\nu}_{\sgraph}} \bracks*{ |\bx| \ge \tau } \leq \parens*{\frac{2q}{\tau^2}}^{q/2} \leq e^{-\Omega(\tau^2)} = O(\beta) \mper
    $$
    Combining this with \cref{lem:infty-norm-bound} it follows that 
    $$
        \dtv\parens*{\ul{\pi}_{\sgraph}, \ul{\nu}_{\sgraph}} \leq O\parens*{\beta \sqrt{\log(1/\beta)}} \leq O\parens*{\sqrt{\beta}}  = O(\delta^{\alpha \e^2 / 2} + n^{-\alpha \e^2/2}) \mper  \qedhere
    $$
\end{proof}

\subsection{Subgraph Counts of Noisy Inputs Are Close to Noisy Subgraph Counts}
\label{sec:noisy-subgraph}

In this section we prove \cref{lem:noisy-planted} restated below.
Recall that $\ul{\pi}_\sgraph$ (respectively $\ul{\nu}_\sgraph$), corresponds to the law of $\sqrt{1-\eps}^{e(\sgraph)} \chi_{\sgraph}(\rvM) + \wt{\sigma}\parens*{\rvM}\cdot\rvg$, where $\rvg \sim \calN(0,1)$ and $\rvM \sim \Planted$ ($\rvM \sim \Null$ respectively), and for a fixed matrix $M \in \R^{n\times n}$
\[
        \wt{\sigma}^2(M) \coloneqq \max\braces*{ \Var_{\rvG\sim\Null} \chi_{\sgraph}\parens*{ \sqrt{1-\eps}M + \sqrt{\eps}\rvG }, \frac{\eps}{2}\parens*{1-\eps}^{e(\sgraph)-1}   }\mcom
\]
where $\rvG$ is a symmetric $n \times n$ matrix with i.i.d. $\calN(0,1)$ entries (up to symmetry).

\restatelemma{lem:noisy-planted}

\begin{proof}
    For an $n\times n$ matrix $M$, we will apply \Cref{lem:shivam} with $F(\rvG)\coloneqq \chi_{\sgraph}\parens*{\sqrt{1-\eps}M+\sqrt{\eps}\rvG}$.\footnote{To make this fully formal, because of the symmetry in the matrices, one would have to work with a vector that only contains the upper triangular part. This can be easily done and we will omit it for clarity.}

    Following notation from \Cref{lem:shivam}, define:
    \begin{align*}
        \kappa_1(M) &\coloneqq \parens*{  \E_{\rvG\sim\Null} \norm{\grad_{\rvG}\, \chi_{\sgraph}\parens*{\sqrt{1-\eps}M+\sqrt{\eps}\rvG} }_2^4 }^{1/4}\mcom \\
        \kappa_2(M) &\coloneqq \parens*{  \E_{\rvG\sim\Null} \norm{\grad^2_{\rvG}\, \chi_{\sgraph}\parens*{\sqrt{1-\eps}M+\sqrt{\eps}\rvG} }_{\op}^4 }^{1/4}\mcom \\
        \sigma^2(M) &\coloneqq \Var_{\rvG\sim\Null} \chi_{\sgraph}\parens*{\sqrt{1-\eps}M + \sqrt{\eps}\rvG}\mper
    \end{align*}
    It suffices for us to prove high-probability upper bounds on $\kappa_1(\rvM)$ and $\kappa_2(\rvM)$, and a high-probability lower bound on $\sigma^2(\rvM)$ for $\rvM\sim\Planted$.
    For a random variable $\bX$, we will use $\ell_q(\bX)$ to refer to $\parens*{\E \bX^q}^{1/q}$.%
    \footnote{The reason for this choice of notation is to avoid confusing notation like $\norm{\norm{\bX}_2}_q$ where the outer norm is a random variable norm and the inner norm is a vector norm.}
    We will make repeated use of the fact that $\ell_q\parens*{\cdot}$ is a norm.

    \parhead{High-probability upper bound on $\kappa_1(\rvM)$.}
    For $\rvM\sim\Planted$, we will prove a bound on $\ell_q\parens*{\kappa_1(\rvM)}$ for an appropriate choice of $q\ll D$, and an adequate high-probability bound will follow from Markov's inequality.
    \begin{align*}  \label{eq:ainesh-bakshi}
        \ell_q\parens*{ \kappa_1(\rvM) } = \ell_q \parens*{
            \ell_4\parens*{
                \norm{
                    \grad_{\rvG}\,\chi_{\sgraph}\parens*{\sqrt{1-\eps}\rvM + \sqrt{\eps}\rvG}
                }_2 \Big| \rvM
            }
        }   \numberthis
    \end{align*}
    Let us disentangle the above term.
    To keep the terms we will be propagating through our calculations light, let us define some more notation.
    For $e\in E(\sgraph)$ and $i,j\in[n]$, we will define $\calL_{\sgraph,e,\{i,j\}}$ as the set of all $\pi\in\calL_{\sgraph}$ such that $\pi(e)=\{i,j\}$.
    Define
    \[
        \chi_{\sgraph,e}^{i,j}(M) \coloneqq \frac{1}{\sqrt{\abs{\calL_{\sgraph}}}} \sum_{\pi\in\calL_{\sgraph,e,\{i,j\}}} \prod_{ab\in E(\sgraph)\setminus e} M_{\pi(a),\pi(b)}\mcom
    \]
    and $u_{\sgraph,e}(M)$ as the ${n\choose 2}$-dimensional vector where $u_{\sgraph,e}(M)[\{i,j\}] = \chi_{\sgraph,e}^{i,j}(M)$.
    Now, observe that:
    \[
        \grad_{\rvG}\, \chi_{\sgraph}\parens*{ \sqrt{1-\eps} \rvM + \sqrt{\eps} \rvG } =
        \sum_{e\in E(\sgraph)} \sqrt{\eps} \cdot u_{\sgraph, e}\parens*{ \sqrt{1-\eps} \rvM + \sqrt{\eps} \rvG }
    \]
    Plugging in the above into \Cref{eq:ainesh-bakshi}, we get:
    \begin{align*}
        \ell_q\parens*{\kappa_1\parens*{\rvM}}
        &\le \sqrt{\eps} \sum_{e\in E(\sgraph)}
        \ell_q\parens*{
            \ell_4\parens*{
                \norm{ u_{\sgraph, e}\parens*{
                    \sqrt{1-\eps}\rvM + \sqrt{\eps}\rvG
                } \Big| \rvM}_2
            }
        } \\
        &\le \sqrt{\eps} \sum_{e\in E(\sgraph)}
        \ell_{4q}\parens*{
            \norm{ u_{\sgraph, e}\parens*{
                    \sqrt{1-\eps}\rvM + \sqrt{\eps}\rvG
                }
            }_2
        }   \numberthis \label{eq:moth}
    \end{align*}
    where the first inequality follows from applying the triangle inequality for norms, and the second inequality follows from Jensen's inequality.

    We will now analyze a single summand.
    Defining $\rvM_{\eps} \coloneqq \sqrt{1-\eps}\rvM + \sqrt{\eps}\rvG$, we have:
    \begin{align*}
        \ell_{4q}\parens*{
            \norm{ u_{\sgraph, e}\parens*{ \rvM_{\eps} } }_2
        }^{4q} &= \E_{\rvM_{\eps}\sim\OU_{\eps}\Planted}
        \norm{ u_{\sgraph, e}\parens*{ \rvM_{\eps} } }_2^{4q} \mper
    \end{align*}
    Since this is a polynomial of degree $C_\vartheta q \ll D$ in $\rvM_\e$, where $C_\vartheta$ only depends on $\vartheta$, it follows using that $\Var X \leq \E X^2$ and $\E X \leq \sqrt{\E X^2}$ that 
    \begin{align}
        \nonumber \ell_{4q}\parens*{
            \norm{ u_{\sgraph, e}\parens*{ \rvM_{\eps} } }_2
        }^{4q} &= \E_{\rvM_{\eps}\sim\OU_{\eps}\Planted}
        \norm{ u_{\sgraph, e}\parens*{ \rvM_{\eps} } }_2^{4q} \leq \E_{\rvG\sim\Null}
        \norm{ u_{\sgraph, e}\parens*{ \rvG } }_2^{4q} + \sqrt{
                \chisquaretrunc{\Planted}{\Null} \cdot \Var_{\rvG\sim\Null}
                \norm{ u_{\sgraph, e}\parens*{ \rvG } }_2^{4q}
            } \\
        &\leq (1+\sqrt{n}) \cdot \sqrt{\E_{\rvG\sim\Null}
        \norm{ u_{\sgraph, e}\parens*{ \rvM_{\eps} } }_2^{8q}} \mper \label{eq:ldlr-trace}
    \end{align}

    In summary, we have proved:
    \begin{align*}
        \ell_{4q}\parens*{ \norm{u_{\sgraph,e}\parens*{\rvM_{\eps}}}_2 } \le O\parens*{n^{1/8q}} \cdot
        \parens*{
            \E_{\rvG\sim\Null}
        \norm{ u_{\sgraph, e}\parens*{ \rvM_{\eps} } }_2^{8q}
        }^{1/8q}
        \numberthis \label{eq:ell4q-garbage}
    \end{align*}
    Now, note that $\norm{ u_{\sgraph, e}\parens*{ \rvM_{\eps} } }_2^{8q} = \Tr\parens*{ \parens*{ u_{\sgraph,e}\parens*{\rvG} u_{\sgraph,e}\parens*{\rvG}^{\top} }^{4q}}$.
    By \cite[Lemmas 8.7 and 8.15]{AMP16} (``graph matrix moment bounds''):\footnote{In our setting our matrix is of size $\binom{n}{2} \times 1$. In graph matrix terms, there is only one set of distinguished vertices, corresponding to the indices for the entries of the vector.}
    \begin{align*}
        \E_{\rvG\sim\Null} \Tr\parens*{ \parens*{ u_{\sgraph,e}\parens*{\rvG} u_{\sgraph,e}\parens*{\rvG}^{\top} }^{4q}}
        &\le
        \frac{1}{\abs{ \calL_{\sgraph} }^{4q}} (8q)^{8q(e(\sgraph)+v(\sgraph))} \cdot n^{4q v(\sgraph)} \\
        &\le \frac{v(\sgraph)^{4qv(\sgraph)}}{ n^{4q v(\sgraph)} } (8q)^{8q(e(\sgraph)+v(\sgraph))} \cdot n^{4q v(\sgraph)} \\
        &\le (8v(\sgraph)q)^{8q(e(\sgraph)+v(\sgraph))}\mper
    \end{align*}
    Plugging the above into \Cref{eq:ell4q-garbage} gives:
    \[
         \ell_{4q}\parens*{ \norm{u_{\sgraph,e}\parens*{\rvM_{\eps}}}_2 } \le
          O\parens*{n^{1/8q}} \cdot
          (8v(\sgraph)q)^{e(\sgraph)+v(\sgraph)}\mper
    \]
    We explicitly choose $q = \log n$, and plug the above into \Cref{eq:moth} to obtain the bound:
    \begin{align*}
        \ell_{\log n}\parens*{ \kappa_1(\rvM) } \le O(\sqrt{\eps}) \cdot e(\sgraph) \cdot (8v(\sgraph)\log n)^{e(\sgraph)+v(\sgraph)}\mper   \numberthis \label{eq:final-gradgrad}
    \end{align*}

    \parhead{High-probability upper bound on $\kappa_2(\rvM)$.}
    Akin to the previous case, will bound $\ell_q\parens*{\kappa_2(\rvM)}$ for some large $q$, and an adequate high-probability bound will follow from Markov's inequality.
    By carrying out a similar calculation to the previous case up until \Cref{eq:moth}, we get:
    \begin{align*}
        \ell_q\parens*{ \kappa_2(\rvM) }
        \le
        \eps \cdot \sum_{e\ne e'\in E(\sgraph)}
        \ell_{4q}\parens*{
            \norm{ u_{\sgraph, e, e'}\parens*{ \rvM_{\eps} } }_{\op}
        }\mper   \numberthis \label{eq:slitherwing}
    \end{align*}
    In the above, $u_{\sgraph,e,e'}(\rvM)$ is defined as the following ${n \choose 2} \times {n \choose 2}$ matrix that arises from computing $\grad^2 \chi_{\sgraph}\parens*{\rvM_{\eps}}$
    \[
        u_{\sgraph,e,e'}(M)[\{i,j\}, \{i', j'\}] =
        \frac{1}{\sqrt{\abs{\calL_{\sgraph}}}}
        \sum_{\pi\in\calL_{\sgraph,e,\{i,j\}} \cap \calL_{\sgraph,e',\{i',j'\}} }
        \prod_{ab\in E(\sgraph)\setminus\{e,e'\} } M_{\pi(a),\pi(b)}\mper
    \]
    We now study a single summand of \Cref{eq:slitherwing}.
    For a fixed $e\ne e'$, we have:
    \begin{align*}
        \ell_{4q}\parens*{
            \norm{ u_{\sgraph, e, e'}\parens*{ \rvM_{\eps} } }_{\op}
        }^{4q}
        &\le
        \E_{\rvM_{\eps}\sim\OU_{\eps}\Planted} \norm{ u_{\sgraph, e, e'}\parens*{ \rvM_{\eps} } }_{\op}^{4q} \\
        &\le \E_{\rvM_{\eps}\sim\OU_{\eps}\Planted} \Tr\parens*{
            \parens*{u_{\sgraph, e, e'}\parens*{ \rvM_{\eps} }u_{\sgraph, e, e'}\parens*{ \rvM_{\eps} }^{\top}
            }^{2q}
        } \\
        &\le O(n) \cdot \sqrt{ \E_{\rvG\sim\Null} \Tr\parens*{
            \parens*{u_{\sgraph, e, e'}\parens*{ \rvG }u_{\sgraph, e, e'}\parens*{ \rvG }^{\top}
            }^{4q} }
        },
    \end{align*}
    where the final inequality is derived along analogous lines to \Cref{eq:ldlr-trace}.
    The above establishes:
    \begin{align*}
        \ell_{4q}\parens*{ \norm{ u_{\sgraph,e,e'}\parens*{\rvM_{\eps}} } } \le O\parens*{ n^{1/4q} } \cdot \parens*{
            \E_{\rvG\sim\Null} \Tr\parens*{
                \parens*{u_{\sgraph, e, e'}\parens*{ \rvG }u_{\sgraph, e, e'}\parens*{ \rvG }^{\top}
                }^{4q}
            }
        }^{1/8q}    \numberthis \label{eq:iron-bundle}
    \end{align*}
    Once again, we bound the above using ``graph matrix moment bounds''.
    To apply \cite[Lemmas 8.7 and 8.15]{AMP16}, we will need a lower bound on the size of the minimum vertex separator between $e$ and $e'$ in $\sgraph$.
    Observe that since $\sgraph$ is connected, the minimum vertex separator between $e$ and $e'$ in $\sgraph$ has size at least $1$.
    Thus:
    \begin{align*}
        \E_{\rvG\sim\Null}
        \Tr\parens*{
            \parens*{u_{\sgraph, e, e'}\parens*{ \rvG }u_{\sgraph, e, e'}
                \parens*{ \rvG }^{\top}
            }^{4q}
        }
        &\le
        \frac{1}{\abs{\calL_{\sgraph}}^{4q}} \parens*{ 8q }^{8q(e(\sgraph) + v(\sgraph))} \cdot n^{4q(v(\sgraph)-1)+1} \\
        &\le (8v(\sgraph)q)^{8q(e(\sgraph) + v(\sgraph))} n^{-4q+1}\mper
    \end{align*}
    Plugging the above into \Cref{eq:iron-bundle} gives:
    \[
        \ell_{4q}\parens*{ \norm{ u_{\sgraph,e,e'}\parens*{\rvM_{\eps}} } } \le
        O(n^{3/8q})\cdot\frac{(8v(\sgraph)q)^{e(\sgraph)+v(\sgraph)}}{\sqrt{n}}
    \]
    We explicitly choose $q = \log n$ and plug in the above into \Cref{eq:slitherwing} and obtain:
    \begin{align*}
        \ell_{\log n}\parens*{\kappa_2(\rvM))} \le O(1) \cdot \frac{\eps\cdot(8v(\sgraph)\log n)^{e(\sgraph)+v(\sgraph)}}{\sqrt{n}}\mper
        \numberthis \label{eq:final-bound-kappa2}
    \end{align*}

    \parhead{High-probability lower bound on $\sigma^2(\rvM)$.}
    %
    To lower bound $\sigma^2(\rvM)$ with high probability, we first give a deterministic lower bound on $\sigma^2(\rvM)$ by an expression that is analytically more tractable.
    We then show a high probability lower bound on this expression.

    Fix an arbitrary matrix $M \in \R^{n\times n}$.
    Let $p(\rvG) = \chi_{\sgraph}\parens*{ \sqrt{1-\eps}M + \sqrt{\eps}\rvG }$.
    For a function $f \colon \R^d \rightarrow \R$, let $D^t \,f$ be the tensor containing all partial derivatives of order $t$ where there are no repetitions in the directions we take derivatives in.
    Since $p$ is multi-linear, all other derivatives of order $t$ are 0.
    Thus, by Taylor expansion
    \begin{align*}
        p(\rvG) = \sum_{t=0}^{e(\vartheta)} \frac 1 {t!} \angles*{D^t\, p(0), \rvG^{\otimes t}} \mcom
    \end{align*}
    where by $\rvG^{\otimes t}$ we mean the tensor that contains all products of $t$ distinct entries of $\rvG$.
    Note that for indices $s \neq t \geq 1$, the corresponding terms in the sum above are uncorrelated and have mean 0.
    Thus, using that the constant term doesn't influence the variance we obtain
    $$
        \sigma^2(M) = \Var_{\rvG} p(\rvG) = \sum_{t=1}^{e(\vartheta)} \frac 1 {(t!)^2} \E_{\rvG} \angles*{D^t \,p(0), \rvG^{\otimes t}}^2 \mper
    $$
    Further, 
    $$
        D^t \,p(0) = \sqrt{\e}^t \cdot D^t \,\chi_\vartheta\parens*{\sqrt{1-\e} M} = \sqrt{\e}^t \sqrt{1-\e}^{e(\vartheta-t)}\cdot D^t \,\chi_\vartheta\parens*{M} \,. 
    $$
    
    So we get 
    \begin{align*}
        \sigma^2(M)  &= \sum_{t\ge 1} \frac{\eps^t (1-\eps)^{e(\sgraph)-t}}{(t!)^2} \cdot \E_{\rvG} \angles*{\mathrm{D}^t\, \chi_{\sgraph}\parens*{ M }, \rvG^{\otimes t} }^2 = \sum_{t\ge 1} \frac{\eps^t (1-\eps)^{e(\sgraph)-t}}{(t!)^2}  \cdot \norm{ \mathrm{D}^t\, \chi_{\sgraph}\parens*{ M }}_2^2 \\
        &\ge \eps\cdot(1-\eps)^{e(\sgraph)-1}\cdot \norm{ \grad \chi_{\sgraph}\parens*{M} }_2^2 \numberthis \label{eq:variance-lower-bound}
    \end{align*}
    Thus, to prove a high-probability lower bound on $\sigma^2(\rvM)$, it suffices to prove a high-probability lower bound on $\norm{ \grad \chi_{\sgraph}\parens*{\rvM} }^2$.
    To do so, we will express $\norm{ \grad \chi_{\sgraph}(\rvM) }_2^2$ as the sum of a ``deterministic term'' and a ``fluctuation term''.
    The deterministic and fluctuation part arise from writing the expression as a linear combination of \emph{subgraph statistics} of $\rvM$.

    We will also consider subgraph statistics of \emph{multigraphs} in the sequel. For a multigraph $\iota$ and a matrix $M$, we define the subgraph statistic $\chi_{\iota}(M)$ as $\sum_{\pi\in\calL_{\iota}} \prod_{ab\in E(\iota)} h_{\mathrm{mult}(ab)}\parens*{ M_{\pi(a),\pi(b)} }$,
    where $\mathrm{mult}(ab)$ refers to the multiplicity of the edge $ab$ in $\iota$, and $h_k$ refers to the $k$-th (monic) Hermite polynomial.

    Now, let's write:
    \begin{align*}
        \norm{ \grad \chi_{\sgraph}(\rvM) }_2^2 &=
        \sum_{\{i,j\}\in{[n]\choose 2}} \frac{1}{\abs{ \calL_{\sgraph} } }
        \parens*{
            \sum_{e\in E(\sgraph)}
            \sum_{ \pi \in \calL_{\sgraph, e, \{i,j\} } }
            \prod_{ ab \in E(\sgraph) \setminus e } \rvM_{\pi(a), \pi(b)}
        }^2 \\
        &= \sum_{\{i,j\}\in{[n]\choose 2}} \frac{1}{\abs{ \calL_{\sgraph} } }
        \sum_{e,e'\in E(\sgraph)}
        \sum_{\substack{ \pi \in \calL_{\sgraph, e, \{i,j\}} \\ \pi' \in \calL_{\sgraph, e', \{i,j\}} } } \prod_{ab\in E(\sgraph)\setminus e} \rvM_{\pi(a), \pi(b)} \cdot
        \prod_{ab \in E(\sgraph)\setminus e'} \rvM_{\pi'(a),\pi'(b)}\mper  \numberthis \label{eq:grad-square}
    \end{align*}

    Now, we define the notion of an \emph{overlap function}.
    An overlap function $\Ov$ is an injective map from a subset of vertices $S\subseteq V(\sgraph)$ to $V(\sgraph)$.
    We use $\calO_S$ to refer to the set of all overlap functions from $S$.
    Given an overlap function $\Ov\in\calO_S$, we use $\sgraph_{\Ov}$ to refer to the graph obtained by taking two copies of $\sgraph$, and overlapping them by overlaying vertex $v\in S$ in the first copy of $\sgraph$ with $\Ov(v)$ in the second copy of $\sgraph$.

    With this notation in hand, we can write \Cref{eq:grad-square} as:
    \begin{align*}
        \eqref{eq:grad-square} = \frac{1}{\abs{ \calL_{\sgraph} } } \sum_{e,e'\in E(\sgraph)} \sum_{V(\sgraph)\supseteq S \supseteq e} \sum_{\substack{\Ov\in\calO_S \\ \Ov(e)=e'}} \sum_{\pi\in\calL_{\sgraph_{\Ov}}} \prod_{ab\in E(\sgraph_{\Ov})\setminus \{e,e'\} } \rvM_{\pi(a),\pi(b)}^{\mathrm{mult}(ab)} \numberthis \label{eq:overlap-first}
    \end{align*}
    In the above expression, $\mathrm{mult}(ab)$ is either equal to $1$ or $2$, and when it is $2$, the corresponding term can be split into $h_2\parens*{\rvM_{\pi(a),\pi(b)}}+1$.

    Define $\calC_{\text{proto}}(\sgraph, e, e')$ as the multiset of all multigraphs of the form $\sgraph_{\Ov}$ for $\Ov\in\calO_S$ such that $S\supseteq e$ and $\Ov(e) = e'$.

    We now construct a multiset $\calC(\sgraph,e,e')$ from $\calC_{\text{proto}}(\sgraph,e,e')$ via the following iterative method.
    \begin{itemize}
        \item Define $\calC_0$ as $\calC_{\text{proto}}(\sgraph,e,e')$, and for every graph $\iota\in\calC_{\text{proto}}(\sgraph,e,e')$, color all its double-edges red.
        \item For time $t\ge 1$, perform the following until termination.
        \item Suppose every graph $\iota\in\calC_{t-1}$ has no red double-edges, define $\calC(\sgraph,e,e') \coloneqq \calC_{t-1}$ and terminate.
        \item Otherwise, pick a graph $\iota$ that has a red double-edge.
        \item Pick a red double-edge $e$ in $\iota$ and define the following two new graphs: $\iota_1$ by deleting $e$, and $\iota_2$ by recoloring $e$ blue.
        \item Let $\calC_t \coloneqq \calC_{t-1} \cup \{\iota_1,\iota_2\} \setminus\{\iota\}$.
    \end{itemize}
    The reader should interpret the ``splitting'' of $\iota$ into $\iota_1$ and $\iota_2$ as writing $x^2$ as $h_2(x) + 1$.
    Finally, define $\calC(\sgraph)$ as the multiset obtained as $\bigcup_{e,e'}\calC(\sgraph,e,e')$.
    
    It is straightforward to verify that we can write:
    \begin{align*}
        \eqref{eq:overlap-first} = \frac{1}{\abs{ \calL_{\sgraph} } } \sum_{\iota\in\calC(\sgraph)} \chi_{\iota}\parens*{ \rvM }\mper
    \end{align*}
    We can split the above sum into a deterministic term and a fluctuation term, based on whether $\iota$ is empty or nonempty.
    In particular:
    \[
        \eqref{eq:overlap-first} = \underbrace{\frac{1}{\abs{\calL_{\sgraph}}}
            \sum_{ \substack{\iota\in\calC(\sgraph) \\ \iota\text{ empty}} } \chi_{\iota}\parens*{ \rvM }}_{\text{Deterministic term}} +
        \underbrace{\frac{1}{\abs{\calL_{\sgraph}}}
            \sum_{ \substack{\iota\in\calC(\sgraph) \\ \iota\text{ nonempty}} }
            \chi_{\iota}\parens*{ \rvM }}_{\text{Fluctuation term}}
    \]
    The first term does not depend on $\rvM$ as each summand is deterministically equal to $\abs{\calL_{\iota}}$.
    Observe that the magnitude of the deterministic part is at least $1$, since for $\iota$ arising from overlaying two copies of $\sgraph\setminus e$ according to the identity map for some edge $e$, we have $\abs{\calL_{\iota}} = \abs{\calL_{\sgraph}}$.
    We record this observation below:
    \begin{align*}
        \text{Deterministic term} \ge 1 \numberthis \label{eq:deterministic-lb}
    \end{align*}

    We will prove that with high probability, the magnitude of the fluctuation term is at most $\frac{\polylog n}{\sqrt{n}}$ by bounding its $\ell_{q}$-norm for appropriately chosen even $q$.
    \begin{align*}
        \ell_q\parens*{ \text{Fluctuation term} } =
        \ell_q\parens*{ \frac{1}{\abs{\calL_{\sgraph}}}
            \sum_{ \substack{\iota\in\calC(\sgraph) \\ \iota\text{ nonempty}} }
            \chi_{\iota}\parens*{ \rvM } } &\le \frac{1}{\abs{\calL_{\sgraph}}} \sum_{ \substack{\iota\in\calC(\sgraph) \\ \iota\text{ nonempty}} } \ell_q\parens*{ \chi_{\iota}\parens*{ \rvM } }    \numberthis \label{eq:nonempty-terms}
    \end{align*}
    We now analyze a summand of the RHS of the above.
    \begin{align*}
        \ell_q\parens*{ \chi_{\iota}\parens*{\rvM} }^q &= \E_{\rvM\sim\Planted} \chi_{\iota}\parens*{\rvM}^q \\
        &\le \E_{\rvG\sim\Null} \chi_{\iota}\parens*{\rvG}^q + \sqrt{n \cdot \E_{\rvG\sim\Null} \chi_{\iota}\parens*{\rvG}^{2q} } \\
        &\le O\parens*{\sqrt{n}} \cdot \sqrt{ \E_{\rvG\sim\Null} \chi_{\iota}\parens*{\rvG}^{2q} }
    \end{align*}
    By the graph matrix moment bounds \cite[Lemmas 8.7 and 8.15]{AMP16}, we have:
    \[
        \E_{\rvG\sim\Null} \chi_{\iota}\parens*{\rvG}^{2q} \le (8q)^{2q e(\iota)}\cdot(2q)^{2q v(\iota)} n^{q (v(\iota) + \mathrm{iso}(\iota))},
    \]
    where $\mathrm{iso}(\iota)$ is the number of isolated vertices in $\iota$,
    and thus, as an upshot we have:
    \begin{align*}
        \ell_q\parens*{\chi_{\iota}\parens*{\rvM}} \le O\parens*{n^{1/4q}} \cdot (8q)^{e(\iota)+v(\iota)}\cdot n^{(v(\iota)+\mathrm{iso}(\iota))/2}
    \end{align*}
    Consequently:
    \begin{align*}
        \eqref{eq:nonempty-terms} \le O\parens*{n^{1/4q}}\cdot \frac{v(\sgraph)^{v(\sgraph)}}{n^{v(\sgraph)}} \cdot 2^{v(\sgraph)}v(\sgraph)^{v(\sgraph)}\cdot \max_{\substack{\iota\in\calC(\sgraph) \\ \iota\text{ nonempty}}} (8q)^{e(\iota)+v(\iota)}\cdot n^{(v(\iota)+\mathrm{iso}(\iota))/2}   \numberthis \label{eq:chilinko}
    \end{align*}
    Any $\iota$ in $\calC(\sgraph)$ arises from overlaying two copies of $\sgraph$ via some map $\Ov\in\calO_{S(\iota)}$ for some subset of vertices $S(\iota)$ of size at least $2$, and so $v(\iota) = 2v(\sgraph) - \abs{S(\iota)}$.
    In particular, the RHS of the above is at most
    \begin{align*}
        \max_{\substack{\iota\in\calC(\sgraph) \\ \iota\text{ nonempty}}} O\parens*{n^{1/4q}}\cdot (16qv(\sgraph))^{2e(\sgraph) + 2v(\sgraph)} \cdot \sqrt{n}^{\mathrm{iso}(\iota)-|S(\iota)|}   \numberthis \label{eq:max-iota-treatable}
    \end{align*}
    We will prove that for any $\iota\in\calC(\sgraph)$, the power of $n$ in the above term is always at most $-1/2$ by showing that the set of isolated vertices $I$ in $\iota$ is a strict subset of $S = S(\iota)$.

    Suppose $I$ is empty, then the claim is true since $S$ is nonempty.
    Henceforth, we assume that $I$ is nonempty.
    Define $\wt{\iota}$ as the graph obtained by taking the (multiset) union of the two overlapped copies of $\sgraph$ (without deleting any edges) via the map $\Ov$.
    Observe that $\wt{\iota}$ is a connected graph since $\sgraph$ is connected, and $\wt{\iota}$ is obtained by overlapping two copies of $\sgraph$ on a nonempty set.
    Since $\iota$ is not an empty graph, its set of nonisolated vertices $T$ must be nonempty.
    Since $I = V\parens*{\wt{\iota}} \setminus T$, and $\wt{\iota}$ is connected, there must be a nonempty cut between $I$ and $T$.
    Observe that for an isolated vertex $v$ in $\iota$, every edge incident to $v$ in $\wt{\iota}$ is a double-edge, and additionally, if a vertex $v$ is incident to a double-edge, it must be in $S$.
    Every edge in the cut between $I$ and $T$ is a double-edge, and hence there must be some vertex $u\in T$ that is in $S$, and so $I$ is a strict subset of $T$, which leads us to conclude the following:
    \begin{align*}
        \eqref{eq:max-iota-treatable} \le O\parens*{n^{1/4q}} \frac{\parens*{16qv(\sgraph)}^{2(e(\sgraph)+v(\sgraph))}}{\sqrt{n}}
    \end{align*}
    Explicitly choosing $q = \log n$, and using the above along with $\eqref{eq:nonempty-terms}\le\eqref{eq:chilinko}\le\eqref{eq:max-iota-treatable}$, we get:
    \begin{align*}
        \ell_{\log n}\parens*{\text{Fluctuation term}} \le \frac{\parens*{32v(\sgraph)\log n}^{2(e(\sgraph)+v(\sgraph))}}{\sqrt{n}} \numberthis \label{eq:final-fluctuation}
    \end{align*}

    \parhead{Total variation distance bound.}
    By \Cref{eq:final-gradgrad}, \Cref{eq:final-bound-kappa2}, \Cref{eq:variance-lower-bound}, \Cref{eq:deterministic-lb}, and \Cref{eq:final-fluctuation}, and Markov's inequality, we have that with probability at least $1-\frac{1}{n^2}$ over $\rvM\sim\Planted$:
    \begin{align*}
        \kappa_1\parens*{\rvM} &\le \sqrt{\eps}\cdot\polylog n \\
        \kappa_2\parens*{\rvM} &\le \eps\cdot\frac{\polylog n}{\sqrt{n}} \\
        \sigma^2\parens*{\rvM} &\ge \eps(1-\eps)^{v(\sgraph)-1}\cdot\parens*{1 - \frac{\polylog n}{\sqrt{n}}} \mper
    \end{align*}
    We will use $\calE$ to refer to the event that the above bounds hold.
    For a fixed choice of $M$, we define $\pi_{\sgraph,M}$ as the law of $\chi_{\sgraph}\parens*{\sqrt{1-\eps}M + \sqrt{\eps}\rvG}$ for $\rvG\sim\Null$, and define $\ul{\pi}_{\sgraph,M}$ as the law of $\chi_{\sgraph}\parens*{\sqrt{1-\eps}M}+\wt{\sigma}(M)\cdot\rvg$ where $\rvg \sim \calN(0,1)$.
    By \Cref{lem:shivam}, for $M\in\calE$ we have:
    \[
        \dtv\parens*{\pi_{\sgraph,M}, \ul{\pi}_{\sgraph,M} } \le \frac{\eps^{3/2}\cdot \polylog n}{\Omega(\eps)\cdot \sqrt{n} } = \frac{\sqrt{\eps}\cdot\polylog n}{\sqrt{n}}\mper
    \]
    Now, observe that $\pi_{\sgraph} = \E_{\rvM\sim\Planted} \pi_{\sgraph,\rvM}$, and $\ul{\pi}_{\sgraph} = \E_{\rvM\sim\Planted} \ul{\pi}_{\sgraph,\rvM}$.
    
    We can finally prove the desired statement via the following chain of inequalities.
    \begin{align*}
        \dtv\parens*{\pi_{\sgraph}, \ul{\pi}_{\sgraph}} &\le \E_{\rvM\sim\Planted} \dtv\parens*{ \pi_{\sgraph,\rvM}, \ul{\pi}_{\sgraph,\rvM} } \\
        &= \E_{\rvM\sim\Planted} \Ind\bracks*{\calE}\cdot\dtv\parens*{ \pi_{\sgraph,\rvM}, \ul{\pi}_{\sgraph,\rvM} } + \E_{\rvM\sim\Planted} \Ind\bracks*{\ol{\calE}}\cdot\dtv\parens*{ \pi_{\sgraph,\rvM}, \ul{\pi}_{\sgraph,\rvM} } \\
        &\le \frac{\sqrt{\eps}\cdot \polylog n }{\sqrt{n}} + \frac{1}{n^2}
        \qedhere
    \end{align*}
\end{proof}

\subsection{Infinity Norm Bounds}
\label{sec:infty-norm-bound}

In this section we prove \cref{lem:infty-norm-bound} restated below.
\restatelemma{lem:infty-norm-bound}

\begin{proof}
For any $x \in \R$ by Fourier inversion we can write:
    \begin{align*}
        \sqrt{2\pi}\abs{ \ul{\pi}_{\sgraph}(x) - \ul{\nu}_{\sgraph}(x) }
        &\le
        \int_{-\infty}^{\infty} \abs*{ \wh{\ul{\pi}}_{\sgraph}(\xi) - \wh{\ul{\nu}}_{\sgraph}(\xi)}  \dif\xi \\
        &=
        \int_{[-T,T]} \abs*{ \wh{\ul{\pi}}_{\sgraph}(\xi) - \wh{\ul{\nu}}_{\sgraph}(\xi)}  \dif\xi
        +
        \int_{\R\setminus[-T,T]} \abs*{ \wh{\ul{\pi}}_{\sgraph}(\xi) - \wh{\ul{\nu}}_{\sgraph}(\xi)}  \dif\xi   \numberthis \label{eq:split-fourier}
    \end{align*}
    for some $T\in\R$ that we choose later.
    We bound the first term via moment-matching (bounded LDLR) and the second term using that the added noise attenuates the large Fourier coefficients.

    \parhead{Case $|\xi| > T$.}
    Recall that $\wt{\sigma}(\rvM) \geq \tfrac \e 2 (1-\e)^{e(\vartheta)-1}$ and observe that
    \begin{align*}
        \abs*{\ul{\wh{\pi}}_{\sgraph}\parens*{ \xi }}
        &=
        \abs*{\E_{\rvM\sim\Planted} \E_{\rvg\sim\calN(0,1)} \exp\parens*{ i \xi \parens*{ \sqrt{1-\eps}^{e(\sgraph)}\chi_{\sgraph}\parens*{\rvM} + \wt{\sigma}\parens*{\rvM}\cdot\rvg } }} \\
        &\le \abs*{\E_{\rvM\sim\Planted}  \exp\parens*{ i \xi \parens*{ \sqrt{1-\eps}^{e(\sgraph)}\chi_{\sgraph}\parens*{\rvM} } } \cdot \exp\parens*{-\xi^2\wt{\sigma}^2\parens*{ \rvM}/2 } } \\
        &\le \exp\parens*{ -\xi^2 \eps^2 (1-\eps)^{2e(\sgraph)-2} / 8 }\mper
    \end{align*}
    When $\Planted = \Null$, the above establishes the same bound for $\abs*{\ul{\wh{\nu}}_{\sgraph}\parens*{ \xi }}$.
    Plugging in this bound into the second term of \Cref{eq:split-fourier} and integrating gives:
    \begin{align*}
         \int_{\R\setminus[-T,T]} \abs*{ \wh{\ul{\pi}}_{\sgraph}(\xi) - \wh{\ul{\nu}}_{\sgraph}(\xi)}  \dif\xi
         \le
         \sqrt{\frac{4\pi}{\eps^2(1-\eps)^{2e(\sgraph)-2}}} \cdot
         \exp\parens*{ -\frac{T^2 \eps^2 (1-\eps)^{2e(\sgraph)-2}}{8} } \numberthis \label{eq:iron-hands}
    \end{align*}

    We are now ready to bound the first term.

    \parhead{Case $|\xi| \le T$.}
    In this case, observe that:
    \begin{align*}
        \abs*{\ul{\wh{\pi}}_{\sgraph}\parens*{\xi} - \ul{\wh{\nu}}_{\sgraph}\parens*{\xi}}
        &\le
        \abs*{{\wh{\pi}}_{\sgraph}\parens*{\xi} - {\wh{\nu}}_{\sgraph}\parens*{\xi}} + \dtv\parens*{ \ul{\pi}_{\sgraph}, \pi_{\sgraph} } + \dtv\parens*{ \ul{\nu}_{\sgraph}, \nu_{\sgraph} } \\
        &\le \abs*{ \wh{\pi}_{\sgraph}\parens*{\xi} - \wh{\nu}_{\sgraph}\parens*{\xi} } + \frac{\mathrm{\polylog n}}{\sqrt{n}}\mper \numberthis \label{eq:pass-to-pi}
    \end{align*}
    To control $\abs*{ \wh{\pi}_{\sgraph}\parens*{\xi} - \wh{\nu}_{\sgraph}\parens*{\xi} }$, we will use the fact that the moments of $\pi_{\sgraph}$ and $\nu_{\sgraph}$ approximately match.
    Indeed, by Taylor's theorem, and for a choice of $k \ll D$ we will make later in the proof, we can write:
    \begin{align*}
        \abs*{ \wh{\pi}_{\sgraph}\parens*{\xi} - \wh{\nu}_{\sgraph}\parens*{\xi} }
        &\le
        \abs*{ \sum_{t=0}^{k-1} \E_{\rvM\sim\OU_{\eps}\Planted} \frac{\parens*{ i\xi\chi_{\sgraph}(\rvM) }^t }{t!} - \E_{\rvM\sim\Null} \frac{\parens*{ i\xi\chi_{\sgraph}(\rvM) }^t }{t!} } + 
        \E_{\rvM\sim\OU_{\eps}\Planted}\frac{\abs*{ \xi\chi_{\sgraph}(\rvM) }^k }{k!}
        +
        \E_{\rvM\sim\Null}\frac{\abs*{ \xi\chi_{\sgraph}(\rvM) }^k }{k!} \\
        &\le
        \sum_{t=0}^{k-1} \frac{|\xi|^t}{t!} \sqrt{\chisquaretrunc{\Planted}{\Null} \E_{\rvM\sim\Null}  \chi_{\sgraph}(\rvM)^{2t} } + 2 \E_{\rvM\sim\Null} \frac{\abs*{ \xi \chi_{\sgraph}(\rvM) }^k }{k!} + \frac{|\xi|^k}{k!} \sqrt{ \chisquaretrunc{\Planted}{\Null} \E_{\rvM\sim\Null} \chi_{\sgraph}(\rvM)^{2k} } \mper
    \end{align*}

    Let $C$ be a sufficiently large absolute constant.
    Now by \cref{lem:subgraph-moments} we know that $\E_{\rvM\sim\Null}\chi_{\sgraph}(\rvM)^q \le O(q)^{q/2}$, so the last two terms are at most $(1+\sqrt{\delta})\frac{C^k \abs{\xi}^k}{k^{k/2}}$.
    Similarly, the first term is at most $\sqrt{\delta}\sum_{t= 0}^{k-1} \frac{C^t\abs*{\xi}^t}{t^{t/2}}$.
    Overall, we obtain
    \begin{align*}
        \abs*{ \wh{\pi}_{\sgraph}\parens*{\xi} - \wh{\nu}_{\sgraph}\parens*{\xi} }& \leq \sqrt{\delta}\sum_{t\ge 0} \frac{C^t\abs*{\xi}^t}{t^{t/2}} + \frac{C^k|\xi|^k}{k^{k/2}} \le \sqrt{\delta}(1+C|\xi|)\exp\parens*{ C^2|\xi|^2 } + \frac{C^k|\xi|^k}{k^{k/2}}\\
        &\le \sqrt{\delta}(1+CT)\exp(C^2 T^2) + \frac{C^k T^k}{k^{k/2}} \mper   \numberthis \label{eq:moments-matched}
    \end{align*}

    By \Cref{eq:pass-to-pi} and \Cref{eq:moments-matched}, we have:
    \begin{align*}
        \int_{[-T,T]} \abs*{ \wh{\ul{\pi}}_{\sgraph}(\xi) - \wh{\ul{\nu}}_{\sgraph}(\xi)}  \dif\xi \le \frac{T\polylog n}{\sqrt{n}} + \frac{C^k T^{k+1}}{k^{k/2}} + \sqrt{\delta} (CT^2 + T)\exp(C^2 T^2)\mper  \numberthis \label{eq:matcha-sinistcha-gotcha}
    \end{align*}
    We choose $k = \log n$, and $T = \alpha \min\braces*{\sqrt{\log n}, \sqrt{\log\frac{1}{\delta}} }$ for a sufficiently small constant $\alpha > 0$.
    By plugging in \Cref{eq:iron-hands} and \Cref{eq:matcha-sinistcha-gotcha} into \Cref{eq:split-fourier}, we get for a (new) sufficiently large constant $C$ and sufficiently small constant $c$:
    \begin{align*}
        \norm{ \ul{\pi}_{\sgraph} - \ul{\nu}_{\sgraph} }_{\infty}
        &\le
        O\parens*{\frac{1}{\eps}} \cdot \exp\parens*{ -c\eps^2 T^2 } + \frac{T\polylog n}{\sqrt{n}} + \frac{C^k T^{k+1}}{k^{k/2}} + C\sqrt{\delta} T^2 \exp(C^2T^2)\mper
    \end{align*}
    Using that $\e \geq 1/\polylog n$, for a (new) sufficiently small constant $\alpha$, we can see that:
    \begin{align*}
        \norm{ \ul{\pi}_{\sgraph} - \ul{\nu}_{\sgraph} }_{\infty} &\le O\parens*{n^{-\alpha\eps^2} + \delta^{\alpha\eps^2}}\mper   \qedhere
    \end{align*}
\end{proof}

\section*{Acknowledgments}
\label{sec:ack}
We would like to thank Sam Hopkins for countless hours of illuminating conversations on the low-degree heuristic.
S.M.\ is grateful to Sam Hopkins and Siqi Liu for discussions on proto-versions of the problems considered here, as well as Prasad Raghavendra and Tselil Schramm for many inspiring conversations.
We are indebted to Shivam Nadimpalli for pointing us to \cite{Cha09}.
We would also like to thank Rares Darius Buhai for helpful conversations, and anonymous reviewers for their heroic feedback on an earlier version of this paper.

Finally, we are grateful to the American Institute of Mathematics for their hospitality and support at the workshop on \emph{``Low-degree polynomial methods in average-case complexity''} where this collaboration began, and the organizers Sam Hopkins, Tselil Schramm, and Alex Wein who facilitated this event.

\bibliographystyle{alpha}
\bibliography{main}

\appendix
\section{Binomial Distributions}
\label{sec:binomial-appendix}
In this section, we provide proofs that we omit in \Cref{sec:binomial-prelims}.

\restatelemma{lem:binomial-prob}

\begin{proof}
    Let $\beta \coloneqq \berp + y\sqrt{\berp(1-\berp)/n}$.
    Then,
    \begin{align*}
        \Pr[\Bin(n,\berp) = \beta n] = \binom{n}{\beta n} \berp^{\beta n} (1-\berp)^{(1-\beta)n} \mper
    \end{align*}
    First, we have
    $$\binom{n}{\beta n} \geq \sqrt{\frac{1}{8n\beta (1-\beta)}} e^{n h(\beta)} \geq \frac{1}{\sqrt{2n}} e^{n h(\beta)}\mcom$$
    where $h(\beta) = -\beta \log \beta - (1-\beta) \log(1-\beta)$.
    Next, a simple calculation shows that 
    $$e^{nh(\beta)} \berp^{\beta n} (1-\berp)^{(1-\beta)n} = e^{-n D(\beta \| \berp)}\mcom$$
    where $D(\beta \| \berp) = \beta \log \frac{\beta}{\berp} + (1-\beta)\log(\frac{1-\beta}{1-\berp})$ is the KL-divergence between $\Ber(\beta)$ and $\Ber(\berp)$.

    We use the bound $D(\berp+\eps\|\berp) \leq \frac{\eps^2}{2\berp(1-\berp)} + O(|\eps|^3/\berp^2)$ from \Cref{fact:KL-bernoulli}, assuming $\berp \leq 1/2$ and $\eps \leq o(\berp)$.
    Setting
    $$|\eps| = |y|\sqrt{\berp(1-\berp)/n} \leq |y|\sqrt{\berp/n} \leq o(\berp)$$
    (by assumption), we get
    $$n\cdot D(\beta\|\berp) \leq y^2/2 + O(|y|^3/\sqrt{\berp n})\mper$$
    This completes the proof.
\end{proof}

\restatelemma{lem:binomial-moments}
\begin{proof}
    Let $\bz_i = \frac{1}{2\sqrt{\berp(1-\berp)}} (\bx_i-(2\berp-1))$.
    A standard calculation shows that $\E[\bz_i] = 0$, $\E[\bz_i^2]=1$, and $0 < \E[\bz_i^k] \leq (\frac{1-\berp}{\berp})^{k/2-1} \leq \berp^{-(k-2)/2}$ for $k\geq 3$ (here we assume $\berp \leq 1/2$).

    We have $\by = \frac{1}{\sqrt{n}}\sum_{i=1}^n \bz_i$ and thus
    \begin{align*}
        \E[\by^{2k}] = n^{-k} \sum_{i_1, i_2,\dots, i_{2k} \in [n]} \E[\bz_{i_1} \bz_{i_2}\cdots \bz_{i_{2k}}] \mper
    \end{align*}
    Since $\E[\bz_i] = 0$, the summand is zero if any $u\in [n]$ appears exactly once in the tuple $(i_1,\dots, i_{2k})$.
    Let $m_u$ be the multiplicity of $u$ in $(i_1,\dots,i_{2k})$, and let $\ell = \sum_u (m_u-2) \one(m_u \geq 3)$, i.e., the multiplicities that exceed $2$.
    Note that $\ell$ must be even, otherwise some $m_u$ would be $1$.
    Then,
    \begin{align*}
        \E[\bz_{i_1}\cdots \bz_{i_{2k}}]
        = \prod_{u\in [n]} \E[\bz_u^{m_u}] \leq \prod_{u\in [n]} \berp^{-(m_u-2)/2 \cdot \one(m_u \geq 3)}
        = \berp^{-\ell/2} \mper
    \end{align*}
    The dominating term is when $\ell = 0$, i.e., each index appears $0$ or $2$ times.
    In this case, the number of such tuples can be formed by first choosing a perfect matching between $1,2,\dots, 2k$ (for the repeating pattern), and then choosing $k$ distinct indices from $[n]$ in order.
    The number of matchings is $(2k-1)!!$, and the number of choices for the $k$ indices is at most $n^k$.

    For $\ell > 0$, we may choose the tuples as follows: (1) select $\ell$ positions among $\{1,2,\dots,2k\}$ to be the ``high-multiplicity'' indices; (2) for each position, we match it with another position in $\{1,2,\dots,2k\}$; (3) for the remaining $2k-\ell$ positions (an even number), we choose a perfect matching; (4) select $k-\ell/2$ distinct indices from $[n]$.
    Thus,
    \begin{align*}
        \E[\by^{2k}] &\leq n^{-k} \sum_{\ell=0,2,4,\dots} \berp^{-\ell/2} (2k)^{2\ell} (2k-\ell-1)!! \cdot n^{k-\ell/2} \\
        &\leq (2k-1)!! \sum_{\ell=0,2,4,\dots} \berp^{-\ell/2} (2k)^{2\ell} \parens*{\frac{2}{k}}^{\ell/2} n^{-\ell/2} \\
        &\leq (2k-1)!! \cdot \parens*{1 + O\parens*{\frac{k^3}{\berp n}}} \mper
    \end{align*}
    This completes the proof.
\end{proof}

\section{Upper Bounds on \Krawdaunt Polynomials}
\label{sec:krawtchouk-bound-proof}

In this section, we provide a proof of \Cref{lem:krawtchouk-bound} that closely follows the proof given in \cite{Pol19,DILV24}.

\restatelemma{lem:krawtchouk-bound}

\begin{proof}
    We first note that by definition,
    \begin{align*}
        \prod_{i=1}^n (1 + \chi_i(x) z) = \sum_{k=0}^n s_k(x) z^k \mper
        \numberthis \label{eq:generating-function}
    \end{align*}
    Let $w \coloneqq \sum_{i=1}^n \frac{x_i+1}{2} \in \{0,1,\dots,n\}$, which equals the number of $1$s in $x$ and is distributed as $\text{Bin}(n,\berp)$.
    Denote $\berq \coloneqq 1-\berp$, we have
    \begin{align*}
        y = \frac{1}{\sqrt{n}}\sum_{i=1}^n \frac{x_i-(2\berp-1)}{2\sqrt{\berp \berq}} = \frac{1}{\sqrt{n\berp \berq}} (w-\berp n)
        \implies w = \berp n + \sqrt{n\berp \berq} y \mper
        \numberthis \label{eq:w-y-relation}
    \end{align*}
    Let $K_k(w)$ be the degree-$k$ univariate polynomial representing $s_k(x)$.
    Then, we have $\chi_i(x) = \sqrt{\berq/\berp}$ if $x_i = 1$ and $-\sqrt{\berp/\berq}$ if $x_i=-1$, and thus \Cref{eq:generating-function} is equivalent to
    \begin{align*}
        f_w(z) \coloneqq \parens*{1 + \sqrt{\frac{\berq}{\berp}}z}^w \parens*{1 - \sqrt{\frac{\berp}{\berq}}z}^{n-w}
        = \sum_{k=0}^n K_k(w) z^k \mper
    \end{align*}
    This is the generating function for the usual \Krawdaunt polynomials.
    In particular, we have $K_k(w) = \frac{1}{k!} f_w^{(k)}(0)$, and thus by Cauchy's integral formula,
    \begin{align*}
        K_k(w) = \frac{1}{2\pi i} \oint_C f_w(z) z^{-k} \frac{dz}{z} \mcom
    \end{align*}
    where $C$ is an arbitrary circle on the complex plane centered at $z=0$, oriented counterclockwise.
    We will choose an appropriate radius $r$ and bound $|f_w(z) z^{-k}|$ for $|z|=r$.

    For simplicity, denote $\alpha = w/n$ and $\beta = k/n$.
    First, we rewrite $f_w(z)$ as
    \begin{align*}
        f_w(z) = \berp^{-\alpha n/2} \berq^{-(1-\alpha)n/2} \parens*{\sqrt{\berp} + \sqrt{\berq} z}^{\alpha n} \parens*{\sqrt{\berq} - \sqrt{\berp} z}^{(1-\alpha)n} \mper
    \end{align*}
    Then,
    \begin{align*}
        \log |f_w(z) z^{-k}|
        &= \frac{n}{2} \Big(-\alpha \log \berp - (1-\alpha) \log \berq \\
        &\quad \quad \quad + \alpha \log |\sqrt{\berp} + \sqrt{\berq} z|^2 + (1-\alpha) \log |\sqrt{\berq} - \sqrt{\berp} z|^2 
        - \beta \log |z|^2 \Big) \\
        &= \frac{n}{2} \parens*{ D(\alpha \| \berp) + \alpha \log \frac{|\sqrt{\berp} + \sqrt{\berq} z|^2}{\alpha} + (1-\alpha) \log \frac{|\sqrt{\berq} - \sqrt{\berp} z|^2}{1-\alpha} -\beta \log |z|^2 } \mper
    \end{align*}
    Here $D(\alpha \| \alpha') \coloneqq  \alpha \log(\frac{\alpha}{\alpha'}) + (1-\alpha) \log(\frac{1-\alpha}{1-\alpha'})$ is the KL divergence between two Bernoulli distributions with parameter $\alpha$ and $\alpha'$.
    Next, we use the concavity of $\log$ and Jensen's inequality:
    \begin{align*}
        \alpha \log \frac{|\sqrt{\berp} + \sqrt{\berq} z|^2}{\alpha} + (1-\alpha) \log \frac{|\sqrt{\berq} - \sqrt{\berp} z|^2}{1-\alpha}
        &\leq \log \parens*{|\sqrt{\berp} + \sqrt{\berq} z|^2 + |\sqrt{\berq} - \sqrt{\berp} z|^2} \\
        &= \log(1 + |z|^2) \mper
    \end{align*}
    The equality follows because $|\sqrt{\berp} + \sqrt{\berq} z|^2 = \berp + \berq z^2 + \sqrt{\berp \berq}(z + \ol{z})$ while $|\sqrt{\berq} - \sqrt{\berp} z|^2 = \berq + \berp z^2 - \sqrt{\berp \berq}(z + \ol{z})$, and $\berp+\berq=1$.

    Thus, we may choose the radius $r$ of $C$ that minimizes $\log(1+r^2) - \beta\log r^2$.
    Differentiating this with respect to $r^2$, we get that $r^2 =  \frac{\beta}{1-\beta}$ is the optimal.
    Therefore, for $|z| = r$,
    \begin{align*}
        \log |f_w(z) z^{-k}| &\leq \frac{n}{2} \parens*{D(\alpha \|\berp) + \log (1+r^2) - \beta \log r^2 } \\
        &= \frac{n}{2} \parens*{D(\alpha \|\berp) + \log \frac{1}{1-\beta} - \beta \log \frac{\beta}{1-\beta} } \\
        &= \frac{n}{2} \parens*{D(\alpha \|\berp) + h(\beta) }
    \end{align*}
    where $h(\beta) = -\beta \log \beta - (1-\beta)\log(1-\beta)$.
    Then,
    \begin{align*}
        |K_k(w)| \leq e^{\frac{n}{2}(D(\alpha\|\berp) + h(\beta))} \cdot \frac{1}{2\pi} \oint_C \abs{\frac{dz}{z}} \leq e^{\frac{n}{2}(D(\alpha\|\berp) + h(\beta))} \mper
    \end{align*}
    Recall that $\alpha = w/n$ and $\beta = k/n$.
    By the inequality $\binom{n}{k} \geq \sqrt{\frac{n}{8k(n-k)}}e^{n h(k/n)} \geq \Omega(1/\sqrt{k}) e^{n h(k/n)}$ for $k \leq n/2$.
    On the other hand, we have $D(\berp+\eps \|\berp) \leq \frac{\eps^2}{2\berp(1-\berp)} + O_{\berp}(|\eps|^3)$ by \Cref{fact:KL-bernoulli}.
    From \Cref{eq:w-y-relation}, we have $\alpha = \berp + \sqrt{\frac{\berp \berq}{n}}y$, thus
    \begin{align*}
        D(\alpha\|\berp) \leq \frac{y^2}{2n} + O_{\berp}\parens*{\abs*{\frac{y}{\sqrt{n}}}^3} \mper
    \end{align*}
    Therefore,
    \begin{align*}
        |\Kr_k(y)| = \binom{n}{k}^{-1/2} |K_k(w)| \leq O(k^{1/4}) e^{y^2/4 + O_{\berp}(y^3 n^{-1/2})}
        \leq O(k^{1/4}) e^{y^2/4}
    \end{align*}
    for $|y| \leq o(n^{1/6})$.
\end{proof}

\section{Bounds on Characteristic Functions}
\label{sec:characteristic}
\restatelemma{lem:PolyFourierBoundLem}

\begin{proof}
Let $\mu = \E_{\rvg\sim\calN(0,1)}[p(\rvg)]$ and $\sigma^2 = \Var_{\rvg\sim\calN(0,1)}[p(\rvg)]$.
We will bound $|\E_{\rvg}[e^{i(p(\rvg)-\mu)}]|$ instead.
By Taylor expansion of cosine,
\begin{align*}
    \E_{\rvg} \mathrm{Re}(e^{i (p(\rvg)-\mu)}) \in 1 - \frac{1}{2} \E_{\rvg}[(p(\rvg)-\mu)^2] \pm \frac{1}{24} \E_{\rvg}[(p(\rvg)-\mu)^4] \mper
\end{align*}
Since $p$ is a polynomial of degree $k$, by the 2-to-4 hypercontractivity in the Gaussian space \cite{o2014analysis}, we have $\E_{\rvg}[(p(\rvg)-\mu)^4] \leq 9^k \cdot \E_{\rvg}[(p(\rvg)-\mu)^2]^2 = 9^k \cdot \Var_{\rvg}[p(\rvg)]^2$.
For the imaginary part, we have
\begin{align*}
    \E_{\rvg} \mathrm{Im}(e^{i (p(\rvg)-\mu)}) \in  \E_{\rvg}[(p(\rvg)-\mu)] \pm \frac{1}{6} \E_{\rvg}[|p(\rvg)-\mu|^3] \mper
\end{align*}
The first term is 0, and by 2-to-3 hypercontractivity~\cite{o2014analysis}, we have $\E_{\rvg}[|p(\rvg)-\mu|^3] \leq 2^{3k/2} \cdot \Var_{\rvg}[p(\rvg)]^{3/2}$.

Thus, if $\sigma^2 = \Var_{\rvg}[p(\rvg)] \leq 9^{-k}$, then
\begin{align*}
    \abs*{\E_{\rvg}[e^{i(p(\rvg)-\mu)}]} \leq 1 - \frac{1}{2} \sigma^2 + \frac{1}{24} 9^k \sigma^4 + \frac{1}{6} 2^{3k/2} \sigma^3
    \leq 1 - \frac{1}{4} \sigma^2 \mper
\end{align*}

For the second case, by changing $p$ by a constant, which does not affect the estimate, we can assume that $\E[ip(\rvg)]$ is positive real, and it suffices to bound the real part of this expectation above. 
Let $\rvx$ and $\rvx'$ be independent Gaussians, and let $\rvx_\theta = \cos(\theta) \rvx + \sin(\theta) \rvx'$.
Note that $\rvx_\theta$ is also a standard Gaussian. Let $f_{\rvx, \rvx'}(\theta) = p(\rvx_\theta)$, so that
\[
    \E_{\rvg} [\exp (i p(\rvg))] = \E_{\rvx, \rvx'} \E_{\theta \sim \mathrm{Unif} [0, 2 \pi]} \left[ \exp \left( i f_{\rvx, \rvx'}(\theta) \right) \right] \; .
\]
We will show that with probability at least $k^{-|c| k } / {4k}$ over the choice of $\rvx$, $\rvx'$ and $\theta$ that $f_{\rvx,\rvx'}(\theta)$ is at least $\left(k^{-|c|k/2} / 4\right)$-far from an integer multiple of $2\pi$.
This will immediately yield the desired estimate.

To begin with, by Carbery-Wright we have that with $2/3$ probability that $p(\rvx)$ and $p(\rvx')$ differ by at least $k^{-|c| k / 2}$. If this happens, then by continuity there are values of $\theta$ so that $f_{\rvx,\rvx'}(\theta)$ is more than $k^{-|c|k / 2} / 2$-far from the nearest integer multiple of $2\pi$. Additionally, we have that
\begin{align*}
\E_{\rvx,\rvx'}[\E_\theta [(f_{\rvx,\rvx'}'(\theta))^2]] & = \E_{\rvx,\rvx',\theta}[(p'(\rvx_\theta))^2 \rvx_{\theta+\pi/2}^2].
\end{align*}
Since $\rvx_\theta$ and $\rvx_{\theta+\pi/2}$ are independent Gaussians, this is
$$
    \E[(p'(\rvx_\theta))^2] \leq k \cdot \var_{\rvg\sim\calN(0,1)}[p(\rvg)].
$$
Therefore, by Markov's inequality and a union bound, there is a probability of at least $1/2$ over $\rvx$ and $\rvx'$ that there exists $\theta$ so that $f_{\rvx,\rvx'}(\theta)$ is more than $k^{ck / 2} / 2$ from an integer multiple of $2\pi$ and so that $\E_\theta[(f'_{\rvx,\rvx'}(\theta))^2] < k^{ck + 1}$. If both hold, there must be an interval in $\theta$ of length at least $\frac{1}{2k} k^{-|c| k / 2}$ for which $f_{\rvx,\rvx'}(\theta)$ is at least $(k^{-|c| k /2} / 4)$ far from the nearest multiple of $2\pi$. This completes the proof of the second case.

For the third result, we note that $\E_{\rvg\sim\calN(0,1)}[\exp(i p(\rvg))]$ is proportional to
$$
\int_{-\infty}^\infty e^{ip(x)}e^{-x^2/2} dx.
$$
Integrating by parts, we find that:
\begin{align*}
\int e^{ip(x)}e^{-x^2/2}\ dx & = \int (i p'(x)) e^{ip(x)} \cdot \frac{1}{i p'(x)} e^{-x^2/2} \ dx\\
& = \frac{e^{ip(x)} e^{-x^2/2}}{i p'(x)} - i \int e^{ip(x)} e^{-x^2/2} \left(\frac{x}{p'(x)} + \frac{p''(x)}{(p'(x))^2} \right)dx.
\end{align*}
Let $\gamma > 0$ be a parameter we set later. 
By \Cref{lem:deriv-bound}, we have that $|p'|_2^2 \geq \var_{\rvg\sim\calN(0,1)}(p(\rvg))$, and therefore by Carbery-Wright, except with probability $\gamma / 3$ we have that $|p'(\rvg)| \geq \Omega(\gamma/k)^k |p|_2$. Additionally, $|p''|_2 \leq k |p'|_2$ by \Cref{lem:deriv-bound}.
Thus, by Gaussian hypercontractivity, we have that $|p''(\rvg)| \leq O(\log(1/\gamma))^k |p'|_2$ with probability $1-\gamma / 3$. Finally, except with probability $\gamma / 3$ we have $|\rvg|\leq \log(1/\gamma)$. 
Let $E$ be the event that (1) $|p'(x)| \geq \Omega(\gamma/k)^k |p'|_2$, (2) $|p''(x)| \leq O(\log(1/\gamma)^k) |p'|_2$, and (3) $|x| \leq \log(1/\gamma)$.
By the arguments above and a union bound, we have that $\Pr[\rvg \not\in E] \leq \gamma$.
Then, we have that:
\begin{align}
\left| \int_{-\infty}^\infty e^{ip(x)}e^{-x^2/2} dx \right| &= \left| \int_E e^{ip(x)}e^{-x^2/2} dx + \int_{E^c} e^{ip(x)}e^{-x^2/2} dx \right| \notag \\
&\leq \left| \int_E e^{ip(x)}e^{-x^2/2} dx \right| + O(\gamma) \; . \label{eq:e-ec}
\end{align}
We now note that $E$ is a union of $O(k)$ intervals, since for any polynomial $q$ of degree at most $k$ and any threshold $c \in \R$, the set of $x$ so that $q(x) \leq c$ is a union of at most $k$ intervals.
So, let $E = \cup_{\ell = 1}^{O(k)} I_\ell$ for some disjoint intervals $I_\ell = [a_\ell, b_\ell]$.
Then, by the integration by parts formula, we have that
\begin{align*}
    \left| \int_E e^{ip(x)}e^{-x^2/2} dx \right| &\leq \sum_{\ell = 1}^{O(k)} \left( \left| \frac{e^{ip(b_\ell)} e^{-b_\ell^2/2}}{p'(b_\ell)} - \frac{e^{ip(a_\ell)} e^{-a_\ell^2/2}}{p'(a_\ell)} \right| +  \int_{a_\ell}^{b_\ell} \left| e^{-x^2/2} \left(\frac{x}{p'(x)} + \frac{p''(x)}{(p'(x))^2} \right) \right| dx  \right)
\end{align*}

By the definition of $E$, we have that for any $y \in E$,
\begin{align*}
    |p'(y)| &\geq \Omega\left( \frac{\log(1/\gamma)}{k} \right)^k |p'|_2 \\
    &\geq \Omega\left( \frac{\log(1/\gamma)}{k} \right)^k \sqrt{\var (p(\rvg))} \\
    &\geq \Omega(\log 1 / \gamma)^k \Var(p(\rvg))^{\Omega (1)} \; ,
\end{align*}
where the last bound follows by our assumption on $\var (p(\rvg))$.
So we have that
\begin{align*}
    \sum_{\ell = 1}^{O(k)} \left| \frac{e^{ip(b_\ell)} e^{-b_\ell^2/2}}{p'(b_\ell)} - \frac{e^{ip(a_\ell)} e^{-a_\ell^2/2}}{p'(a_\ell)} \right| &\leq O(k) \cdot (\log 1 / \gamma)^{-\Omega(k)} \Var(p(\rvg))^{-\Omega (1)} \; ,
\end{align*}
which is significantly stronger than we need, as long as $\gamma < 1/2$.

To bound the remaining term, we observe that for all $y \in E$, and all $\gamma \leq 1$, it holds that
\begin{align*}
    \frac{x}{p'(x)} + \frac{p''(x)}{(p'(x))^2} &= O\left( \frac{(\log 1 / \gamma)^k}{(\gamma / k)^{2k} \var (p)^{1/2}} \right)
\end{align*}
\noindent
So, if we set $\gamma$ to be a small multiple of $\var(p(\rvg))^{-1/(4k)}$, this quantity is at most $\Var(p(\rvg))^{-\Omega (1)}$.
Therefore,
\[
\sum_{\ell = 1}^{O(k)} \int_{a_\ell}^{b_\ell} \left| e^{-x^2/2} \left(\frac{x}{p'(x)} + \frac{p''(x)}{(p'(x))^2} \right) \right| dx \leq \Var(p(\rvg))^{-\Omega (1)} \cdot \int_{-\infty}^\infty e^{-x^2 / 2} dx = \Var(p(\rvg))^{-\Omega (1)} \; .
\]
Combining this with~\eqref{eq:e-ec} and our choice of $\gamma$ yields the desired claim.
\end{proof}

\section{Moments of Subgraph Polynomials}
\label{sec:subgraph-moments}
Let $\sgraph$ be a constant-sized connected graph.
Recall the notation from \Cref{sec:subgraph-stats}, where for an $n\times n$ matrix,
\[
    \chi_{\sgraph}(M) \coloneqq \frac{1}{\sqrt{\abs{\calL_{\sgraph}}}}\sum_{\pi\in\calL_{\sgraph}} \prod_{ab\in\sgraph} M_{\pi(a),\pi(b)}
\]
and $\calL_{\sgraph}$ is the set of all injective maps from $V(\sgraph)$ to $[n]$.

We will prove that the subgraph polynomial of a random matrix has Gaussian-like moments.
\restatelemma{lem:subgraph-moments}
\begin{proof}
    We begin by expanding out the expression we would like to compute.
    We will drop the subscript from $\calL_{\sgraph}$, and write $\calL$. 
    \begin{align*}
        \E_{\rvM \sim \Null} \chi_\vartheta(\rvM)^q = \frac{1}{\abs*{\calL}^{q/2}}\sum_{\pi_1,\dots,\pi_q\in\calL} \E_{\rvM\sim\Null} \prod_{i=1}^{q} \prod_{ab\in E(\sgraph)} {\rvM}_{\pi_i(a),\pi_i(b)}\mper  \numberthis \label{eq:moment}
    \end{align*}
    The summand vanishes unless every edge in the multiset union $H = \pi_1(\sgraph)\sqcup\dots\sqcup\pi_q(\sgraph)$ is an evenly covered edge.

    For a single choice $\pi_1,\dots,\pi_q$, we can associate:
    \begin{itemize}
        \item $r = r(H)$ defined as the number of connected components in $H$,
        \item The connected components $H_1,\dots,H_r$ of $H$ with labels erased,
        \item The number of copies of $\sgraph$ per connected component $\beta_1,\dots,\beta_r$,
        \item Sets $S_1,\dots,S_r\subseteq[q]$ where $|S_i|=\beta_i$ comprises of all $\pi_j$ such that $\pi_j(\sgraph)$ is in $H_i$,
        \item A map $\psi_i$ such that $\psi_i(j)$ for $j\in S_i$ is equal to the subgraph of $H_i$ that arises from $\pi_j(\sgraph)$.
    \end{itemize}
    For fixed choices of $H_1,\dots,H_r$, $S_1,\dots,S_r$ and $\psi_1,\dots,\psi_r$, there are at most $n^{|V(H)|}$ terms in the sum in \Cref{eq:moment} that correspond to these choices.
    Additionally, the value of each term is nonnegative, and depends only on the choices.
    Abbreviating the choices to $\vec{H}, \vec{S}, \vec{\psi}$, denoting the correspond value of each term as $\mathrm{wt}\left(\vec{H}, \vec{S}, \vec{\psi}\right)$, and using $|\calL| = n^{v(\sgraph)}$ we can bound \Cref{eq:moment} as:
    \begin{align*}
        \eqref{eq:moment} &\le \sum_{\vec{H},\vec{S},\vec{\psi}} \mathrm{wt}\parens*{\vec{H}, \vec{S}, \vec{\psi}} \cdot n^{V(H) - qv(\sgraph)/2}\mper
    \end{align*}
    Before we proceed, we introduce some notation.
    We use $\models$ to mean ``is consistent with''.
    \begin{align*}
        \sum_{\vec{H},\vec{S},\vec{\psi}} \mathrm{wt}\parens*{\vec{H}, \vec{S}, \vec{\psi}} \cdot n^{V(H)-qv(\sgraph)/2} &=
        \sum_{r\ge 1} \sum_{\vec{\beta}\models r} \sum_{\vec{S}\models \vec{\beta}} \sum_{(\vec{H},\vec{\psi})\models\vec{S}} \mathrm{wt}\parens*{\vec{H}, \vec{S}, \vec{\psi}} \cdot n^{\sum_i (V(H_i)-\beta_iv(\sgraph)/2)}   \numberthis \label{eq:genau-bundle}
    \end{align*}
    We begin by counting the number of ways in which $(\vec{H},\vec{\psi}) \models \vec{S}$.

    In particular, for each $i\in[r]$, we give an upper bound on the number of ways in which $(H_i,\psi_i)\models S_i$.
    Suppose $\beta_i = 2$, then there is exactly one choice for $(H_i,\psi_i)$ that would lead to nonzero weight.
    Henceforth, we assume that $\beta_i \ge 3$.
    Without loss of generality, say $S_i = \{1,\dots,\beta_i\}$, and that for all $j\in[t]$, the graph $\cup_{s=1}^j\psi_i(s)$ is connected.
    Given $(\vec{H},\vec{\psi}) \models \vec{S}$, we describe an encoding and then give an upper bound on the number of valid encodings.

    \parhead{Encoding scheme.}
    In our encoding scheme, we process $\psi_i(j)$ in order $j = 1,\dots,\beta_i$, and describe the procedure at each step.
    \begin{itemize}
        \item Call a vertex $v$ in $\psi_i(j)$ \texttt{fresh} if it does not appear in $\bigcup_{s=1}^{j-1} \psi_i(s)$.
        \item Call a vertex $v$ in $\psi_i(j)$ \texttt{boundary} if it exists in $\bigcup_{s=1}^{j-1} \psi_i(s)$ but is incident to some singleton edge in $\bigcup_{s=1}^{j} \psi_i(s)$.
        \item Call a vertex $v$ in $\psi_i(j)$ \texttt{stale} if it exists in $\bigcup_{s=1}^{j-1} \psi_i(s)$ and all its incident edges in $\bigcup_{s=1}^{j} \psi_i(s)$ are double-covered.
        \item For each vertex in $\sgraph$, we mark whether it appears in $\psi_i(j)$ as \texttt{fresh}, \texttt{boundary}, or \texttt{stale}, of which there are $3^{v(\sgraph)}$ choices.
        \item For a vertex that is \texttt{boundary} or \texttt{stale}, we write down the name of the vertex in $\bigcup_{s=1}^{j-1} \psi_i(s)$, of which there are at most $\beta_i v(\sgraph)$ many choices.
    \end{itemize}
    Using $f_i, b_i, s_i$ to denote the number of \texttt{fresh}, \texttt{boundary} and \texttt{stale} vertices respectively, we see that the number of encodings is bounded by:
    \begin{align*}
        \prod_{i\in[r]:\beta_i\ge 3} 3^{\beta_i}\cdot (\beta_i v(\sgraph))^{b_i+s_i}\mper
    \end{align*}
    Substituting this bound in gives:
    \begin{align*}
        \eqref{eq:genau-bundle} \le \sum_{r\ge 1} \sum_{\vec{\beta}\models r} \sum_{\vec{S}\models\vec{\beta}} \mathrm{wt}\parens*{\vec{H},\vec{S},\vec{\psi}} \cdot  \prod_{i\in[r]:\beta_i\ge 3} 3^{\beta_i}\cdot (\beta_i v(\sgraph))^{b_i+s_i} \cdot n^{V(H_i)-\beta_iv(\sgraph)/2}\mper
    \end{align*}
    Next, we bound $\mathrm{wt}\parens*{\vec{H},\vec{S},\vec{\psi}}$.
    Within $H_i$, the multiplicity of any edge is at most $\beta_i$, and thus, we have:
    \[
        \mathrm{wt}\parens*{\vec{H},\vec{S},\vec{\psi}} \le \prod_{i\in[r]:\beta_i\ge 3} O\parens*{\sqrt{\beta_i}}^{\beta_i e(\sgraph)} = \prod_{i\in[r]:\beta_i\ge 3} O\parens*{\sqrt{\beta_i}^{e(\sgraph)}}^{\beta_i}
    \]
    Along with the fact that $|V(H_i)| = f_i$ and $b_i+s_i\le\beta_i v(\sgraph)$, this leads us to conclude:
    \begin{align*}
        \eqref{eq:genau-bundle}\le \sum_{r\ge 1}\sum_{\vec{\beta}\models r}\sum_{\vec{S}\models\vec{\beta}} \prod_{i\in[r]:\beta_i\ge 3} O\parens*{\beta_i^{e(\sgraph)/2 + v(\sgraph)} v(\sgraph)^{v(\sgraph)} }^{\beta_i} \cdot n^{f_i - \beta_i v(\sgraph)/2}
    \end{align*}
    Note that $v(\sgraph)^{v(\sgraph)}$ is a constant, and henceforth we absorb it into the $O(\cdot)$-notation.
    Next, we will prove that $f_i - \beta_i v(\sgraph)/2$ is significantly negative.
    First, observe that $f_i + b_i + s_i = \beta_i v(\sgraph)$.
    Since $s_i\ge f_i$, we have $2f_i + b_i \le \beta_i v(\sgraph)$.
    We will now prove a lower bound on $b_i$.
    The first step contributes zero \texttt{boundary} vertices to $b_i$.
    Since $\psi_i(j)$ is connected to $\cup_{s=1}^{j-1}\psi_i(s)$, it must be the case that every subsequent step after the first either contributes at least one \texttt{boundary} vertex or $v(\sgraph)$ \texttt{stale} vertices.
    At most half of the vertices introduced are \texttt{stale}, and hence there can be at most $\beta_i/2$ instances where $v(\sgraph)$ \texttt{stale} vertices are introduced.
    Consequently, the number of \texttt{boundary} vertices $b_i$ is at least $\beta_i/2-1$.
    Thus, we have the inequality
    \begin{align*}
        2f_i &\le \beta_i\parens*{ v(\sgraph) - \frac{1}{2} } + 1\mcom
    \end{align*}
    which implies
    $$
        f_i - \frac{\beta_i v(\sgraph)}{2} \le -\frac{\beta_i}{4} + \frac{1}{2} = \frac{2-\beta_i}{4} \le -\frac{\beta_i}{12}\mcom
    $$
    where the final inequality used $\beta_i\ge 3$.

    As an upshot, we have:
    \begin{align*}
        \eqref{eq:genau-bundle} \le \sum_{r\ge 1} \sum_{\vec{\beta}\models r} \sum_{\vec{S}\models\vec{\beta}} \prod_{i\in[r]:\beta_i\ge 3} O\parens*{ \frac{\beta_i^{e(\sgraph) + v(\sgraph)} }{n^{1/12}} }^{\beta_i}
    \end{align*}
    For a fixed choice of $\vec{\beta}$, it is easy to see that the number of choices for $\vec{S}\models\vec{\beta}$ is at most:
    \[
        q^{q/2} \cdot \prod_{i\in[r]:\beta_i\ge 3} q^{\beta_i}\mper
    \]
    Consequently,
    \begin{align*}
        \eqref{eq:genau-bundle} \le q^{q/2} \sum_{r\ge 1} \sum_{\vec{\beta}\models r} \prod_{i\in[r]:\beta_i \ge 3} O\parens*{ \frac{q\beta_i^{e(\sgraph)+v(\sgraph)}}{n^{1/12}} }^{\beta_i}
    \end{align*}
    We can choose $\alpha(\sgraph)$ small enough so that even after we use $q,\beta\le n^{\alpha(\sgraph)}$, we have:
    \begin{align*}
        \eqref{eq:genau-bundle} &\le q^{q/2} \sum_{r\ge 1} \sum_{\beta\models r} \prod_{i\in[r]:\beta_i\ge 3} O\parens*{n^{-1/24}}^{\beta_i} \\
        &\le q^{q/2} \sum_{\beta_1,\dots,\beta_q\ge 0} \prod_{i=1}^q O\parens*{n^{-1/24}}^{\beta_i} \\
        &\le q^{q/2} \parens*{\sum_{\beta \ge 0} O\parens*{n^{-1/24}}^{\beta}}^q \\
        &\le q^{q/2} \parens*{1 + O\parens*{n^{-1/24}} }^q \\
        &\le q^{q/2}\cdot\parens*{1 + \frac{1}{\poly(n)}}\mper  \qedhere
    \end{align*}
\end{proof}

\end{document}